\documentclass[11pt]{article}
\usepackage{enumerate}
\usepackage{amsfonts}
\usepackage{amssymb}
\usepackage{amsmath}
\usepackage{xspace}
\usepackage{algpseudocode}

\usepackage{graphicx}
\usepackage{multicol}
\usepackage[margin=1in]{geometry}
\usepackage{hyperref}
\usepackage{amsthm}

\usepackage{times}

\newtheorem{lemma}{Lemma}
\newtheorem{theorem}[lemma]{Theorem}
\newtheorem{corollary}[lemma]{Corollary}

\newtheorem{proposition}[lemma]{Proposition}

\newtheorem{xmpl}[lemma]{Example}
\newenvironment{example}{\begin{xmpl}\rm}{\end{xmpl}}

\newtheorem{rmark}[lemma]{Remark}
\newenvironment{remark}{\begin{rmark}\rm}{\end{rmark}}

\theoremstyle{definition}
\newtheorem{definition}[lemma]{Definition}
\usepackage[T1]{fontenc}

\usepackage{amssymb,amsfonts,amsthm}
\usepackage{changepage}
\usepackage{chngcntr}
\usepackage[inline,shortlabels]{enumitem}
\usepackage{mathtools}
\usepackage{scrextend}
\usepackage{stmaryrd }
\usepackage{thmtools}
\usepackage{thm-restate}
\usepackage{tikz}
\usetikzlibrary{chains,matrix,scopes,decorations.shapes,arrows,shapes}
\usepackage{wrapfig}
\usepackage{xspace}

\newcommand{\norm}[1]{\tilde{#1}}

\numberwithin{equation}{section}

\makeatletter
\let\epsilon\varepsilon
\let\phi\varphi
\let\emptyset\varnothing
\makeatother

\newcommand{\diff}{H}
\newcommand{\init}[1]{#1_{0}}

\renewcommand{\Pr}{\mathbb{P}}

\newcommand{\MP}{\textsf{MP}}
\newcommand{\sMP}{\textsf{sMP}}

\newcommand{\asMP}{\textsf{asMP}}
\newcommand{\RT}{\textsf{RT}}

\newcommand{\FP}{\textsf{FP}}

\newcommand{\APRich}{\text{AP-Rich}\xspace}
\newcommand{\FPPoor}{\text{FP-poor}\xspace}

\newcommand{\energy}{\textsf{energy}}
\newcommand{\payoff}{\textsf{payoff}}

\newcommand{\play}{\textsf{play}}

\newcommand{\lolli}{\G_{\bowtie}}
\newcommand{\St}{\mbox{St}}
\newcommand{\Pot}{\mbox{Pot}}

\newcommand{\minim}{Min}

\newcommand{\gainmax}{G^+}

\newcommand{\investmax}{I^+}

\newcommand{\aux}{\lambda}

\newcommand{\Unif}{\mathit{Unif}}

\newcommand{\G}{{\cal G}}

\renewcommand{\P}{{\cal P}}
\newcommand{\PO}{Player~$1$\xspace}
\newcommand{\PT}{Player~$2$\xspace}
\newcommand{\PLi}{Player~$i$\xspace}

\newcommand{\thresh}{\texttt{Th}\xspace}
\newcommand{\Max}{\text{Max}\xspace}
\newcommand{\Min}{\text{Min}\xspace}

\newcommand{\Real}{\mathbb{R}}
\newcommand{\Nat}{\mathbb{N}}
\newcommand{\Rat}{\mathbb{Q}}

\newcommand{\zug}[1]{\langle #1  \rangle}
\newcommand{\set}[1]{\{ #1 \}}
\newcommand{\stam}[1]{}

\title{\Large Infinite-Duration All-Pay Bidding Games\thanks{This research was supported in part by the Austrian Science Fund (FWF) under grant Z211-N23 (Wittgenstein Award), ERC CoG 863818 (FoRM-SMArt), and  by the European Union's Horizon 2020 research and innovation programme under the Marie Sk\l odowska-Curie Grant Agreement No.~665385.}}
\author{Guy Avni\thanks{University of Haifa}
\and Ism\"ael Jecker\thanks{IST Austria} \and {\DJ}or{\dj}e \v{Z}ikeli\'c\thanks{IST Austria}}

\date{}
\begin{document}
\maketitle

\begin{abstract}
In a two-player zero-sum graph game the players move a token throughout a graph to produce an infinite path, which determines the winner or payoff of the game. Traditionally, the players alternate turns in moving the token. In {\em bidding games}, however, the players have budgets, and in each turn, we hold an ``auction'' (bidding) to determine which player moves the token: both players simultaneously submit bids and the higher bidder moves the token. The bidding mechanisms differ in their payment schemes. Bidding games were largely studied with variants of {\em first-price} bidding in which only the higher bidder pays his bid. We focus on {\em all-pay} bidding, where both players pay their bids. Finite-duration all-pay bidding games were studied and shown to be technically more challenging than their first-price counterparts. We study for the first time, infinite-duration all-pay bidding games. Our most interesting results are for {\em mean-payoff} objectives: we portray a complete picture for games played on strongly-connected graphs. We study both pure (deterministic) and mixed (probabilistic) strategies and completely characterize the optimal sure and almost-sure (with probability $1$) payoffs that the players can respectively guarantee. We show that mean-payoff games under all-pay bidding exhibit the intriguing mathematical properties of their first-price counterparts; namely, an equivalence with {\em random-turn games} in which in each turn, the player who moves is selected according to a (biased) coin toss. The equivalences for all-pay bidding are more intricate and unexpected than for first-price bidding. 
\end{abstract}
\clearpage
\setcounter{page}{1}
\newpage

\section{Background, Definitions, and Summary of Results}
Graph games are two-player zero-sum games with deep connections to foundations of logic \cite{Rab69}. They have numerous practical applications, e.g., verification \cite{EJS93}, reactive synthesis \cite{PR89}, and reasoning about multi-agent systems \cite{AHK02}. There are interesting theoretical problems on graph games: e.g., solving {\em parity games} is a rare problem in NP and coNP for which no polynomial-time algorithm is known and only recently a quasi-polynomial time algorithm was found~\cite{CJ+17}.

A graph game is played on a finite directed graph. The game proceeds by placing a token on one of the vertices and allowing the players to move it throughout the graph to produce an infinite path, which determines the winner or payoff of the game. Traditionally, the players alternate turns when moving the token. We study {\em bidding games} \cite{LLPU96,LLPSU99} in which the players have budgets, and in each turn, we hold an ``auction'' (bidding) to determine which player moves the token. 

\subsection{Bidding mechanisms, budget ratios, and strategies}
In all the mechanisms we consider, in each turn, both players simultaneously submit a bid that does not exceed their available budget, and the higher bidder moves the token. The mechanisms differ in their payment schemes. We classify the payment schemes according to two orthogonal properties: {\em who pays} and {\em who is the recipient}. For the first, we consider {\em first-price} bidding, in which only the higher bidder pays, and {\em all-pay} bidding in which both players pay their bids. For the latter, two mechanisms were defined in \cite{LLPSU99}: in {\em Richman bidding} (named after David Richman), payments are made to the other player, and in {\em poorman bidding} the payments are made to the ``bank'' thus the money is lost. A third payment scheme called {\em taxman} spans the spectrum between Richman and poorman, as we elaborate in Sec.~\ref{sec:tax}. 

Previously, bidding games were largely studied in combination with first-price bidding. In this work we study, for the first time, infinite-duration bidding games under all-pay bidding and portray a complete picture for both all-pay Richman and poorman bidding. 

%In a nutshell, when comparing Richman and poorman bidding, Richman bidding tend to be technically cleaner and easier to work with while poorman bidding is more suited for practical applications. We hope to convince the reader that the results on poorman bidding are more surprising.

We make the payment schemes precise below. For $i \in \set{1,2}$, suppose \PLi's budget is $B_i$ prior to a bidding and his bid is $b_i \in [0,B_i]$, and assume for convenience that \PO wins the bidding, thus $b_1 > b_2$. The budgets are updated as follows:
\begin{itemize}[noitemsep,topsep=0pt]
\item {\bf First-price:} Only the higher bidder pays.
\begin{itemize}[noitemsep,topsep=0pt]
\item {\bf Richman:} $B'_1 = B_1 - b_1$ and $B'_2 = B_2 + b_1$.
\item {\bf Poorman:} $B'_1 = B_1 - b_1$ and $B'_2 = B_2$.
%\item {\bf Taxman:} For a fixed $\tau \in [0,1]$, we have $B'_1 = B_1 - b_1$ and $B'_2 = B_2+(1-\tau) \cdot b_1$.
\end{itemize}
\item {\bf All-pay:} Both players pay their bids.
\begin{itemize}[noitemsep,topsep=0pt]
\item {\bf Richman:} $B'_1 = B_1 - b_1 + b_2$ and $B'_2 = B_2 + b_1 - b_2$. Thus, \PO pays \PT the difference between the two bids.
\item {\bf Poorman:} $B'_1 = B_1 - b_1$ and $B'_2 = B_2 - b_2$.
%\item {\bf Taxman.} For a fixed $\tau \in [0,1]$, we have $B'_1 = B_1 - b_1 + \tau \cdot b_2$ and $B'_2 = B_2- b_2+\tau \cdot b_1$.
\end{itemize}
\end{itemize}

For convenience, we assume ties are broken in favor of \PT, and our results are independent of the tie-breaking mechanism that is used.
A central quantity in bidding games is the following:

\begin{definition}
{\bf (Budget ratio).}
Suppose \PLi's budget is $B_i$, for $i \in \set{1,2}$, then \PLi's {\em ratio} is $\frac{B_i}{B_1 + B_2}$
\end{definition}

A {\em strategy} in a bidding game is a function that, in full generality, takes a finite {\em history} of the game, which includes the visited vertices, the bids made by the players, their outcomes, etc. It prescribes a probability distribution over legal bids, i.e., bids that do not exceed the available budget, and a neighboring vertex to move the token to upon winning the bidding. We say that a strategy is {\em pure} when it prescribes one bid with probability $1$ and otherwise we call the strategy {\em mixed}. One of our contributions is that we construct {\em budget based} strategies in which the bid depends only on the current vertex and the current budget, and the choice of move does not depend on the budget.

\paragraph*{Applications.}
All-pay bidding is often better suited than first-price bidding for modelling practical settings. Applications arise from viewing the players' budgets as resources with little or no inherent value, e.g., time or strength, and a strategy as a recipe to invest resources with the goal of maximizing the expected utility. In many settings, invested resources are lost, thus all-pay bidding is more appropriate than first-price bidding. All-pay poorman games can be seen as a dynamic variant of {\em Colonel Blotto} games, which date back to \cite{Bor21} and have been extensively studied since. Applications of Colonel Blotto games, which carry over to all-pay bidding games, include political lobbying and campaigning, rent seeking \cite{Tul80}, and modelling biological processes \cite{CRN12}. In fact, due to their dynamic nature, bidding games are a better model for these applications. 

Another application of bidding games is reasoning about systems in which the scheduler accepts payment in exchange for priority. {\em Blockchain technology} is one such example. Simplifying the technology, a blockchain is a log of transactions issued by clients and maintained by {\em miners}, who accept transaction fees from clients in exchange for writing transactions to the blockchain. In {\em Etherium}, the blockchain consists of snippets of code (called {\em smart contracts}). Verification of Etherium programs is both challenging and important since bugs can cause loss of money (e.g., \cite{CGV18}). Bidding games, and specifically all-pay poorman games, can model Etherium programs: we associate players with clients and, as is standard in model checking, we associate the states of the program with the vertices of the graph. All-pay poorman bidding is the most appropriate bidding mechanism since in Etherium, the transaction fees are always paid to the miners.

\subsection{Reachability bidding games}
\begin{definition}
{\bf (Reachability games).}
A reachability game has two target vertices $t_1$ and $t_2$. The game ends once a target $t_i$ is visited, for $i \in \set{1,2}$. Then, \PLi is the winner. 
\end{definition}
\begin{figure}[ht]
\begin{minipage}[b]{0.45\linewidth}
\centering
\includegraphics[height=2cm]{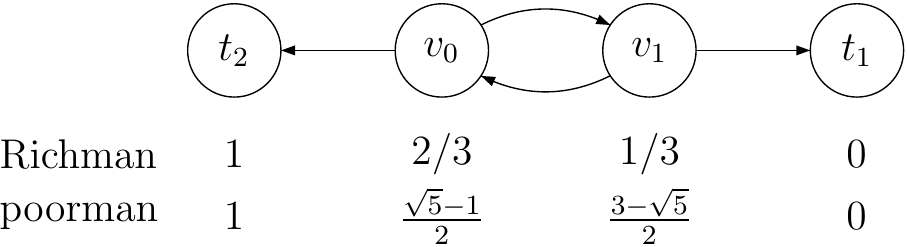}
\caption{A reachability bidding game with the threshold ratios under first-price Richman and poorman bidding.}
\label{fig:reach}
\end{minipage}
%\quad
\hspace{0.1\linewidth}
\begin{minipage}[b]{0.45\linewidth}
\centering
\includegraphics[height=2cm]{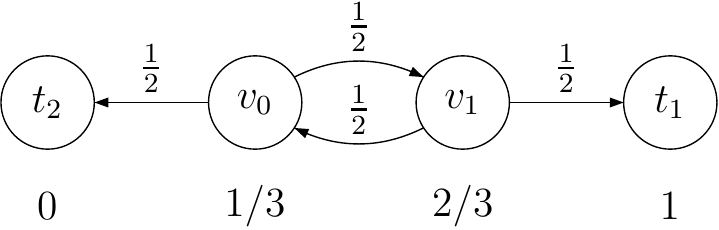}
\caption{The (simplified) random-turn game that corresponds the game in Fig.~\ref{fig:reach} with the probabilities to reach $t_1$ from each vertex.}
\label{fig:RTreach}
\end{minipage}
\end{figure}

In \cite{LLPU96,LLPSU99}, only first-price bidding mechanisms were considered and only in combination with reachability objectives. The main question studied concerned a necessary and sufficient initial budget ratio for winning the game called the {\em threshold ratio}, and denoted $\thresh(v)$, for a vertex $v$. Formally, for a vertex $v$, if \PO's ratio is greater than $\thresh(v)$, he deterministically wins the game from $v$, and if \PT's ratio is greater than $1-\thresh(v)$, she wins the game. Threshold ratios were shown to exist for reachability first-price bidding games. See for example Fig.~\ref{fig:reach}. Moreover, threshold ratios in first-price Richman bidding are particularly favorable: the threshold ratio in a vertex is the average of two of its neighbors. This implies an intriguing equivalence with a class of games called {\em random-turn games} (see Def.~\ref{def:RT}), which is well-studied in its own account since the seminar paper \cite{PSSW09}. We illustrate a simplified version of the equivalence on games with out-degree $2$ and the general statement can be found in \cite{LLPSU99,AHC19}.
\begin{example}
Consider the game depicted in Fig.~\ref{fig:reach}. Construct a Markov chain by labeling each edge with probability $0.5$ (see Fig.~\ref{fig:RTreach}). Note that for each vertex $u$, we have $\thresh(u) = 1-\Pr[\text{reach}(u, t_1)]$.\hfill$\triangleleft$
\end{example}

\begin{remark}
\label{rem:reach-poor}
For reachability objectives, apart from Richman bidding, no equivalence is known between bidding games and random-turn games, for any other bidding mechanism. Moreover, such an equivalence is unlikely to exist since values in stochastic games are rational numbers (they constitute a solution to a linear program) and threshold ratios under  first-price poorman bidding are irrational already in the game depicted in Fig.~\ref{fig:reach}.\hfill$\triangleleft$
\end{remark}

\noindent{\bf Reachability all-pay poorman bidding games}
were only recently studied \cite{AIT20}. Technically, these games are significantly harder than first-price bidding and there are large gaps in our understanding of this model. To illustrate, contrary to reachability first-price bidding, mixed strategies are required already in the simplest interesting bidding game ``\PO needs to win two biddings in a row'', whose solution was left as an open question in \cite{LLPSU99}. It was shown in \cite{AIT20} that for $n > 1$, when \PT's budget is $1$ and \PO's budget is in $(1+\frac{1}{n}, 1+\frac{1}{n-1}]$, in the first bidding, a \PO optimal strategy bids uniformly at random from $\set{\frac{i}{n}: 1 \leq i \leq n}$ and guarantees winning with probability $\frac{1}{n}$. Reachability all-pay bidding games become complicated very fast; e.g., strategies in ``\PO wins three times in a row'' require infinite support, and experiments hint that unlike ``win twice in a row'' the optimal winning probability as a function of the initial budget ratio is a continuous function. Moreover, the basic question of the existence of a value of the game, remains open.

  \begin{figure}
    \begin{minipage}[b]{0.45\linewidth}
    \centering
\includegraphics[height=1.5cm]{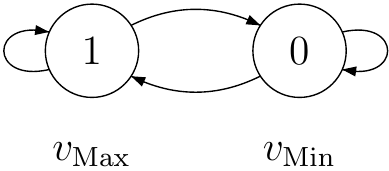}
\caption{\small The mean-payoff game $\lolli$ with the weights in the vertices.}
\label{fig:bowtie}
\end{minipage}
\hspace{0.1\linewidth}
  \begin{minipage}[b]{0.45\linewidth}
    \centering
\includegraphics[height=1.5cm]{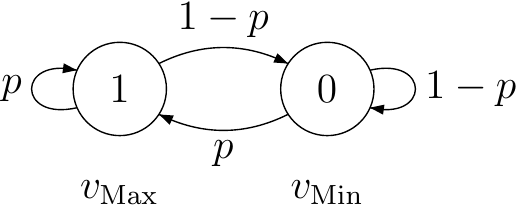}
\caption{The simplified random-turn game that corresponds to $\lolli$ w.r.t. $p \in [0,1]$.}
\label{fig:RTBowtie}
\end{minipage}
\end{figure}

 \stam{
\includegraphics[height=2cm]{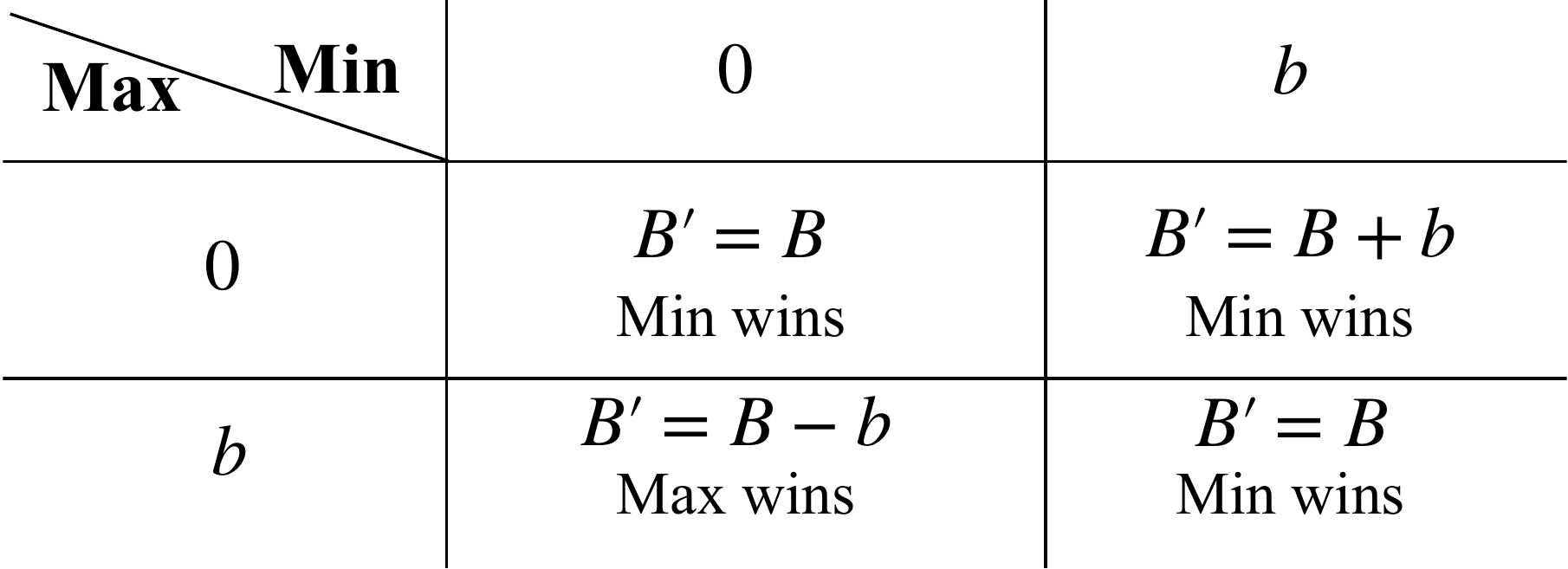}
\caption{\small \Max's budget updates in four bidding outcomes under \APRich.}% assuming \Max bids according to rows.}
\label{tab:Richman}
\end{minipage}
}

\subsection{Mean-payoff bidding games}
{\em Mean payoff} games are quantitative games. Each infinite play has a {\em payoff}, which is \PO's reward and \PT's cost, thus we call the players in a mean-payoff game \Max and \Min, respectively.  We illustrate the mean-payoff objective in the following example.

\begin{example}
Suppose that two advertisers repeatedly (e.g., daily) compete to publish their ad on a content-provider's website (e.g., New York Times). We associate the advertisers with two players in a mean-payoff bidding game. The payoff is the long-run average time that \Max's ad shows (e.g., the number of days his ad appears in a year). Since bids are paid to the content provider, poorman bidding is the appropriate bidding mechanism. When the site has only one ad slot, only the higher bidder's ad shows and only he pays his bid, thus we use first-price poorman bidding. Alternatively, when there are two ad slots (e.g., at the top and bottom of the page), the players compete on who gets the better position, and both pay their bids, thus we use all-pay poorman bidding. Our goal is to find an optimal bidding strategy for an advertiser that, given his budget constraints, maximizes the long-run ratio of the time that his ad shows. To find such a strategy, we reason about $\lolli$ in Fig.~\ref{fig:bowtie}: \Max moves to the vertex $v_\Max$ upon winning a bidding, which represents his ad showing or showing in the favorable position.\hfill$\triangleleft$
\end{example}

Formally, a {\em mean-payoff game} is played on a weighted graph $\zug{V, E, w}$, where $w: V \rightarrow \Rat$. The payoff is defined as follows.

\begin{definition}
\label{def:MP}
{\bf (Payoff and energy).} Consider an infinite path $\eta = \eta_0, \eta_1, \ldots$. For $n >1$, let $\eta^n = \eta_0, \ldots, \eta_n$ be a prefix of $\eta$. The {\em energy} of $\eta^n$, denoted $\energy(\eta^n)$, is the sum of weights it traverses, thus $\energy(\eta^n) = \sum_{0 \leq i < n} w(\eta_i)$. The {\em payoff} of $\eta$, denoted $\payoff(\eta)$, is $\payoff(\eta) = \lim \inf_{n \to \infty} \frac{1}{n} \cdot \energy(\eta^n)$. Note that the use of $\lim\inf$ gives \Min an advantage.
\end{definition}

\begin{remark}
Unless stated otherwise, we consider games played on strongly-connected graphs. Under first-price bidding, this implies a solution to general games since we first solve the bottom-strongly connected components (BSSCs) and then construct a reachability game in which a player's goal is to force the game to a BSCC that is ``good'' for him. A similar solution would apply also under all-pay bidding, but reachability games are not yet solved for these bidding mechanisms.
\end{remark}

The central question in mean-payoff bidding games is (see Def.~\ref{def:MP-val} for a formal definition): 
\begin{center}{\it What is the optimal payoff a player can guarantee given an initial budget ratio?}\end{center}
For example, suppose \Max's ratio is $2/3$. What is the optimal payoff he can guarantee in $\lolli$ (Fig.~\ref{fig:bowtie}) under first-price Richman bidding? Would \Max prefer first-price Richman or poorman bidding? Does the answer change when the ratio is $1/3$? 

We answer these questions by showing an equivalence between bidding games and random-turn games, which are defined as follows.

\begin{definition}
\label{def:RT}
{\bf (Random-turn games).} Consider a bidding game $\G$ that is played on a graph over a set of vertices $V$. For $p \in [0,1]$, the random-turn game that corresponds to $\G$ w.r.t. $p$, denoted $\RT(\G, p)$, is a game in which instead of bidding, in each turn we toss a (biased) coin to determine which player gets to move the token: \PO is chosen with probability $p$ and \PT with probability $1-p$. Formally, $\RT(\G, p)$ is a stochastic game \cite{Con92}. For each vertex $v\in V$, we add two vertices $v_1$ and $v_2$. The vertex $v$ is a ``Nature'' vertex and simulates the coin toss, thus it has two outgoing edges: one with probability $p$ to $v_1$ and a second with probability $1-p$ to $v_2$. For $i \in \set{1,2}$, the vertex $v_i$ is controlled by \PLi and there are deterministic edges from $v_i$ to $u$, for every neighbor $u$ of $v$. The  objective in $\RT(\G,p)$ matches that of $\G$, thus when $\G$ is a mean-payoff game, so is $\RT(\G, p)$. Its {\em mean-payoff value}, denoted $\MP\big(\RT(\G, p)\big)$, is a well-known concept and is defined as the expected payoff under optimal play of the two players. It is known that the optimal value exists and that it can be achieved using optimal pure {\em positional} strategies; namely, a strategy that prescribes, at each vertex, a successor that does not depend on the history of the game. Since $\G$ is strongly-connected, the value does not depend on the initial vertex. 
\end{definition}

\begin{example}
In $\lolli$, since when \Max and \Min win a bidding they respectively move to $v_\Max$ and $v_\Min$, we can simplify $\RT(\lolli, p)$ to a weighted Markov chain (see Fig.~\ref{fig:RTBowtie}). Informally, we expect that a random walk ``stays'' in $v_\Max$ portion $p$ of the time and since the weights are simple, we have $\MP\big(\RT(\lolli, p)\big) = p$. \hfill$\triangleleft$
\end{example}

\subsubsection{Mean-payoff First-price bidding games}
\label{sec:intro-MP-FP}
%Mean-payoff first-price bidding games were studied under Richman \cite{AHC19}, poorman \cite{AHI18}, and taxman \cite{AHZ19} bidding. 
We survey the results on first-price bidding games obtained in \cite{AHC19,AHI18}.

\medskip
\noindent{\bf First-price Richman bidding:} The initial budgets do not matter in mean-payoff first-price Richman bidding games. Moreover, these games are equivalent to fair random-turn games by associating optimal payoff in the bidding game with expected payoff in the random-turn game. Formally, consider a strongly-connected mean-payoff game $\G$ and suppose both players have positive initial ratios. Then, for every $\epsilon >0$, \Max has a pure strategy that guarantees a payoff of at least $\MP\big(\RT(\G, 0.5)\big) - \epsilon$. Since Def.~\ref{def:MP} favors \Min, this implies that \Min can guarantee a payoff of at most $\MP\big(\RT(\G, 0.5)\big) + \epsilon$.
For example, in $\lolli$, both players can (roughly) guarantee a payoff of $0.5$, no matter the initial ratios. 

\medskip
\noindent{\bf First-price poorman bidding:} While the equivalence for Richman bidding can be seen as a generalization of the equivalence for reachability objectives, recall that no such equivalence is known for first-price poorman bidding. We thus find it surprising that mean-payoff first-price poorman bidding {\em are} equivalent to random-turn games. In fact, the equivalence is richer than under Richman bidding. In a mean-payoff game $\G$ with a ratio that exceeds $r \in (0,1)$ and for every $\epsilon > 0$, under first-price poorman bidding, \Max can deterministically guarantee a payoff of $\MP\big(\RT(\G, r)\big)-\epsilon$. Again, a dual result holds for \Min. For example, in $\lolli$, with a ratio of $2/3$, the optimal payoff \Max can guarantee is $2/3$. Thus, when \Max's initial ratio is greater than \Min's ratio, he prefers playing with poorman bidding, when his ratio is less than \Min, he prefers Richman, and interestingly, when the ratios are the same, the payoffs under both bidding rules coincide. 

\medskip
A secondary contribution of this work is a new and significantly simpler construction of optimal budget-based strategies under first-price Richman and poorman bidding. 

\begin{remark}
\label{rem:strategies}
{\bf (Strategies in bidding games vs. stochastic games).}
We point out that strategies in bidding games are much more complicated than in stochastic games. At a vertex $v$ in a stochastic game, a strategy only needs to select a vertex $u$ to move the token to from $v$. In a bidding game, in addition to the choice of $u$, a strategy prescribes a bid. While in reachability games, a bidding strategy can easily be extracted from the solution of the random-turn game, in mean-payoff games, this is no longer the case: knowing the optimal payoff a player can achieve in a game does not give any hint on the optimal bidding strategy and finding the right bids is indeed a challenging task.
\end{remark}

\subsection{Mean-payoff all-pay bidding games}
The starting point of this research is inspired by the results for first-price poorman bidding: the moral of those results is that as we ``go to the infinity'', bidding games become cleaner and exhibit a more elegant mathematical structure. We ask: Does this phenomenon also hold for all-pay bidding, where reachability games are highly complex? Would infinite-duration all-pay bidding games reveal a clean mathematical structure like their first-price counterparts? We answer both of these questions positively. 

In this section, we survey our most technically-challenging contribution in which we portray a complete picture for mean-payoff all-pay Richman and poorman bidding played on strongly-connected graphs: we study both pure and mixed strategies and completely characterize the optimal and almost-sure (with probability $1$) payoffs the players can respectively guarantee. We draw corollaries of these results on qualitative objectives (Sec.~\ref{sec:qual}) and on computational complexity (Sec.~\ref{sec:compl}). In Sections~\ref{sec:AP-Rich} and~\ref{sec:AP-poor}, we respectively prove the results for mean-payoff all-pay Richman and poorman bidding.

Before we state our results, we need several definitions. Let $f$ and $g$ be two strategies for \Max and \Min, respectively. When both strategies are deterministic, together with an initial vertex, they give rise to a unique {\em play}, which we denote $\play(f,g)$, where for ease of notation we omit the initial vertex since it usually does not play a role in our results. Roughly, we obtain $\play(f,g)$ inductively. Suppose a finite play $\pi$ that ends in $v$ is defined. Then, we feed $\pi$ into $f$ and $g$, to obtain actions $\zug{b_i, u_i}$, for each $i \in \set{1,2}$, where $b_i$ is a legal bid and $u_i$ is a neighbor of $v$. Then, if $b_1 > b_2$, the token moves to $u_1$ and otherwise it moves to $u_2$. When $f$ and $g$ are mixed, they give rise to a distribution over infinite plays, denoted $dist(f,g)$. Since we consider mixed strategies with continuous support, the definition requires us to define a probability space using a {\em cylinder construction}~\cite[Theorem 2.7.2]{AD00}, which is technical but standard and we do not present it here (see more details in Sec.~\ref{sec:martingale}).
% When $f$ and $g$ are clear from the context, we omit them and simply write $\mathbb{P}$ and $\mathbb{E}$ instead of $\mathbb{P}_{dist(f,g)}$ and expectation $\mathbb{E}_{dist(f,g)}$.

\begin{definition}
\label{def:MP-val}
{\bf (Mean-payoff value).} Consider a mean-payoff game $\G$ and a ratio $r$. 
\begin{itemize}[noitemsep,topsep=0pt]
\item The {\em sure-value} of $\G$ w.r.t. $r$, denoted $\sMP(\G,r)$, is $c \in \Real$ if with a ratio that exceeds $r$, \Max can deterministically guarantee a payoff of $c$: for every $\epsilon >0$ and no matter where the game starts, there is a deterministic \Max strategy $f$ s.t. for every deterministic \Min strategy $g$, we have $\payoff(\play(f,g)) > c-\epsilon$. And, \Max cannot do better: for every deterministic \Max strategy $f$, with a ratio that exceeds $1-r$, there is a deterministic \Min strategy $g$ that guarantees $\payoff(\play(f,g)) < c+\epsilon$.

\item The {\em almost-sure value} of $\G$ w.r.t. $r$, denoted $\asMP(\G, r)$, is $c \in \Real$ if for every $\epsilon >0$ and no matter where the game starts, when \Max's ratio exceeds $r$, he has a mixed strategy $f$ s.t. for every deterministic \Min strategy $g$, we have $\Pr_{\pi \sim dist(f,g)}[\payoff(\pi) > c-\epsilon]=1$, and dually, when \Min's ratio exceeds $1-r$, she has a mixed strategy $g$ s.t. for every deterministic \Max strategy $f$, we have $\Pr_{\pi \sim dist(f,g)}[\payoff(\pi) < c+\epsilon]=1$.
\end{itemize}
\end{definition}

%For $\gamma \in \set{\FP, \AP} \times \set{\text{Rich}, \text{Poor}, \text{Tax}}$, we sometimes use $\MP_\gamma(\G, r)$ to highlight to which bidding mechanism the value relates. 

\paragraph*{All-pay Richman bidding:}
A simple argument shows that deterministic strategies are ``useless'': for every \Max strategy, \Min has a strategy that wins all but a constant number of biddings. For example, in $\lolli$, no matter what the initial ratio is, \Max cannot deterministically guarantee any positive payoff. On the positive side, we show that with mixed strategies first-price and all-pay Richman bidding coincide. For example, in $\lolli$, with any positive initial ratio, \Max can guarantee an almost-sure payoff of $0.5$. We prove the following result on mean-payoff all-pay Richman games.

%To formally define ``useless'', we need to refer to the minimal payoff that can be attained in a mean-payoff game $\G$, denoted $\minMP(\G)$. Let $\cycles(\G)$ denote the set of simple cycles in $\G$. Then, we define $\minMP(\G) = \min_{\eta \in \cycles(\G)} \energy(\eta)/|\eta|$. Intuitively, if \Min was to win all biddings, she would draw the game to a cycle obtaining the minimum in the expression and repeat the cycle indefinitely. The payoff would then be $\minMP(\G)$. 

\begin{theorem}
\label{thm:AP-Rich}
Consider a strongly-connected mean-payoff all-pay Richman bidding game $\G$. For every ratio $r \in (0,1)$, we have:
\begin{itemize}[noitemsep,topsep=0pt]
\item Deterministic strategies: $\sMP(\G, r) = \MP\big(\RT(\G, 0)\big)$.
\item Mixed strategies: $\asMP(\G, r) = \MP\big(\RT(\G, 0.5)\big)$.
\end{itemize}
\end{theorem}

\paragraph*{All-pay poorman bidding:}
Given the results on all-pay Richman, it seems safe to guess that under all-pay poorman bidding, deterministic strategies are useless and that first-price and all-pay poorman coincide. Both guesses, however, turn out to be incorrect. Consider again the game $\lolli$ and suppose \Max's budget is $B = 0.75$ and \Min's budget is $C = 0.25$, thus the initial ratio is $0.75$. As a baseline, recall that under first-price poorman, the optimal payoff \Max can guarantee is $0.75$. 

First, deterministic strategies are {\em useful} in all-pay poorman bidding for the player who has the higher ratio. For example, in $\lolli$, with a ratio that exceeds $0.75$, \Max can deterministically guarantee a payoff of $2/3$. On the other hand, when $B \leq C$, a simple argument shows that deterministic strategies are useless.

The real surprise is with mixed strategies. Given a choice between all-pay and first-price poorman bidding, with a ratio of $0.75$, \Max strictly prefers all-pay bidding! In $\lolli$, he can guarantee an almost-sure payoff of $5/6$. This is tight; namely, with a ratio that exceeds $0.25$, \Min can guarantee an almost-sure payoff of $1/6$. Thus, when \Max's ratio is at most $0.5$, he would prefer first-price over all-pay poorman bidding. 

We prove the following result on mean-payoff all-pay poorman games. 

\begin{theorem}\label{thm:AP-poor}
Consider a strongly-connected mean-payoff all-pay poorman bidding game $\G$ and initial budgets $\init{B}$ for \Max and $\init{C}$ for \Min, thus the ratio is $r = \frac{\init{B}}{\init{B}+\init{C}}$.
\begin{itemize}[noitemsep,topsep=0pt]
\item Deterministic strategies: If $\init{B}>\init{C}$ then $\sMP(\G, r) \geq \MP\big(\RT(\G, 1-\frac{\init{C}}{\init{B}})\big)$,
and if $\init{B} \leq \init{C}$, then $\sMP(\G, r) = \MP\big(\RT(\G, 0)\big)$. 
\item Mixed strategies: If $\init{B}>\init{C}$ then $\asMP(\G, r) = \MP\big(\RT(\G, 1- \frac{\init{C}}{2\init{B}})\big)$,
and if $\init{B} \leq \init{C}$, then $\asMP(\G, r) =  \MP\big(\RT(\G, \frac{\init{B}}{2\init{C}})\big)$. 
\end{itemize}
\end{theorem}

\begin{figure}
\begin{center}
\begin{tabular}{|c|c|c|c|c|}
\hline
&\multicolumn{2}{c}{Richman}\vline&\multicolumn{2}{c}{poorman}\vline\\
\hline
First-price&\multicolumn{2}{c}{$\RT(\G, \frac{1}{2})$ \cite{AHC19}}\vline&\multicolumn{2}{c}{$\RT(\G, r)$ \cite{AHI18}}\vline\\
%&\multicolumn{2}{c}{$(\frac{1}{2})$ \cite{AHC19} }\vline &\multicolumn{2}{c}{$(\frac{3}{4})$ \cite{AHI18}} \vline \\
\hline
All-pay&Pure&Mixed&Pure&Mixed\\
 & $\RT(\G, 0)$ & $\RT(\G, \frac{1}{2})$ & $\RT(\G, \frac{2r-1}{r})$ & $\RT(\G, \frac{3r-1}{r})$ \\
%& $(0)$ & $(\frac{1}{2})$ & $(\frac{2}{3})$ &  $(\frac{5}{6})$ \\
\hline
 \end{tabular}
 \caption{For a strongly-connected mean-payoff game $\G$ and a ratio $r \in (0,1)$, the table summarizes the equivalences with random-turn games for the various bidding mechanisms and allowed strategies.}
\end{center}
\end{figure}

\subsection{Taxman bidding}
\label{sec:tax}
Taxman bidding span the spectrum between Richman and poorman bidding. It is parameterized by a constant $\tau \in [0,1]$, and when \PLi, for $i \in \set{1,2}$, wins a bidding with a bid $b$, he pays $\tau\cdot b$ to the bank and $(1-\tau)\cdot b$ to the other player. Thus, poorman bidding is $\tau=1$ and Richman bidding is $\tau=0$. Threshold ratios were shown to exist in reachability first-price taxman bidding games \cite{LLPSU99}. Mean-payoff first-price taxman bidding games were studied in \cite{AHZ19}, where the equivalence for first-price Richman and poorman bidding was unified: a mean-payoff game $\G$ with taxman parameter $\tau \in [0,1]$ and ratio $r \in (0,1)$ is equivalent to the random-turn game $\RT(\G,\frac{r+\tau\cdot (1-r)}{1+\tau})$.

Our proof for all-pay poorman can be extended to all-pay taxman bidding. Since the notation is already heavy, to ease the presentation we omit the general proof. The properties of a mean-payoff game $\G$ with taxman parameter $\tau \in [0,1]$ and budgets $X$ for  \Max and $Y$ for \Min are as follows. Let $\norm{X} = X + \tau Y$ and $\norm{Y} = Y + \tau X$. With deterministic strategies, when $X>Y$, we have $\sMP(\G, r) \geq \MP\big(\RT(\G, 1 - \frac{\norm{Y}}{\norm{X}})\big)$, and when $X \leq Y$, deterministic strategies are useless. For mixed strategies, when $X>Y$, we have $\asMP(\G, r) = \MP\big(\RT(\G, 1 - \frac{\norm{Y}}{2\norm{X}})\big)$ and when $X \leq Y$, we have $\sMP(\G, r) =  \MP\big(\RT(\G, \frac{\norm{X}}{2\norm{Y}})\big)$.

\subsection{Qualitative all-pay bidding games}
\label{sec:qual}
We focus on {\em parity} objectives, which are important, for example, since the problem of LTL synthesis reduces to solving a parity game \cite{PR89}. 
\begin{definition}{\bf (Parity objectives).}
A parity game is played on a graph $\zug{V, E, p}$, where $p: V \rightarrow \Nat$ is a parity function. \PO wins an infinite play iff the maximal index that is visited infinitely often is odd. 
\end{definition}

Under first-price bidding, parity bidding games reduce to reachability bidding games. The proof relies on a lemma shown in \cite{AHC19,AHI18,AHZ19} that in a strongly-connected parity taxman bidding game, one of the players deterministically wins with any positive initial ratio. Intuitively, when the highest parity index is odd, \PO wins with any positive initial budget since no matter how small (but positive) his initial budget is, he can draw the game to the vertex with the highest parity index. Below, we describe a corresponding result for parity all-pay bidding games.

\begin{theorem}
\label{thm:parity}
Consider a strongly-connected parity game with a highest odd parity index, a cycle with highest even parity, and an initial ratio $r \in (0,1)$ for \PO. 
\begin{itemize}[noitemsep,topsep=0pt]
\item Under all-pay Richman: for any $r$, \PO almost-surely wins and cannot surely win in $\G$.
\item Under all-pay poorman: \PO almost-surely wins with any $r$, and surely-wins only when $r > 0.5$.
\end{itemize}
\end{theorem}

In Sec.~\ref{sec:qual-full}, we prove Thm.~\ref{thm:parity} by reducing parity bidding games to mean-payoff bidding games and using Thms.~\ref{thm:AP-poor} and~\ref{thm:AP-Rich}. This proof technique applies also to first-price bidding games and significantly simplifies the previous techniques, which are based on reasoning on reachability bidding games. In parity first-price bidding games, the solution to general games follows from solutions to games played on SCCs and a solution to reachability bidding games. Thm.~\ref{thm:parity} gives one of these ingredients and the second, namely a solution to reachability all-pay bidding games, is yet to be solved.

\subsection{Computational complexity}
\label{sec:compl}
The computational complexity problem we are interested in is given a mean-payoff bidding game and a budget ratio, find the optimal sure or almost-sure value. The following theorem follows from the complexity of the corresponding problem in mean-payoff stochastic games, since random-turn games are a special case of stochastic games.
\begin{theorem}
Given a strongly-connected mean-payoff all-pay Richman or poorman bidding game $\G$ and an initial ratio $r \in (0,1)$, deciding whether the sure or almost-sure value in $\G$ w.r.t. $r$ is at least $0.5$ is in NP and coNP. Deciding whether \PO almost-surely or surely wins a strongly-connected parity all-pay Richman or poorman bidding game can be done in linear time.
\end{theorem}

We leave open the problem of improving the bounds. Since the upper bounds for mean-payoff games are derived from solving random-turn games, the problem is relevant and open also for first-price bidding. It is possible that solving random-turn games is in P and it is possible that it is as hard as solving general stochastic games, which is a long standing open problem.

\subsection{Related work}
All the results surveyed above highly depend on the fact that the players' bids can be arbitrarily small. This is a problematic assumption for practical applications. To address this limitation {\em discrete bidding games} were studied in \cite{DP10}, where the budgets are given in ``cents'' and the minimal positive bid is one cent. Their motivation came from recreational play like bidding chess \cite{BP09,LW18}. Discrete all-pay Richman bidding has been studied in \cite{MWX15} (we encourage the reader to try playing all-pay Richman tic-tac-toe online: \url{https://bit.ly/2WmOjHO}). While the issue of tie breaking does not play a key role in continuous bidding, it is important in discrete bidding \cite{AAH19}. Non-zero-sum first-price Richman games were studied in~\cite{MKT18}.% and used to reason about ongoing negotiations.

\subsection{Outline of the rest of the paper}
Richman bidding is technically easier than poorman bidding. We thus start by proving Thm.~\ref{thm:AP-Rich} for Richman bidding in Sec.~\ref{sec:AP-Rich} while presenting general techniques that will also be used in the proof of Thm.~\ref{thm:AP-poor} for poorman bidding.
Namely, in Sec.~\ref{sec:strength} we describe a framework for extending a solution to $\lolli$ to games played on strongly-connected graphs, which was developed for first-price bidding and extends to all-pay bidding. In Sec.~\ref{sec:martingale}, we survey notations and results from probability theory and martingale theory that will be relevant in the rest of the paper.
The proof of Thm.~\ref{thm:AP-poor} for poorman bidding is presented in  Sec.~\ref{sec:AP-poor}. In Sec.~\ref{sec:parityproof}, we present a proof of Thm.~\ref{thm:parity} on parity bidding games (both Richman and poorman). Finally, we close with a conclusion section (Sec.~\ref{sec:conc}).

\section{Mean-Payoff All-Pay Richman Games}
\label{sec:AP-Rich}
In this section we prove Thm.~\ref{thm:AP-Rich}. We start in Sec.~\ref{sec:richmandet} by proving that pure strategies are useless. In Sec.~\ref{sec:fpRichman} we revisit mean-payoff first-price Richman games and present a new construction of optimal strategies. This serves both as a warm-up for all-pay bidding with mixed strategies and the construction is of independent interest. We then turn to construct optimal mixed strategies.

\subsection{Deterministic strategies are useless}\label{sec:richmandet}
We prove the claim of Thm.~\ref{thm:AP-Rich} on pure strategies; namely, for a strongly-connected mean-payoff game $\G$, for every ratio $r$, we have $\sMP(\G, r) = \MP\big(\RT(\G, 0)\big)$. It suffices to show that for any initial budgets and given a strategy of \Max, \Min can counter it with a strategy that ensures winning all but a constant number of biddings. Thus as the underlying game graph is strongly-connected, given any deterministic \Max strategy \Min can eventually push the game to the cycle in $\G$ of minimal weight and keep looping the cycle. By doing this, \Min ensures the mean-payoff equal to $\MP\big(\RT(\G,0)\big)$. Hence, it remains to prove the following lemma.

\begin{lemma}
\label{lem:negative-Richman}
Let $\G$ be a strongly-connected all-pay Richman bidding game.
For any initial ratio $r \in (0,1)$ and a deterministic strategy of \Max,
\Min has a strategy that wins all but a constant number of biddings.
\end{lemma}
\begin{proof}
Let $\init{B}$ and $\init{C}$ respectively denote \Max and \Min's initial budgets. Suppose \Max plays according to some pure strategy. Let $C$ be \Min's budget prior to a bidding. Suppose \Max bids $b$. Knowing \Max's bid, \Min bids as follows. If $b > C$, \Min bids $0$ and otherwise she bids $b$. \Min's strategy is clearly legal. Recall that \Min wins ties. Thus, every time she wins a bidding, the budgets are unchanged. The only biddings that she loses are the ones in which \Max bids more than $C$. But this can happen at most $\lceil \init{B}/\init{C} \rceil+1$ times.
\end{proof}

\subsection{Warm up; Revisiting mean-payoff first-price Richman games}\label{sec:fpRichman}
Constructions of optimal strategies in mean-payoff first-price Richman games were shown in \cite{AHC19,AHZ19}. The construction we show here is significantly simpler. Moreover, it is the first budget-based strategy (the bids depend only on the current vertex and budget), which will be crucial later in all-pay bidding.

Our constructions throughout the paper are based on the {\em shift function} $\lambda:(0,1) \rightarrow (1,+\infty)$, which is defined as $\lambda(x) = -\frac{\log(1-x)}{\log(1+x)}$.

\stam{
\begin{definition}
\label{def:shift-function}
{\bf (Shift function).} The \emph{shift function} $\lambda:(0,1) \rightarrow (1,+\infty)$ is defined as $\lambda(x) = -\frac{\log(1-x)}{\log(1+x)}$.
\end{definition}
}
%The proof of the following lemma can be found in App.~\ref{app:shift-function}.
\begin{lemma}
\label{lem:shift-function}
The shift function has the following properties:
\begin{itemize}[noitemsep,topsep=0pt]
\item For every $c \in (1,+\infty)$, there exists $\alpha \in (0,1)$ such that $\lambda(\alpha)=c$.
\item For $c \in (1,+\infty)$ and $c = \lambda(\alpha)$, we have $(1-\alpha) = (1+\alpha)^{-c}$.
\end{itemize}
\end{lemma}
\begin{proof}
The shift function is surjective since (1) $\lim_{x \rightarrow 0}\aux(x)=1$ (l'H\^opital rule), (2) $\lim_{x \rightarrow 1}\aux(x)=+\infty$, and (3) $\lambda$ is continuous as its denominator is strictly positive over the domain, and $\log$ is continuous.
\stam{
\begin{enumerate}[(a)]
  \item
  $\lim_{x \rightarrow 0}\aux(x)=1$ (l'H\^opital rule);
  \item
  $\lim_{x \rightarrow 1}\aux(x)=+\infty$;
  \item
  $\lambda$ is continuous as its denominator is strictly positive over the domain, and $\log$ is continuous.
\end{enumerate}
}
As a consequence, for every $y \in [1,\infty)$, there exists $x \in [0,1]$ such that $\lambda(x)=y$.
The second item is a direct consequence of the definition of the shift function.
\end{proof}

For ease of presentation, we illustrate the construction on the simple game $\lolli$, and it can easily be extended to general SCCs using the framework in the next section. 

%Previous constructions of optimal strategies for mean-payoff first-price Richman games can be found in \cite{AHC19,AHZ19}. Below we revisit mean-payoff \FPRich games and devise optimal budget-based strategies. 

\begin{proposition}
\label{prop:Richman-lolli}
In the mean-payoff game $\lolli$ (Fig.~\ref{fig:bowtie}), under first-price Richman bidding, for every initial ratio $r \in (0,1)$ and $\epsilon >0$, 
\Max has a deterministic budget-based strategy that guarantees a payoff of at least $0.5-\epsilon$,
thus the sure mean-payoff value of $\G$ is $\MP\big(\RT(\lolli, 0.5)\big) = 0.5$.
\end{proposition}
\begin{proof}
Let $\epsilon >0$ and let $\init{B}>0$ be \Max's initial budget. We show that \Max can guarantee a payoff of at least $\frac{1}{2+\epsilon}$. We re-normalize the weights to be $w(v_\Max) = c = 1+\epsilon$ and $w(v_\Min) = -1$. Recall that the energy of a finite play is the sum of the weights it traverses. The following observation is a direct consequence of the definition of payoff. 

\noindent{\bf Observation:} Suppose \Max plays according to a strategy that guarantees that the energy is bounded from below by a constant. Then, the payoff with the updated weights is non-negative, and the payoff with the original weights is at least $\frac{1}{2+\epsilon}$.

Let $\alpha$ such that $\alpha = \lambda(c)$ (see Lem.~\ref{lem:shift-function}). Let $k_0 \in \Nat$ be the initial energy. We devise a strategy of \Max that maintains the invariant that when the energy is $k \in \Nat$, his budget exceeds $(1+\alpha)^{-k-k_0}$. The invariant implies $k > -k_0$. Indeed, recall that the sum of budgets in Richman bidding is $1$. Thus, $k=k_0$ is impossible since the invariant would imply that \Max's budget exceeds $(1+\alpha)^{0} = 1$. The observation above implies that the strategy guarantees a payoff of at least $\frac{1}{2+\epsilon}$, as required. 

We turn to construct \Max's strategy. We choose $k_0\in\mathbb{N}$ such that $\init{B} = B+\delta$, where $B = (1+\alpha)^{-k_0}$ and $\delta >0$. This is possible since $\lim_{k \to \infty} (1+\alpha)^{-k} =0$. We call $\delta$ the ``spare change'', and it is never used for bidding. We refer to $B$ as \Max's main budget. \Max's strategy bids as follows: when \Max's main budget is $B$, he bids $\alpha \cdot B$. Note that the strategy is budget based since the bid depends only on the budget. 

We prove by induction that by following this strategy \Max maintains the invariant that when the energy is $k$, his main budget is at least $(1+\alpha)^{-k-k_0}$. Initially, the invariant holds by our choice of $k_0$. For the inductive step, we distinguish between the two outcomes of a bidding. If \Max loses, the energy decreases to $k-1$. Moreover, \Min overbids \Max, thus \Max's new main budget $B'$ is at least $B + \alpha B \geq \frac{1}{(1+\alpha)^{k+k_0}} + \frac{\alpha}{(1+\alpha)^{k+k_0}} = (1+\alpha)^{-(k-1)-k_0}$. On the other hand, if \Max wins, the energy increases to $k+c$ and his new main budget $B'$ is at least $B - \alpha B = \frac{1 - \alpha}{(1+\alpha)^{k+k_0}}$. Since $(1-\alpha) = (1+\alpha)^{-c}$ (see Lem.~\ref{lem:shift-function}), we obtain $B' = (1+\alpha)^{-(k+c)-k_0}$, and we are done.
\end{proof}

\subsection{A framework for solving SCCs}
\label{sec:strength}
In this section, we describe a framework that was developed in \cite{AHC19,AHI18,AHZ19} for any first-price bidding mechanism and intuitively extends a solution to $\lolli$ to general SCCs. The framework extends from first-price to all-pay bidding and we rely on it throughout the rest of the paper. 

Intuitively, in Prop.~\ref{prop:Richman-lolli}, in order to bound the energy from below, we bid in such a way that bounds the difference between \Max bidding wins and loses in a finite play. In $\lolli$, the bound on the difference of wins translates immediately to a bound on the energy, and thus to a guarantee on the payoff. In general SCCs, vertices have different ``importance'', called {\em strength}, and each bid is ``scaled'' according to its strength. The strengths are chosen in such a way that in each finite path, bounding the difference of bidding wins and losses implies a bound on the accumulated energy.

The definition of strengths relies on {\em potentials}, which were originally defined in the context of the strategy iteration algorithm \cite{How60}. Let $\G$ be a strongly-connected mean-payoff game and $p \in (0,1)$. It is well-known that optimal {\em positional} strategies exist in mean-payoff stochastic games \cite{Put05}; namely, strategies in which moves depend only on the current position. Consider such optimal strategies $\sigma_\Max$ and $\sigma_\Min$ for the two players. For every vertex $v$, we denote $v^+ = \sigma_\Max(v)$ and $v^- = \sigma_\Min(v)$. Intuitively, when \Max and \Min win a bidding in $v$, they should move to $v^+$ and $v^-$, respectively. We denote the potential of a vertex $v$ by $\Pot_p(v)$ and the strength of $v$ by $\St_p(v)$, and we define them as solutions to the following equations. The potential equation roughly coincides with the equation to compute the expected energy in a path to a target. 

\begin{equation*}
\begin{split}
&\Pot_p(v) = p \cdot \Pot_p(v^+) + (1-p) \cdot \Pot_p(v^-) + w(v) - \MP(\RT(\G,p))\\
&\St_p(v) = p\cdot (1-p) \cdot \big(\Pot_p(v^+) - \Pot_p(v^-)\big)
\end{split}
\end{equation*}

%We suggest an intuitive way to read the definition of potentials. Consider a weighted Markov chain in which each vertex $v$ has two neighbors and the probability of proceeding from $v$ to $v^+$ and $v^-$ is respectively $p$ and $1-p$. Then, the potential of $v$ roughly coincides with the expression to compute the expected energy in a path from $v$ to the target.
Note that $\St(v) \geq 0$, for every $v \in V$.
We denote the maximal strength by $S_{\max} = \max_{v \in V} \St_p(v)$
and we assume $S_{\max} > 0$ otherwise the game is trivial as all weights are equal.

Consider a finite path $\eta = v_1,\ldots, v_n$ in $\G$.
We intuitively think of $\eta$ as a play,
where for every $1 \leq i < n$,
the bid of \Max in $v_i$ is $\St(v_i)$
and he moves to $v_i^+$ upon winning.
Thus, when $v_{i+1} = v_i^+$, we think of \Max as {\em investing} $\St_p(v_i)$
and when $v_{i+1} \neq v_i^+$, we think of \Min winning the bid
thus \Max\ {\em gains} $\St_p(v_i)$.
We denote by $\investmax(\eta)$ and $\gainmax(\eta)$ the sum of investments and gains,
respectively. The difference between \Max's wins and loses in $\eta$ is then $\investmax(\eta) - \gainmax(\eta)$. 
Note that $\investmax(\eta)$ and $\gainmax(\eta)$ are defined w.r.t $\RT(\G, p)$ and $p$ will be clear from the context.
Recall that the energy of $\eta$ is the sum of weights it traverses.
The following lemma connects the energy, potentials, and strengths.

\begin{lemma}
\label{lem:magic}
\cite{AHC19,AHI18,AHZ19}
Consider a strongly-connected game $\G$, and $p = \frac{\nu}{\nu+\mu} \in (0,1)$, thus \Max's budget is $\nu$ and \Min's budget is $\mu$. For every finite path $\eta$ from $v$ to $u$, we have \[\Pot_p(v) - \Pot_p(u) + (n-1)\cdot\MP\big(\RT(\G,p)\big) \leq \energy(\eta) + \frac{\nu+\mu}{\nu\mu}\cdot\big( \gainmax(\eta) \cdot \nu - \investmax(\eta)\cdot \mu \big).\]
\end{lemma}

For example, suppose $\eta$ is a cycle, i.e., $u = v$, that $\nu=\mu=1$, and $\MP\big(\RT(\G,p)\big)=0$. Then, exactly as in $\lolli$, we have $\investmax(\eta)-\gainmax(\eta) \leq \energy(\eta)$ (and equality holds when \Min plays optimally). See \cite{AHI18,AHC19} for further examples.
%A simple instantiation of Lem.~\ref{lem:magic} is when $\eta$ is a cycle, from every vertex $v$, $\eta$ continues either to $v^+$ or $v^-$, the ratio is $0.5$, thus $\nu=\mu=1$, and $\MP\big(\RT(\G, 0.5)\big)=0$. The lemma then says that $\energy(\eta) = \investmax(\eta)-\gainmax(\eta)$. Thus, in order to bound the energy from below
We obtain the following corollary by dividing both sides by $n$, 
and letting $n$ tend towards infinity.

\begin{corollary}\label{cor:magic}
Consider a strongly-connected game $\G$, let $\eta$ be an infinite play and let $\mu,\nu \in \Real_{> 0}$.
For every $n \in \mathbb{N}$, let $\eta^n$ denote the prefix of $\eta$ of size $n$.
Then
\[
\payoff(\eta)
\geq
\MP(\RT(\G, p)) +
\frac{\mu + \nu}{\mu\nu} \cdot \liminf_{n\rightarrow \infty}\frac{\mu \cdot \investmax(\eta^n) -\nu \cdot \gainmax(\eta^n)}{n}.
\]
\end{corollary}

\subsection{Mean-payoff all-pay Richman games with mixed strategies}\label{sec:mixedAPRich}
This section consists of the more technically challenging proof of Thm.~\ref{thm:AP-Rich}: we show that no matter the initial ratios, the optimal almost-sure payoff under all-pay Richman bidding equals  the optimal expected payoff in an un-biased random-turn game, thus all-pay and first-price Richman bidding coincide. We illustrate the ideas behind the construction in the following example.

\begin{example}
\label{ex:Richman}
We describe a simple \Max strategy for $\lolli$, which achieves an expected payoff of $0.25$; still not optimal, but better than any deterministic strategy can achieve. We start with the following observation. Suppose \Max chooses a bid uniformly at random from $\set{0, b}$, for some $b > 0$. We assume \Min wins ties. Thus, knowing \Max's strategy, \Min chooses between deterministically bidding $0$ or $b$. There are four possible outcomes (see Fig.~\ref{tab:Richman}). The ``bad'' outcomes for \Max are $\zug{0,0}$ and $\zug{b,b}$ since \Min wins without any budget penalty. The two other outcomes are ``good'' since they are similar to first-price Richman bidding: \Max pays $b$ for winning and gains $b$ when losing. To choose $b$, we rely on an optimal bidding strategy $f_\FP$ for first-price Richman bidding. As seen in Prop.~\ref{prop:Richman-lolli}, $f_\FP$ guarantees that in any finite play, \Max wins roughly half the biddings. Under all-pay Richman bidding, consider a finite play $\pi$ and let $\pi'$ be the restriction of $\pi$ to good bidding outcomes. We choose $b = f_\FP(\pi')$. Intuitively, we expect half the outcomes in a play to be good, out of these, $f_\FP$ guarantees that \Max wins half the biddings, for a total expected payoff of $0.25$. 

We minimize the probability of ending in a bad outcome by bidding uniformly at random in $[0, b]$. This opens a spectrum between good and bad outcomes: \Max is ``lucky'' if his bid is either just above \Min's bid or way below it. We show that lucky events cancel unlucky events, which we formally prove by defining a submartingale called {\em luck} that sums \Max's luck in a finite play. Finally, we note that it is technically not possible to define such a mixed bidding strategy when $f_\FP$ is not budget-based and the previous constructions in \cite{AHC19,AHZ19} are not budget based, hence the importance of the new proof of Prop.~\ref{prop:Richman-lolli}.\hfill$\triangleleft$
\end{example}

\begin{figure}[ht]
\center
\includegraphics[height=2cm]{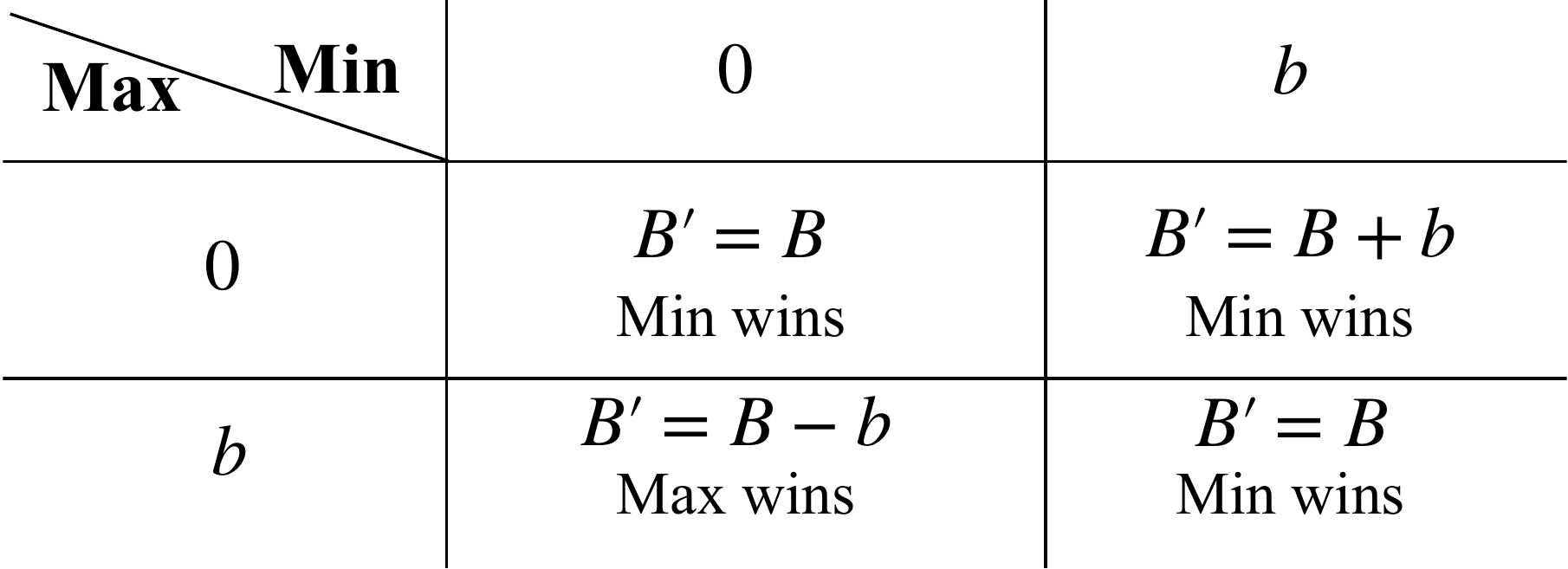}
\caption{\small \Max's budget updates in four bidding outcomes under \APRich.}% assuming \Max bids according to rows.}
\label{tab:Richman}
\end{figure}

\begin{lemma}
\label{lem:MP-Rich}
Let $\G$ be a strongly-connected mean-payoff all-pay Richman bidding game.
For every initial ratio $r \in (0,1)$ of \Max and for every  $\epsilon > 0$,
\Max has a mixed budget-based strategy in $\G$
that guarantees almost-surely a payoff of at least $\MP(\RT(\G, \frac{1}{2 + \epsilon}))$.
\end{lemma}
\begin{proof}
Let $c = 1+\epsilon$, and let $p = \frac{1}{2 + \epsilon} = \frac{1}{1+c}$. We fix two optimal positional strategies in the random-turn game $\RT(\G, \frac{1}{2 + \epsilon})$ for $\Max$ and $\Min$, and use them to define vertex strengths and neighboring vertices $v^+$ and $v^-$ for each vertex $v$ as in Sec.~\ref{sec:strength}.
%We find vertex strengths using $\RT(\G, p)$ as in Section~\ref{sec:strength}. 
%Let $f$ and $g$ be deterministic optimal strategies for Max and Min respectively in the random-turn game on $G$ with coin bias $\frac{1}{k+1}$. For each vertex $v\in V$, let $f(v)=v^+$ and $g(v)=v^-$, and define the potential function $Pot:V\rightarrow \mathbb{R}$ and the strength function $St:V\rightarrow \mathbb{R}_{\geq 0}$ w.r.t.~these two optimal strategies. Recall, potential function can be chosen so that $Pot(v^-)\leq Pot(v) \leq Pot(v^+)$ for each $v\in V$, thus allowing us to assume nonnegative strengths. Finally, let $S_{\max}=\max_{v\in V}St(v)$. From the defining relation of potentials and strengths, all strengths being equal to $0$ is equivalent to all weights in the graph being $0$, in which case the claim of the proposition is trivial. Hence, we may w.l.o.g.~assume that $S_{\max}>0$.
Let $\alpha \in (0,1)$ s.t. $\lambda(\alpha) = 1+\epsilon$ (see Lem.~\ref{lem:shift-function}). % \in (0,1)$ s.t. $1-\alpha=(1+\alpha)^{-c}$ (for each $c>1$ there is unique such $\alpha$). 
%Suppose that the budgets of Max and Min are normalized so they sum up to $1$. 
We define \Max's strategy $f$ in the bidding game $\G$ as follows:
\begin{itemize}[noitemsep,topsep=0pt]
	\item When the token is on vertex $v$ with strength $s = \St_p(v)$, and \Max's budget is $B$, \Max bids
	$x \sim \Unif [0,\alpha B\frac{s}{S_{\max}}]$.
	\item Upon winning, \Max moves the token to $v^+$.
\end{itemize}

The strategy $f$ is clearly budget-based, so in order to prove the lemma it suffices to show that, no matter which mixed strategy $g$ \Min chooses,
we have $\Pr_{\pi \sim dist(f,g)}[\lim \inf_{n \to \infty} \payoff(\pi) \geq \MP\big(\RT(\G, p)\big)] = 1$.

\stam{
Proofs in first-price bidding show an invariant between changes in energy and changes in budget (see Prop.~\ref{prop:Richman-lolli}).
For all-pay bidding, however, such an invariant is not possible.
Our invariant uses a new component, which we refer to as ``luck''.
Intuitively, a \Max bid of $x$ is ``unlucky'' when $x$ is much larger than $y$ or when $x$ is slightly smaller than $y$.
In the first case, \Max pays a lot for winning and in the second, \Min pays little for winning.
The top-left to bottom-right diagonal in Tab.~\ref{tab:Richman} take these scenarios to the extreme.
Dually, when $x$ is just above $y$ or way below $y$, \Max is ``lucky''.
See the other diagonal in the table.
The key idea of the proof is that on average, the unlucky and lucky cases cancel out. 
}

Fix a mixed strategy $g$ of $\Min$. Intuitively, consider the event in which \Max bids $x$ and \Min bids $y$. \Max's budget gain is $x-y$. \Max is ``lucky'' when $x-y$ is maximized, which happens either when $x$ is slightly above $y$ (then \Max pays little for winning) or when $x$ is way lower than $y$ (then \Max gains a lot when losing). We formalize luck below and later show that the expected luck is non-negative in each bidding. Assume Wlog that at each turn \Min bids $y \in [0,\alpha B\frac{s}{S_{\max}}]$, since she has the tie-breaking advantage and does not profit from bidding higher. For any infinite play $\pi$ that can arise from strategies $f$ and $g$, define $L_0(\pi) = \log_{1+\alpha} r$, and for each $i\in\mathbb{N}$ let
\begin{equation}\label{eq:deltah}
	\Delta L_i(\pi) = L_i(\pi) - L_{i-1}(\pi) = \begin{cases}
	c(s+2S_{\max}\frac{y-x}{\alpha B}), &\text{if $x>y$}, \\
	(-s+2S_{\max}\frac{y-x}{\alpha B}), &\text{if $x\leq y$,}
	\end{cases}
\end{equation}
where $B$ is $\Max$'s budget and $x$ and $y$ are the bids of $\Max$ and $\Min$ at the $i$-th bidding in $\pi$, respectively.

\stam{
\noindent{\bf Claim}: For $i \geq 1$, we have $\mathbb{E}_x [L_i(\pi)] \geq 0$, thus $L$ is a sub-martingale. %Suppose that following a finite play, \Max bids $x\sim \Unif[0,\alpha B\frac{s}{S_{\max}}]$, then for every \Min bid $y\in [0,\alpha B\frac{s}{S_{\max}}]$, we have 
%\begin{equation}\label{equ:luck}
%\mathbb{E}_{x\sim \Unif[0,\alpha B\frac{s}{S_{\max}}]}[\Delta L] \geq 0.
%\end{equation}

Let $\beta$ denote $\alpha B \frac{s}{S_{\max}}$, let us fix a finite play $\pi$
and a bid $y \in [0,\alpha B \frac{s}{S_{\max}}]$ of \minim{}.
We obtain:
\[
\begin{array}{lll}
\mathbb{E}_{x\sim \Unif[0,\beta]}[\Delta L] & = &
\frac{1}{\beta}\int_0^{\beta} \! \Delta L(x) \, \mathrm{d}x\\ & = &
\frac{1}{\beta} \Big( \int_0^{y} \! \Delta L(x) \, \mathrm{d}x +
\int_y^{\beta} \! \Delta L(x) \, \mathrm{d}x\Big)\\  & = &
\frac{1}{\beta} \Big(\int_0^{y} \! (- \frac{s}{2S_{max}} + \frac{y-x}{\alpha B}) \, \mathrm{d}x +
\int_y^{\beta} \! c(\frac{s}{2S_{max}} + \frac{y-x}{\alpha B}) \, \mathrm{d}x\Big)\\ & = &
\frac{1}{\beta} \Big( y(- \frac{s}{2S_{max}} + \frac{y}{\alpha B})-\frac{y^2}{2 \alpha B} +
c(\beta-y)(\frac{s}{2S_{max}} + \frac{y}{\alpha B}) - \frac{c\beta^2}{2 \alpha B} + \frac{cy^2}{2 \alpha B} \Big).
\end{array}
\]
Since $\beta = \alpha B \frac{s}{S_{\max}}$, we may substitute $\frac{s}{S_{\max}}=\frac{\beta}{\alpha B}$ above to get
\[
\begin{array}{lll}
\mathbb{E}_{x\sim \Unif[0,\beta]}[\Delta L] & = &
\frac{1}{\beta}\Big( y(-\frac{\beta}{2\alpha B} + \frac{y}{\alpha B}) - \frac{y^2}{2\alpha B} + c(\beta-y)(\frac{\beta}{2\alpha B} + \frac{y}{\alpha B}) - \frac{c\beta^2}{2\alpha B} + \frac{cy^2}{2\alpha B} \Big)\\ & = &
\frac{1}{\beta}\Big( \frac{y(y-\beta)}{2\alpha B} + \frac{c\beta y}{2\alpha B} - \frac{cy^2}{2\alpha B} \Big)\\ & = &
\frac{1}{\beta}\Big( \frac{y(y-\beta)}{2\alpha B} + \frac{cy(\beta-y)}{2\alpha B} \Big)\\ & = &
\frac{(c-1)y(\beta-y)}{2\beta\alpha B} \geq 0,
\end{array}
\]
where the last inequality follows since $c>1$ and $y\in [0,\beta]$.\hfill (of claim) $\triangleleft$
\stam}

In Prop.~\ref{prop:Richman-lolli}, we devise an invariant between \Max's budget and the energy of a finite play. Here, the invariant is more involved. First, the graph is more involved than $\lolli$, thus we depend on the framework in Sec.~\ref{sec:strength} and bound the difference between \Max's wins and loses in a finite play, formally denoted $\investmax(\pi^n) -\gainmax(\pi^n)$, for a prefix $\pi^n$ of a play $\pi$. Recall that Lem.~\ref{lem:magic} implies that such a bound also implies a bound on the energy. Second, in all-pay bidding, we incorporate the luck into the invariant.
%For any finite prefix $\pi^n=\pi_0,\pi_1,\dots,\pi_n$ of $\pi$, recall that $\investmax(\pi^n)$ and $\gainmax(\pi^n)$ are the sums of the strengths of vertices where \Max wins and loses the bidding, respectively. 
Let $B(\pi^n)$ denote \Max's budget following the finite play $\pi^n$ and we use $L(\pi^n)$ instead of $L_n(\pi)$ above. 
The following claim identifies the key invariant that holds throughout the game and on which the rest of our proof is based.

\begin{restatable}{claimth}{AP-Rich-inv}
\label{cl:AP-Rich-inv}
For every finite prefix $\pi^n$ of an infinite play $\pi$ coherent with the strategies $f$ and $g$, we have
\begin{equation}\label{eq:inv_APrichman}
B(\pi^n) \geq (1+\alpha)^{\frac{L(\pi^n)-c\cdot\investmax(\pi^n)+\gainmax(\pi^n)}{2S_{\max}}}.
\end{equation}
\end{restatable}
\noindent{\bf Proof of Claim \ref{cl:AP-Rich-inv}:} Let $\pi^n$ be a finite prefix of a play.
To ease notation, we write $\diff = c\cdot\investmax(\pi^n)-\gainmax(\pi^n)$ and omit references to $\pi^n$.
We show that
%\begin{equation*}
$B \geq (1+\alpha)^{\frac{L-\diff}{2S_{\max}}}$.
%\end{equation*}
We proceed by induction on $n$.
The base case follows from our choice of $L_0$.
Suppose by induction that the equation holds for the values $B$ and $\diff$ obtained after prefix $\pi^n$,
and that in the next bidding \Max bids $x$ and \Min bids $y$.
Let $B'=B+y-x=B+\Delta B$ and $L'=L+\Delta L$.
We want to show that $B'\geq (1+\alpha)^{\frac{L'-\diff'}{2S_{\max}}}$,
where $\diff' = \diff + cs$ if \Max wins, and $\diff' = \diff - s$ if \Min wins.
% if \Min wins the bidding,
% and that $B'\geq (1+\alpha)^{L-\frac{1}{2S_{\max}}(c\cdot(\gainmax(\pi)+s)-\investmax(\pi))}$
% if \Max wins the bidding.
By the definition of $\Delta L$,
we get
\begin{equation}\label{eq:deltab}
	\Delta B = \begin{cases}
	(\Delta L + s)\frac{\alpha B}{2S_{\max}}, &\text{if $x\leq y$};\\
	(\frac{\Delta L}{c}-s)\frac{\alpha B}{2S_{\max}}, &\text{if $x>y$}.
	\end{cases}
	\end{equation}
To conclude, we distinguish between the case in which \Min wins and \Max wins:
\begin{enumerate}
				\item If \Min wins the bidding, i.e.~$x\leq y$, then
				$\diff' = \diff - s$,
				and we get:
				\begin{equation*}
				\begin{array}{lcl}
				B' & \stackrel{{\color{white}()}}{=} & B + \Delta B
				\stackrel{\eqref{eq:deltab}}{=} B + (\Delta L + s)\frac{\alpha B}{2S_{\max}} \\
				   & \stackrel{{\color{white}()}}{=} & B \cdot (1 + (\Delta L + s)\frac{\alpha}{2S_{\max}})
				   \stackrel{\text{Bernoulli}}{\geq} B \cdot (1+\alpha)^{\frac{\Delta L+s}{2S_{\max}}} \\
				   &\stackrel{\text{ind. hyp.}}{\geq} &
				   (1+\alpha)^{\frac{L + \Delta L -\diff+s}{2S_{\max}}}\\
				   & \stackrel{{\color{white}()}}{=} &
				   (1+\alpha)^{\frac{L' - \diff'}{2S_{\max}}}.
				\end{array}
				\end{equation*}
				Here, Bernoulli's inequality could be used since $\alpha > -1$ and $\frac{\Delta L+s}{2S_{\max}}=\frac{y-x}{\alpha B}\in [0,1]$.
				
				\item If \Max wins the bidding, i.e., $x>y$, then $\diff' = \diff + cs$,
				and we get:
				\begin{equation*}
				\begin{array}{lcl}
				B' & = & B+\Delta B
				\stackrel{\eqref{eq:deltab}}{=} B + (\frac{\Delta L}{c}-s)\frac{\alpha B}{2S_{\max}}\\
				& \stackrel{{\color{white}()}}{=} & B \cdot (1-\alpha(-\frac{\Delta L}{c\cdot 2S_{\max}}+\frac{s}{2S_{\max}}))
				\stackrel{\text{Bernoulli}}{\geq} B \cdot (1-\alpha)^{-\frac{\Delta L}{c\cdot 2S_{\max}}+\frac{s}{2S_{\max}}}\\
				& \stackrel{\text{Lemma}~\ref{lem:shift-function}}{\geq} & B \cdot (1+\alpha)^{\frac{\Delta L-cs}{2S_{\max}}}
				\stackrel{\text{ind. hyp.}}{\geq}
				(1+\alpha)^{\frac{L + \Delta L -\diff - cs}{2S_{\max}}}\\
				& \stackrel{{\color{white}()}}{=} &
				(1+\alpha)^{\frac{L'-\diff'}{2S_{\max}}}
				\end{array}
				\end{equation*}
				Here, Bernoulli's inequality could be used since $-\alpha > -1$ and $-\frac{\Delta L}{c\cdot 2S_{\max}}+\frac{s}{2S_{\max}}=\frac{x-y}{\alpha B}\in [0,1]$.
				We also used Lemma \ref{lem:shift-function}: $1-\alpha = (1+\alpha)^{-c}$ since $\lambda(\alpha) = c$.
			\end{enumerate}\hfill (of claim) $\triangleleft$

Since the sum of budgets of the players is $1$, we have that $B(\pi^n)\leq 1 = (1+\alpha)^0$. Hence, by comparing the exponents in eq.~\eqref{eq:inv_APrichman} we obtain $c\cdot\investmax(\pi^n)-\gainmax(\pi^n) \geq L(\pi^n)$. On the other hand, by plugging $\nu = 1$ and $\mu=c$ into Lemma~\ref{lem:magic}, since $p=1/(c+1)$ we obtain $\frac{c+1}{c}(c\cdot\investmax(\pi^n)-\gainmax(\pi^n)) \leq \energy(\pi^n) - P - n\cdot\MP(\RT(\G,p))$ (note that there are $n+1$ vertices along $\pi^n$, hence the factor $n$). Combining the two inequalities gives
\begin{equation}\label{eq:mpenergy}
\begin{split}
\energy(\pi^n) %&= \mathbb{E}[\frac{(c+1)S_{\max}}{c}E+P+(n-1) \MP(\RT(\G, p))] \\
\geq \frac{c+1}{c}\cdot L(\pi^n)+P+n\cdot\MP(\RT(\G,p)).
\end{split}
\end{equation}
The last equation holds for any finite prefix $\pi^n$ of an infinite play $\pi$, so by dividing both sides by $n$ and letting $n\rightarrow \infty$ we get
\begin{equation}\label{eq:mpenergy}
\payoff(\pi) %&= \mathbb{E}[\frac{(c+1)S_{\max}}{c}E+P+(n-1) \MP(\RT(\G, p))] \\
\geq \frac{c+1}{c}\cdot\liminf_{n\rightarrow\infty}\frac{L_n(\pi)}{n}+\MP(\RT(\G,p)).
\end{equation}
The infinite play $\pi$ was arbitrary, hence eq.~\eqref{eq:mpenergy} holds for any play that is coherent with $f$ and $g$. 

To conclude the lemma, recall that \Max tries to maximize his luck. We prove an almost-sure lower-bound on the luck in the following claim:

\begin{restatable}{claimth}{LB-luck}
\label{cl:LB-luck}
$\mathbb{P}_{\pi\sim dist(f,g)}[\liminf_{n\rightarrow\infty}\frac{L_n}{n}\geq 0]=1$.
\end{restatable}

Proving this last claim implies the lemma. Indeed, since eq.~\eqref{eq:mpenergy} holds for any infinite play $\pi$ we conclude that $\payoff(\pi)\geq \MP(\RT(\G,p))$ almost-surely, and thus $f$ guarantees the desired mean-payoff almost-surely. The proof of this claim, however, is intricate. We regard $(L_n)_{n=0}^{\infty}$ as a stochastic process in the probability space over the set of all infinite plays defined by strategies $f$ and $g$. We then show that $(L_n)_{n=0}^{\infty}$ is a {\em submartinale}, %(w.r.t.~a suitably chosen filtration). 
which intuitively means that for every infinite play $\pi$, the expectation of $L_{n+1}(\pi)$ given the finite history $\pi^n$ is at least $L_n(\pi)$. The claim then follows from results from martingale theory.
In the following section we introduce the necessary background and describe the proof.
\end{proof}

Since $\MP\big(\RT(\G, p)\big)$ is continuous in $p$ \cite{Cha12,Sol03}, it follows from Lem.~\ref{lem:MP-Rich} that \Max can ensure a payoff of at least $\MP\big(\RT(\G, 0.5)\big)$, with every initial budget ratio. To deduce that $\asMP(\G, r)=\MP(\RT(\G,0.5))$,  for every $r$, we also need to show that \Min can ensure a payoff of at most $\MP\big(\RT(\G, 0.5)\big) + \epsilon$ for any $\epsilon>0$ and with every initial budget ratio. To construct an optimal strategy for \Min we rely on the advantage that the definition of payoff (Def.~\ref{def:MP}) gives to \Min. We consider the game $\G^-$ obtained from $\G$ by negating the weight of each vertex. It is not hard to show that $\MP\big(\RT(\G^-, 0.5)) = - \MP\big(\RT(\G, 0.5))$. Then, to ensure a payoff of at most $\MP\big(\RT(\G, 0.5)) +\epsilon$, \Min follows an optimal \Max strategy in $\G^-$. The symmetry argument is standard and has already been used in the first-price bidding games setting~\cite{AHC19,AHI18,AHZ19}, so we omit the details. Hence we have $\asMP(\G, r)=\MP\big(\RT(\G,0.5)\big)$ for every $r$, which concludes the proof of Thm.~\ref{thm:AP-Rich}.

\subsection{An aside on martingale theory}\label{sec:martingale}

In this section, we presents results on martingale theory which are needed to prove Claim~\ref{cl:LB-luck}, as well as in the later parts of this paper. We start with an intermezzo on necessary background on probability theory and martingale theory. We keep this exposition brief. For more details,  we refer the reader to~\cite{Williams:book}.
%and only present notions and results that are used in this work. For a detailed introduction to martingale theory, we refer the reader to~\cite{xxx}.

A {\em probability space} is a triple $(\Omega,\mathcal{F},\mathbb{P})$, where $\Omega$ is a non-empty {\em sample space}, $\mathcal{F}$ is a {\em sigma-algebra} over $\Omega$ which is a collection of subsets of $\Omega$ which is closed under complementation and countable unions and contains $\emptyset$, and $\mathbb{P}:\mathcal{F}\rightarrow [0,1]$ is a function such that $\mathbb{P}[\emptyset]=0$, $\mathbb{P}[\Omega\backslash A]=1-\mathbb{P}[A]$ for each $A\in\mathcal{F}$, and $\mathbb{P}[\cup_{i=1}^{\infty}A_i]=\sum_{i=1}^{\infty}\mathbb{P}[A_i]$ for a sequence of pairwise disjoint sets $A_1,A_2,\dots$ in $\mathcal{F}$. An element of $\mathcal{F}$ is said to be an {\em event}.

In the case of a bidding game $\G$ and mixed strategies $f$ and $g$ of $\Max$ and $\Min$, we let $\Omega_{\G}$ be the set of all infinite plays in $\G$, $\mathcal{F}_{\G}$ be the unique smallest sigma-algebra which contains all subsets of $\Omega_\G$ defined by plays with a common finite prefix (thus every finite play defines one such set), and $\mathbb{P}$ be the probability measure $dist(f,g)$ defined by the {\em cylinder construction}~\cite[Theorem 2.7.2]{AD00}. Since the construction of $dist(f,g)$ is standard but technical, we omit it from this exposition however we note that it satisfies all intuitive properties.

We say that a sequence $(\mathcal{F}_i)_{i=0}^{\infty}$ of sigma-algebras in $(\Omega,\mathcal{F},\mathbb{P})$ is a {\em filtration} if $\mathcal{F}_0\subseteq\mathcal{F}_1\subseteq \dots\subseteq\mathcal{F}$. In the case of bidding games, we are particularly interested in the so-called {\em canonical filtration} $(\mathcal{R}_i)_{i=0}^{\infty}$ of $(\Omega_{\G},\mathcal{F}_{\G},dist(f,g))$. Each $\mathcal{R}_i$ is defined as the smallest sigma-algebra containing all subsets of $\Omega_\G$ defined by plays with a common finite prefix of length at most $i$. Intuitively, $\mathcal{R}_i$ contains those events which are defined by what happened in $\G$ during the first $i$ steps.

A {\em random variable} $X$ in $(\Omega,\mathcal{F},\mathbb{P})$ is an $\mathcal{F}$-measurable function $X:\Omega\rightarrow \mathbb{R}$ (w.r.t.~the standard Lebesgue measure on $\mathbb{R}$), i.e.~a function for which $\{\omega\in\Omega\mid f(\omega)\leq c\}\in\mathcal{F}$ for each $c\in\mathbb{R}$. A {\em stochastic process} $(X_i)_{i=0}^{\infty}$ is a sequence of random variables in $(\Omega,\mathcal{F},\mathbb{P})$.

\medskip Before being able to define (sub)martingales, we need to introduce one more important notion. Let $(\Omega,\mathcal{F},\mathbb{P})$ be a probability space, $X$ a random variable, and $\mathcal{F}'\subseteq\mathcal{F}$ a sigma-sub-algebra of $\mathcal{F}$. The {\em conditional expectation} of $X$ w.r.t.~$\mathcal{F'}$ is an $\mathcal{F}'$-measurable random variable $Y$ such that, for each $A\in\mathcal{F}'$, we have that $\mathbb{E}[Y\cdot 1_A] = \mathbb{E}[X\cdot 1_A]$. Here, $1_A$ is an indicator function of $A$, defined as $1_A(\omega)=1$ if $\omega\in A$ and $1_A(\omega)=0$ otherwise.

Intuitively, conditional expectation of $X$ w.r.t.~$\mathcal{F}'$ is an $\mathcal{F}'$-measurable random variable which captures the behavior of $X$ on those events contained in $\mathcal{F}'$. Note that $X$ is not necessarily equal to its conditional expectation as $X$ need not be $\mathcal{F}'$-measurable. In fact, conditional expectation of a random variable $X$ need not even exist. However, it is known that whenever $X$ is {\em integrable} (meaning that $\mathbb{E}[|X|]<\infty$), the conditional expectation of $X$ w.r.t.~$\mathcal{F}'$ exists and is almost-surely unique. Almost-sure uniqueness means that, if two random variables $Y$ and $Y'$ satisfy the definition of conditional expectation of $X$ w.r.t.~$\mathcal{F}'$, then $\mathbb{P}[Y=Y']=1$. In this case, we denote any such random variable as $\mathbb{E}[X\mid \mathcal{F}']$.

We are finally ready to define the notion of a submartingale.

\begin{definition}[Submartingale]
Let $(\Omega,\mathcal{F},\mathbb{P})$ be a probability space, $(\mathcal{F}_i)_{i=0}^{\infty}$ a filtration and $(X_i)_{i=0}^{\infty}$ a stochastic process. Then we say that $(X_i)_{i=0}^{\infty}$ is a {\em submartingale} w.r.t.~$(\mathcal{F}_i)_{i=0}^{\infty}$ if
\begin{itemize}
    \item for each $i\in\mathbb{N}_0$, $X_i$ is integrable and $\mathcal{F}_i$-measurable, and
    \item for each $i\in\mathbb{N}_0$, $\mathbb{E}[X_{i+1}\mid\mathcal{F}_i]\geq X_i$ almost-surely.
\end{itemize}
If in the second point above we have equality for each $i$, we say that $(X_i)_{i=0}^{\infty}$ is a {\em martingale} w.r.t.~$(\mathcal{F}_i)_{i=0}^{\infty}$.
\end{definition}

The following theorem is the key result from martingale theory that will be needed in our proofs.

\begin{theorem}[Azuma-Hoeffding inequality~\cite{azuma1967}]
Let $(\Omega,\mathcal{F},\mathbb{P})$ be a probability space and $(\mathcal{F}_i)_{i=0}^{\infty}$ a filtration. Suppose that $(X_i)_{i=0}^{\infty}$ is a submartingale w.r.t.~$(\mathcal{F}_i)_{i=0}^{\infty}$, and suppose that there exists $c>0$ such that $|X_{i+1}-X_i|\leq c$ almost-surely for each $i\in\mathbb{N}_0$. Then, for each $N\in\mathbb{N}_0$ and $\epsilon>0$ we have that
\begin{equation*}
    \mathbb{P}[X_i-X_0\leq -\epsilon] \leq \mathrm{e}^{\frac{-\epsilon^2}{2Nc^2}}.
\end{equation*}
\end{theorem}

We conclude this section by introducing and proving a lemma on submartingales, which follows from the Azuma-Hoeffding inequality and which will be the main technical ingredient for studying mean-payoffs guaranteed by mixed strategies constructed in our proofs.

\begin{lemma}\label{lemma:concentration}
Let $(\Omega,\mathcal{F},\mathbb{P})$ be a probability space and $(\mathcal{F}_i)_{i=0}^{\infty}$ a filtration. Suppose that $(X_i)_{i=0}^{\infty}$ is a submartingale w.r.t.~$(\mathcal{F}_i)_{i=0}^{\infty}$, and suppose that there exists $c>0$ such that $|X_{i+1}-X_i|\leq c$ almost-surely for each $i\in\mathbb{N}_0$. Furthermore, suppose that $X_0\geq K$ for some $K\in\mathbb{R}$. Then
\begin{equation*}
    \mathbb{P}\Big[\liminf_{N\rightarrow\infty} \frac{X_N}{N}\geq 0 \Big] = 1.
\end{equation*}
\end{lemma}

\begin{proof}
Let $A_0=\{\liminf_{N\rightarrow\infty}X_N / N\geq 0 \}$ be the event whose probability we want to show is $0$. For each $\delta\in\mathbb{Q}_{\geq 0}$, let $A_{-\delta}=\{\liminf_{N\rightarrow\infty}X_N / N < -\delta \}$. Then $A_0 = \Omega\backslash (\cup_{\delta\in\mathbb{Q}_{\geq 0}}A_{-\delta})$, thus it suffices to prove that $\mathbb{P}[\cup_{\delta\in\mathbb{Q}_{\geq 0}}A_{-\delta}]=0$. By the union bound, we have $\mathbb{P}[\cup_{\delta\in\mathbb{Q}_{\geq 0}}A_{-\delta}]\leq \sum_{i=0}^{\infty}\mathbb{P}[A_{-\delta}]$, so it also suffices to prove that $\mathbb{P}[A_{-\delta}]=0$ for each $\delta\in\mathbb{Q}_{\geq 0}$.

Fix $\delta\in\mathbb{Q}_{\geq 0}$. Note that $\liminf_{N\rightarrow\infty}X_N / N < -\delta$ if for each $M\in\mathbb{N}$ there exists $N\geq M$ such that $X_N / N < -\delta$. For fixed $M\in\mathbb{N}$, denote this event by $A^M_{\delta}$. Then $A_{-\delta}=\cap_{M\in\mathbb{N}}A^M_{-\delta}$, and so $\mathbb{P}[A_{-\delta}]\leq \mathbb{P}[A^M_{-\delta}]$ for each $M$. Hence, if we prove that $\lim_{M\rightarrow\infty}\mathbb{P}[A^M_{-\delta}] = 0$, it follows that $\mathbb{P}[A_{-\delta}]=0$ as wanted.

Fixing $M\in\mathbb{N}$ and rewriting the definition of the event $A^M_{-\delta}$, we see that
\begin{equation*}
    A^M_{-\delta} = \cup_{N=M}^{\infty}\{X_N/N < -\delta\} = \cup_{N=M}^{\infty}\{X_N - X_0 < -K-N\cdot\delta\}.
\end{equation*}
By the union bound and by letting $M$ be sufficiently large so that $M>|K|/\delta$, we have
\begin{equation}\label{eq:martingalebound}
\begin{split}
    \mathbb{P}[A^M_{-\delta}] &\leq \sum_{N=M}^{\infty} \mathbb{P}[X_N - X_0 < -K-N\cdot\delta] \leq \sum_{N=M}^{\infty} \mathrm{e}^{\frac{-(-K-N\cdot\delta)^2}{2Nc^2}} \\
    &\leq \sum_{N=M}^{\infty} \mathrm{e}^{\frac{-(-2N\cdot\delta)^2}{2Nc^2}} = \sum_{N=M}^{\infty} \mathrm{e}^{\frac{-2N\cdot\delta^2}{c^2}},
\end{split}
\end{equation}
where for the second inequality we used Azuma-Hoeffding, and for the third inequality that $M>|K|/\delta$. As the series $\sum_{N=0}^{\infty} \mathrm{e}^{-2N\cdot\delta^2/c^2}$ converges since it is a geometric series, we have that the sum on the RHS of eq.~\eqref{eq:martingalebound} tends to $0$ as $M\rightarrow\infty$. Thus we conclude that $\lim_{M\rightarrow\infty}\mathbb{P}[A^M_{-\delta}] = 0$, which finishes the proof.

\end{proof}

\noindent{\bf Proof of Claim~\ref{cl:LB-luck}:}
Recall the definition of $L_i(\pi)$ for each infinite play $\pi$ and $i\in\mathbb{N}_0$. Consider $(L_i)_{i=0}^{\infty}$ as a stochastic process over the probability space $(\Omega_\G,\mathcal{F}_\G,dist(f,g))$ defined by $\G$ and mixed strategies $f$ and $g$ of \Max and \Min, respectively.

Observe that each $L_i$ is $\mathcal{R}_i$-measurable, where recall $\mathcal{R}_i$ is the $i$-th sigma-algebra of the canonical filtration on $(\Omega_{\G},\mathcal{F}_{\G},dist(f,g))$. This is because both the budget $B$ and bids $x$ and $y$ depend only on the first $i$ steps of the game. Moreover, for each $i\in\mathbb{N}$ we easily see that $|\Delta L_i(\pi)|\leq 3c\cdot S_{\max}$ for every $\pi$, thus by triangle inequality $|L_i(\pi)|\leq |\log_{1+\alpha} r|+3ic\cdot S_{\max}$ and so each $L_i$ is integrable.

\noindent{\bf Claim}: $(L_i)_{i=0}^{\infty}$ is a submartingale w.r.t.~the canonical filtration.

The measurability and integrability conditions were checked above. It remains to show that the conditional expectations property holds, i.e.~$\mathbb{E}[L_{i+1}\mid\mathcal{R}_i]\geq L_i$ for each $i\in\mathbb{N}_0$. We first describe what this conditional expectation looks like. Let $\pi\in\Omega_\G$ and let $\pi^i$ be its prefix of length $i$. Then
\[
    \mathbb{E}[L_{i+1}\mid\mathcal{R}_i](\pi) = L_{i-1}(\pi) + \mathbb{E}_{x\sim \Unif[0,\beta], y\sim g(\pi^i)}\Big[ c(s+2S_{\max}\frac{y-x}{\alpha B})\cdot 1_{x>y} + (-s+2S_{\max}\frac{y-x}{\alpha B})\cdot 1_{x\leq y} \Big],
\]
where $\beta=\alpha B\frac{s}{S_{\max}}$ and $g(\pi^i)$ is the distribution over the bids of $\Min$ defined by the mixed strategy $g$ and a finite history $\pi^i$. The fact that this is indeed the right expression for conditional expectation follows from the definition of canonical filtration. Formally showing this is technical but straightforward, so we omit it.

To prove that $(L_i)_{i=0}^{\infty}$ is a submartingale, it thus suffices to prove that
\begin{equation*}
    \mathbb{E}_{x\sim \Unif[0,\beta], y\sim g(\pi^i)}\Big[ c(s+2S_{\max}\frac{y-x}{\alpha B})\cdot 1_{x>y} + (-s+2S_{\max}\frac{y-x}{\alpha B})\cdot 1_{x\leq y} \Big] \geq 0
\end{equation*}
for each $\pi\in\Omega_\G$. By Fubini's theorem (which can be applied since the integrand is a bounded function), we may rewrite this expectation as
\[
\begin{array}{lll}
&\mathbb{E}_{x\sim \Unif[0,\beta], y\sim g(\pi^i)}\Big[ c(s+2S_{\max}\frac{y-x}{\alpha B})\cdot 1_{x>y} + (-s+2S_{\max}\frac{y-x}{\alpha B})\cdot 1_{x\leq y} \Big] \\
&= \mathbb{E}_{y\sim g(\pi^i)}\Big[\mathbb{E}_{x\sim \Unif[0,\beta]}\Big[c(s+2S_{\max}\frac{y-x}{\alpha B})\cdot 1_{x>y} + (-s+2S_{\max}\frac{y-x}{\alpha B})\cdot 1_{x\leq y}\Big]\Big].
\end{array}
\]
Hence to prove non-negativity, it suffices to show that for any $y\in[0,\beta]$ the inner expectation is non-negative. Fix $y\in[0,\beta]$. We obtain:
\[
\begin{array}{lll}
\mathbb{E}_{x\sim \Unif[0,\beta]}&\Big[c(s+2S_{\max}\frac{y-x}{\alpha B})\cdot 1_{x>y} + (-s+2S_{\max}\frac{y-x}{\alpha B})\cdot 1_{x\leq y}\Big] \\
& = 
\frac{1}{\beta} \Big(\int_0^{y} \! (- s + 2S_{\max}\frac{y-x}{\alpha B}) \, \mathrm{d}x +
\int_y^{\beta} \! c(s + 2S_{\max}\frac{y-x}{\alpha B}) \, \mathrm{d}x\Big)\\ & = 
\frac{1}{\beta} \Big( y(- s + 2S_{\max}\frac{y}{\alpha B})-\frac{y^2S_{\max}}{\alpha B} +
c(\beta-y)(s + 2S_{\max}\frac{y}{\alpha B}) - \frac{c\beta^2S_{\max}}{\alpha B} + \frac{cy^2S_{\max}}{\alpha B} \Big).
\end{array}
\]
Since $\beta = \alpha B \frac{s}{S_{\max}}$, we may substitute $\frac{s}{S_{\max}}=\frac{\beta}{\alpha B}$ above to get
\[
\begin{array}{lll}
\mathbb{E}_{x\sim \Unif[0,\beta]}&\Big[c(s+2S_{\max}\frac{y-x}{\alpha B})\cdot 1_{x>y} + (-s+2S_{\max}\frac{y-x}{\alpha B})\cdot 1_{x\leq y}\Big] \\
& = 
\frac{S_{\max}}{\beta}\Big( y(-\frac{\beta}{\alpha B} + \frac{2y}{\alpha B}) - \frac{y^2}{\alpha B} + c(\beta-y)(\frac{\beta}{\alpha B} + \frac{2y}{\alpha B}) - \frac{c\beta^2}{\alpha B} + \frac{cy^2}{\alpha B} \Big)\\ & =
\frac{S_{\max}}{\beta}\Big( \frac{y(y-\beta)}{\alpha B} + \frac{c\beta y}{\alpha B} - \frac{cy^2}{\alpha B} \Big)\\ & = 
\frac{S_{\max}}{\beta}\Big( \frac{y(y-\beta)}{\alpha B} + \frac{cy(\beta-y)}{\alpha B} \Big)\\ & = 
\frac{(c-1)y(\beta-y)}{\beta\alpha B} \geq 0,
\end{array}
\]
where the last inequality follows since $c>1$ and $y\in [0,\beta]$.\hfill (of claim) $\triangleleft$

Thus, $(L_i)_{i=0}^{\infty}$ is a submartingale w.r.t.~the canonical filtration. Moreover, as we observed above it has differences bounded by $3c\cdot S_{\max}$ and also $L_0 = \log_{1+\alpha} r$ which is thus bounded below by a real constant. Hence we may apply Lemma~\ref{lemma:concentration} to conclude that $\mathbb{P}[\liminf_{n\rightarrow\infty}\frac{L_n}{n}\geq 0]=1$, which is precisely Claim~\ref{cl:LB-luck}, hence we are done.\hfill\qed

\section{Mean-Payoff All-Pay Poorman Games}
\label{sec:AP-poor}

This section is devoted to the proof of Thm.~\ref{thm:AP-poor}. We first revisit first-price poorman bidding and describe a significantly simpler proof. We unify the proofs for all-pay poorman bidding by introducing and studying a variant of bidding games called {\em asymmetric bidding games}. We first show the connection between values in asymmetric bidding games and all-pay poorman games. Then, using similar (though more involved) techniques as in the Richman setting, we show that in asymmetric bidding games the sure and almost-sure values do not depend on the initial ratios.

\subsection{Warm up; revisiting mean-payoff first-price poorman games.}
The value of mean-payoff first-price poorman games was first identified in \cite{AHI18}.

\begin{theorem}[\cite{AHI18}]
\label{prop:FP-poorman}
Let $\G$ be a strongly-connected mean-payoff all-pay poorman bidding game.
For every initial ratio $r \in (0,1)$ of \Max,
we have $\sMP_\FPPoor(\G, r) = \MP\big(\RT(\G, r)\big)$.
\end{theorem}

We revisit this result and provide an alternative proof
by constructing new and significantly simpler optimal budget-based bidding strategies.
%We then base our solution to all-pay bidding on the budget-based strategy.

\begin{lemma}
\label{lem:FP-poorman}
Let $\G$ be a strongly-connected mean-payoff all-pay poorman bidding game.
For every initial ratio $r = \frac{\init{B}}{\init{B}+\init{C}} \in (0,1)$ of \Max, for every  $\epsilon > 0$,
\Max has a deterministic budget-based strategy that guarantees a payoff of $\MP(\RT(\G, \frac{\init{B}-\epsilon}{\init{B}+\init{C}}))$.
\end{lemma}

\begin{proof}
Let $B_0 \in \Real$ be \Max's initial budget.
Throughout this proof, we keep \Min's budget normalized to $C = 1$ and use $B$ to denote \Max's budget. Thus, assuming \Max bids $x$ and \Min bids $y$, when \Max wins the bidding ($b>a$), we have $B'= B-b$, and when \Min wins the bidding \Max's new budget is $B' = \frac{B}{1-y}$. 
Let $\epsilon> 0$, and let $W = \init{B} -\epsilon$.
We construct a pure \Max strategy that maintains the invariant that $B \geq W$.
The key new insight is that when \Max loses a bidding, we have 
\begin{equation}\label{equation:budget_FPP_approx}
B' = \frac{B}{1-y} > \frac{B(1-y^2)}{1-y} = B(1+y) > B + Wy.
\end{equation}
Intuitively, the property states that every cent is $W$ times more valuable to \Min than it is to \Max.
For example, if \Max's budget is $2$ and \Min's budget is $1$,
then paying $0.1$ is twice as painful for \Min as it is for \Max.
Roughly, on average, this means that \Max wins $W \sim B_0$ times more biddings than \Min,
thus he guarantees a payoff close to $\MP\big(\RT(\G, \frac{\init{B}}{\init{B}+1})\big)$.

We now proceed to define formally a budget-based bidding strategy $f$ for \Max
that guarantees a payoff of at least $\MP\big(\RT(\G, p)\big)$, where $p= (\init{B}-\epsilon)/(\init{B}+1)$.
We pick $\alpha \in (0,1)$ satisfying
$\lambda(\alpha) = 1+\epsilon$ (see Lemma~\ref{lem:shift-function}).
We find vertex strengths using $\RT(\G, p)$ as in Section~\ref{sec:strength}.
Let $N = \max(W,1) \cdot S_{\max}$.
The strategy $f$ is defined as follows.
\begin{itemize}
	\item When the token is placed on a vertex $v$ with strength $s = \St_p(v)$
	and \Max's budget is $B$, \Max bids $f(B,s) = \frac{\alpha\cdot s}{N}(B-W)$.
	\item Upon winning, \Max moves the token to $v^+$.
\end{itemize}
We first show that \Max's bidding strategy $f$ is legal, by showing that we always have $B>W$.
Indeed, initially, we have $\init{B} > W$,
and whenever \Max loses a bidding his budget increases,
and when \Max wins a bidding his updated budget is $B-f(B,s) = B - \frac{\alpha\cdot s}{N}(B-W)$,
which is still greater than $W$ since $\frac{\alpha\cdot s}{N}<1$.

Next, for any finite play $\pi$, let
$\diff(\pi) = (1+\epsilon) \cdot \investmax(\pi)- (\init{B} - \epsilon) \cdot \gainmax(\pi) -N \cdot \log_{1+\alpha}(\epsilon)$.
Recall that $\investmax(\pi)$ and $\gainmax(\pi)$
denote the sum of the strengths of the vertices of $\pi$ in which \Max wins and loses, respectively.
We prove that the budget $B(\pi)$ of \Max after the play $\pi$
satisfies the following invariant, using induction on the length of $\pi$ and Bernoulli's inequality. The proof is similar to proofs for claims on asymmetric bidding games, which can be found in the full version.

\noindent{\bf Claim:} For every finite play $\pi$ coherent with the strategy $f$ of \Max, we have
\begin{equation}\label{eq:inv_FPPoorman}
    (B(\pi)-W)^N \geq (1+\alpha)^{-\diff(\pi)}.
\end{equation}

Next, we show that the claim above implies a lower bound on $\diff$. We describe the key ideas and similar proofs can be found for asymmetric bidding games in the full version. Observe Eq.~\ref{eq:inv_FPPoorman}. Since both $N$ and $W$ are constants, when $\diff(\pi)$ shrinks, the equation implies that $B(\pi)$ must grow, and in turn \Max's bid grows since it depends on $(B(\pi) - W)$. When \Max's bid is greater than $1$, he necessarily wins the bidding since \Min's budget is fixed to $1$, causing $\diff$ to increase.

\noindent{\bf Claim:} There exists $M \in \mathbb{R}$ such that 
for every finite play $\pi$ coherent with $f$, we have
\begin{equation}\label{eq:bound_FPPoorman}
(1+\epsilon) \investmax(\pi) - (B_0-\epsilon) \gainmax(\pi) \geq M.
\end{equation}

\stam{
%The proof of the claim, found in App~\ref{app:FP-poorman_bound}, can be summarised as follows.
We describe the key ideas of the proof below. Again, similar proofs can be found for asymmetric bidding games in the full version. 
Since the left-hand side of the equation is equal to $\diff(\pi) + N \cdot \log_{1+\alpha}(\epsilon)$,
proving the claim is equivalent to proving a lower bound for $\diff(\pi)$.
To do so, we show that $\diff(\pi)$ cannot get too low,
as past some threshold Equation \eqref{eq:inv_FPPoorman}
guarantees that the budget $B(\pi)$ of \Max is so high that
his next bid according to the strategy $f$
will be above $1$ (the whole budget of \Min).
Thus \Max is guaranteed to win the next bidding, which results in $\diff(\pi)$ going back up.
}

Combining the claim above %short Claim \ref{eq:bound_FPPoorman} 
with Corollary~\ref{cor:magic}
(plugging $\nu = B_0-\epsilon$ and $\mu=1+\epsilon$),
we obtain that any infinite play coherent with the strategy $f$
has a mean-payoff greater than $\MP\big(\RT(\G, \frac{\init{B}-\epsilon}{B_0+1})\big)$.
\end{proof}

\subsection{Asymmetric bidding games}
In this section we study the properties of asymmetric bidding games defined as follows.

\begin{definition}
{\bf (Asymmetric bidding games).}
For $W > 0$, a $W$-asymmetric game is a bidding game with the following payment scheme. Suppose Player~$1$ and~$2$'s bids are respectively $x$ and $y$. Then, \PO pays $x$ and \PT's pays $y\cdot W$ (hence the name ``asymmetric''). The budgets are updated as follows. We keep \PT's budget constant at $1$. Suppose \PO's budget is $B$, then his new budget is $B' = B-x + y \cdot W$.
\end{definition}

\stam{
In order to solve poorman games, we proceed in two steps:
First, we introduce a variant of Richman games,
the \emph{asymmetric games},
and we reduce poorman games to asymmetric games.
Then, we solve asymmetric games by using similar arguments to the ones used in Section \ref{sec:AP-Rich}.

\subsection{Asymmetric Bidding Games}
Given $W > 0$, a $W$-asymmetric game is a bidding game with the following payment scheme.
The budget of \PT is constantly normalised to $1$,
while the budget of \PO fluctuates according to a variant of the Richman payment scheme:
each round, \PO loses his bid and/or gains the bid of his opponent multiplied by $W$.
Formally, if we assume that \PO's budget is $B$ and his bid is $x \in [0,B]$,
while \PT's budget is $C = 1$ and her bid is $y \in [0,1]$,
then the budgets are updated as follows:
$C' = C = 1$ and $B' = B - x + Wy$.
}

\stam{%short -- possibly keep this for the future
\begin{itemize}[noitemsep]
\item {\bf First-price:} Only the highest bid matters:
$C' = C = 1$ and
$B' = 
\left\{
\begin{array}{ll}
B + Wy & \textup{ if } x \leq y;\\
B - x & \textup{ if } x > y.\end{array}
\right.$
\item {\bf All-pay:} Both bids matter: $C' = C = 1$ and $B' = B - x + Wy$.
\end{itemize}
}

The following theorem, whose proof can be found in the following sections (Lemmas~\ref{lem:pure-poorman-useless}, \ref{lem:APPdMax}, and~\ref{lemma:MP_APP_up} for pure strategies and Lemmas~\ref{lem:APPpMax} and~\ref{lem:APPpMin} for mixed strategies), shows that asymmetric bidding games have similar properties to Richman bidding: the initial budgets do not matter and the game has values w.r.t. pure and mixed strategies.

\begin{theorem}
\label{thm:values-asym}
(Informal) Let $\norm{\G}$ be a strongly-connected mean-payoff $W$-asymmetric bidding game. Then, 
\begin{itemize}[noitemsep,topsep=0pt]
\item {\bf Pure strategies:} For $W > 1$, with any positive initial ratio \Max can guarantee a sure-payoff that is arbitrarily close to $\MP\big(\RT(\norm{\G}, 1-\frac{1}{W})\big)$, and this is optimal. 
\item {\bf Mixed strategies:} With any positive initial ratio \Max can guarantee an almost-sure payoff that is arbitrarily close to $\MP\big(\RT(\norm{\G}, 1-\frac{1}{2W})\big)$, when $W >1$, and arbitrarily close to $\MP\big(\RT(\norm{\G}, \frac{W}{2})\big)$, when $W \leq 1$. 
\end{itemize}
\end{theorem}

%Consider a strongly-connected mean-payoff game $\G$ and a ratio $r$. In the next section we show that asymmetric games are similar to Richman bidding games in that the initial budget ratio essentially does not matter. Thus, the asymmetric game $\norm{G}$ that corresponds to $\G$ w.r.t. some $W$ has a sure or almost-sure value. 

The following lemma relates the values of a $W$-asymmetric game $\norm{\G}$ with the values of $\G$ under all-pay poorman w.r.t. a ratio $r$. Intuitively, we obtain a strategy for \Max in $\G$ by simulating his strategy in $\norm{\G}$. Technically, to simulate the strategy in $\norm{\G}$, we need \Max's budget to be at least $\delta> 0$. This is indeed a technicality since the theorem above shows that the value does not depend on the ratio in asymmetric bidding games.

%We prove that budget-based strategies can easily be adapted from asymmetric games to poorman games.

\begin{lemma}\label{lem:asym}
Consider a mean-payoff all-pay poorman game $\G$ with initial ratio $r = \frac{\init{B}}{\init{B} + \init{C}}$.
Let $W <  \frac{\init{B}}{\init{C}}$, let $\norm{G}$ denote the $W$-asymmetric bidding game
played on the same graph as $\G$. There exists $\delta > 0$ such that
%and let $\init{\norm{B}}$ denote the budget $\frac{\init{B}}{\init{C}} - W$.
\[
\sMP(\G, r) \geq \sMP(\norm{G}, \delta),
\textup{ and } \asMP(\G, r) \geq \asMP(\norm{G}, \delta).
\]
\end{lemma}

\begin{proof}
Suppose the initial budgets in $\G$ are $\init{B}$ and $\init{C}$ for \Min and \Max, respectively. Let $W < \frac{\init{B}}{\init{C}}$ and $\norm{\G}$ be the $W$-asymmetric bidding game that corresponds to $\G$. Let $\init{\norm{B}} = \frac{\init{B}}{\init{C}} - W$ and $\delta = \frac{\init{\norm{B}}}{\init{\norm{B}} + 1}$. We construct a strategy for \Max in $\G$ that simulates his strategy in $\norm{\G}$ when his initial budget in the latter is $\delta$. We use $C$ and $B$ to respectively denote \Min's and \Max's budgets in $\G$ and $\norm{B}$ to denote \Max's budget in $\norm{\G}$. \Max bids to maintain the following invariant:
\[\frac{B}{C} - W \geq \norm{B}.\] 

The definition of $\init{\norm{B}}$ implies that the invariant holds initially. Assuming the invariant holds, we show how \Max maintains it. Suppose \Max's bid in $\norm{\G}$ is $\norm{x} \in [0,\norm{B}]$. Then, \Max bids $\norm{x} \cdot C$ in $\G$. We claim that the bid is legal, i.e., that $\norm{x} \cdot C \leq B$. Indeed, plugging in $\norm{x} \leq \norm{B}$ in the invariant, and multiplying both sides by $C$, we obtain $\norm{x} \cdot C \leq B-C\cdot W \leq B$. 

Suppose \Min bids $y \leq C$ in $\G$. We simulate his strategy in $\norm{\G}$ using the bid $\norm{y} = \frac{y}{C}$. Recall that \Min's budget in $\norm{\G}$ is $1$. Since $y \leq C$, we have $\norm{y} \leq 1$, thus the bid is legal. Note that $\norm{y}$ accurately simulates $y$ since $x > y$ if and only if $\tilde{x} > \tilde{y}$. That is, \Max wins the bidding in $\G$ if and only if he  wins in $\norm{\G}$. Thus, when \Min wins the bidding in $\G$, we can continue the simulation of $\norm{\G}$ by using \Min's move in $\G$.

To conclude, we show that the invariant is preserved. We first claim that $\frac{B - x}{C} \geq W$. Indeed, \Max's bid in $\norm{\G}$ is $\norm{x} = \frac{x}{C}$ and re-arranging the invariant, we obtain $\frac{B}{C} \geq W- \norm{B}$. Combining the two we have:
    \[
    \frac{B - x}{C}
    \geq \frac{B}{C} - \frac{x}{C}
     = \frac{B}{C} - \norm{x}
    \geq W + \norm{B} - \norm{x}
    \geq W.
    \]
Recall that under all-pay poorman bidding, the budgets are updated to $B'= B-x$ and $C'= C-y$.  Combining with the above, we obtain
    \[
    \frac{B'}{C'}
    = \frac{B - x}{C - y}
    = \frac{B - x}{C} + \frac{\frac{B - x}{C}y}{C - y}
    \geq \frac{B - x}{C} + \frac{Wy}{C - y} \geq \]\[
    \geq \frac{B - x + Wy}{C}.
    \]

Recall that in the $W$-asymmetric bidding game \Max's budget is updated to $\norm{B}'= \norm{B}-\norm{x}+W\norm{y}$. 
We conclude by combining with the above:
    \begin{equation*}
    \frac{B'}{C'} - W
    \geq \frac{B - x + Wy}{C} - W
    \geq \norm{B} - \norm{x} + W \norm{y}
    = \norm{B}'.
    \end{equation*}
\stam{%old version
    Let $\init{\norm{B}} = \frac{\init{B}}{\init{C}} - W$ and  
    $\delta = \frac{\init{\norm{B}}}{\init{\norm{B}} + 1}$.
    This is a direct consequence of the fact that
    \Max can use the poorman game $\G$ with initial budgets $\init{B}$ and $\init{C}$
    to completely simulate the $W$-asymmetric game $\norm{G}$ with initial budget $\init{\norm{B}}$
    by preserving the following invariant:
    \[\frac{B}{C} - W \geq \norm{B},\]
    where $B$ and $C$ are the budgets of \Max and \Min in $\G$, and $\norm{B}$ is the budget of \Max in $\norm{\G}$.
    
    Note that, by definition of $\init{\norm{B}}$, the invariant holds initially,
    hence \Max can initiate the simulation.
    Let us now show how \Max can push the simulation further as long as the invariant is preserved:
    \Max simulates the action of bidding $\norm{x} \in [0,\norm{B}]$ in the game $\norm{G}$
    by bidding $\norm{x} \cdot C$ in the game $\G$.
    Note that this bid is valid, since, by the invariant,
    $\norm{x} \cdot C \leq \norm{B} \cdot C \leq B$.
    Then, for every bid $y \in [0,C]$ of \Min in $\G$,
    we can continue the simulation by supposing that \Min bids $\tilde{y} = \frac{y}{C} \in [0,1]$ in $\norm{G}$.
    Note that bidding $x$ against $y$ in $\G$ accurately represents
    bidding $\norm{x}$ against $\norm{y}$ in $\norm{\G}$:
    $x > y$ if and only if $\tilde{x} > \tilde{y}$,
    hence \Max wins the round in $\G$ if and only if he also wins in $\norm{\G}$,
    and can pursue the simulation by choosing in $\G$ the same successor vertex as in $\norm{\G}$.
    To conclude, we show that the invariant is preserved:
    by definition of the budget update in a poorman game, $B$ and $C$ are updated as
    $B'= B-x$ and $C'= C-y$, and by definition of the budget update in a $W$-asymmetric game,
    $\norm{B}$ is updated as $\norm{B}'= \norm{B}-\norm{x}+W\norm{y}$.
    Since the invariant was holding previously,
    and $\norm{x} \leq \norm{B}(\norm{\pi})$,
    we get
    \[
    \frac{B - x}{C}
    \geq \frac{B}{C} - \frac{x}{C}
     = \frac{B}{C} - \norm{x}
    \geq W + \norm{B} - \norm{x}
    \geq W.
    \]
    As a consequence, we obtain
    \[
    \frac{B'}{C'}
    = \frac{B - x}{C - y}
    = \frac{B - x}{C} + \frac{\frac{B - x}{C}y}{C - y}
    \geq \frac{B - x}{C} + \frac{Wy}{C - y}
    \geq \frac{B - x + Wy}{C}.
    \]
    Therefore,
    \begin{equation*}
    \frac{B'}{C'} - W
    \geq \frac{B - x + Wy}{C} - W
    \geq \norm{B} - \norm{x} + W \norm{y}
    = \norm{B}'.
    \end{equation*}
    }%of stam - old version
\end{proof}

%%%%%%%%%%%%%%%%%%%%%%%%%%%%%%%%%%%%%%%%%%%%
\subsection{Mean-payoff all-pay poorman games under pure strategies}\label{subsec:APPdMax}
%We now prove the part of Theorem \ref{thm:AP-poor} concerning deterministic strategies:
In this section, we show that, for a strongly-connected mean-payoff all-pay poorman bidding game $\G$
with initial budgets $\init{B}$ for \Max and $\init{C}$ for \Min
(thus the ratio is $r = \frac{\init{B}}{\init{B}+\init{C}}$), we have
\[
\sMP(\G, r) =
\left\{
\begin{array}{ll}
     \MP\big(\RT(\G, 1-\frac{\init{C}}{\init{B}})\big) & \textup{ if } \init{B}>\init{C};\\
     \MP\big(\RT(\G, 0)\big)  & \textup{ if } \init{B}\leq\init{C}.
\end{array}
\right.
\]

%\subsubsection*{Pure strategies are useless for $r \leq 0.5$}
We start with the second part and show that deterministic strategies are useless when the initial budget of \Max is not larger than the initial budget of \Min.
%deterministic strategies are useless once again, as \Min can ensure she wins every bidding.
%For this particular result, we do not need to use asymmetric games,
%and we prove it directly for poorman games.

\begin{lemma}
\label{lem:pure-poorman-useless}
Consider a strongly-connected mean-payoff all-pay poorman game $\G$
an initial budget $\init{B}$ of \Max and an initial budget $\init{C}$ of \Min such that $\init{B} \leq \init{C}$.
Then, \Min can counter every deterministic strategy of \Max
with a strategy that wins all biddings.
\end{lemma}

\begin{proof}
\Min maintains the invariant that her budget exceeds \Max's budget while winning all biddings.
The assumptions of the lemma imply that the invariant holds initially. 
Suppose \Max's budget is $B$, \Min's budget is $C \geq B$, and \Max bids $b$. Then, \Min bids $b$ as well, wins the bidding (since she wins ties), and the new budgets are $C-b \geq B-b$ thus the invariant is maintained.
%Since $\init{B} \leq \init{C}$,
%starting from the first round
%\Min can always match the bid of \Max,
%winning the round (since she wins ties)
%while keeping a greater budget (since in an all-pay poorman game both bids are payed to the bank).
\end{proof}

% In order to get the desired result when the initial budget of \Max is higher than the initial budget of \Min,
% we prove the following lemma:

% \begin{lemma}
% \label{lem:APPdMax}
% In a strongly-connected mean-payoff all-pay poorman game $\G$ where the initial budget $\init{B}$ of \Max
% is greater than the initial budget $\init{C}$ of \Min,
% for every $\epsilon > 0$, \Max has a deterministic budget-based strategy that guarantees an expected value of
% $\MP(\RT(\G, 1-(1+\epsilon)\frac{\init{C}}{\init{B}}))$.
% \end{lemma}

%\subsubsection*{Optimal pure strategies for \Max}

%We proceed to the case where $\init{B} > \init{C}$, and show that $\sMP(\G,r) = \MP(\RT(\G, (1-\frac{C_0}{B_0})))$.
%First, we prove that $\sMP(\G,r) \geq \MP(\RT(\G, (1-\frac{C_0}{B_0})))$.
We continue to the case in which \Max's budget is greater than \Min's budget. Consider a strongly-connected all-pay poorman game $\G$ and initial budget $\init{B}$ and $\init{C}$ such that $\init{B} > \init{C}$. Lem.~\ref{lem:APPdMax} below shows that for every $W>1$ and every positive initial budget,
\Max has a strategy in the asymmetric game $\norm{\G}$ that ensures a payoff
which is arbitrarily close to $\MP\big(\RT(\G, 1-\frac{1}{W})\big)$.
By letting $W\rightarrow\init{B}/\init{C}$ from below and using the connection between asymmetric bidding games and all-pay poorman bidding games in Lem.~\ref{lem:asym}, we show that \Max can guarantee a payoff in $\G$ that is arbitrarily close to $\MP\big(\RT(\G, 1-\frac{\init{C}}{\init{B}})\big)$.

%In Lemma~\ref{lemma:MP_APP_up} we then construct an optimal strategy for \Min, thus showing that this bound is tight and proving the claim of Theorem~\ref{thm:AP-poor} concerning deterministic strategies.

%We proceed to the case where $\init{B} > \init{C}$.  The following lemma formalizes the intuition that the initial ratio does not matter in asymmetric games when $W>1$. Consider a mean-payoff all-pay poorman game $\G$. It shows that for every $W > 1$, the sure-value of the asymmetric game $\norm{\G}$ essentially equals $\MP\big(\RT(\G, 1-\frac{1}{W})\big)$. Letting $W\rightarrow\init{B}/\init{C}$ from below and using Lem.~\ref{lem:asym}, we obtain that the sure value of $\G$ under all-pay poorman bidding with $r = \frac{\init{B}}{\init{B} + \init{C}}$ is $\MP\big(\RT(\G, 1-\frac{\init{C}}{\init{B}})\big)$.

\stam{
To get the desired result when $\init{B} > \init{C}$, 
we rely on Lemma \ref{lem:asym}:
Since, as proved in \cite{Cha12,Sol03}, $\MP\big(\RT(\G, p)\big)$ is continuous in $p$,
it is sufficient to prove that for every $W > 1$,
if we now consider the $W$-asymmetric game $\norm{\G}$
played on the same graph as $\G$, then
for every initial ratio $\norm{r}$ we have
\[
\sMP(\norm{G}, \norm{r}) \geq \MP\big(\RT(\G, 1-\frac{1}{W})\big).
\]
This result is an immediate corollary of the following lemma.
\stam}

\begin{lemma}
\label{lem:APPdMax}
Consider a strongly-connected mean-payoff game $\G$. For $W > 1$, let $\norm{\G}$ be the $W$-asymmetric game obtained from $\G$.
For every $\epsilon > 0$, \Max has a pure budget-based strategy ensuring the payoff
$\MP(\RT(\G, 1-\frac{1+\epsilon}{W + \epsilon}))$.
\end{lemma}

\begin{proof}
Let $\epsilon > 0$.
We set $p = 1-\frac{1+\epsilon}{W+ \epsilon}$,
and find the strengths of the vertices of $\norm{\G}$ using $\RT(\G, p)$.
We define \Max's strategy $f$ in $\norm{G}$ as follows.
Let $\alpha \in (0,1)$ such that $\lambda(\alpha) = 1+\epsilon$ (see Lemma~\ref{lem:shift-function}),
and let $N = \max(S_{\max},(W-1) S_{\max})$.
When the token is on a vertex $v$ with strength $s = \St_{p}(v)$ and \Max's budget is $B$,
\begin{itemize}[noitemsep]
	\item
	\Max bids $\frac{s}{N}\alpha \cdot B$;
	\item Upon winning, \Max moves the token to $v^+$.
\end{itemize}
To prove the lemma, we show that every infinite play $\pi$ coherent with $f$ satisfies
$\payoff(\pi)
\geq
\MP\big(\RT(\G, p)\big)$.
Suppose \Max's budget is $B$, he bids $x$, and \Min bids $y$. Let $B'$ be \Max's updated budget. Then, 

%First, by definition of the budget update in an asymmetric game, we get that
%for every finite play $\pi$ coherent with $f$,
%if in the next bidding \Max bids $x \in \mathbb{R}$ and \Min bids $y \in \mathbb{R}$,
%the updated play $\pi'$ satisfies

\begin{equation*}
B'
=
B-x + Wy
\geq
\left\{
\begin{array}{ll}
B + (W-1)\cdot x & \textup{ if } x \leq y;\\
B-x & \textup{ if } x > y.
\end{array}
\right.
\end{equation*}
\noindent
Where the top bound is obtained since $y \geq x$ and the bottom one since $y \geq 0$.

Intuitively, this means that \Max gains $W-1$ times more for a loss than what he pays for a win. 
As a consequence, \Max wins at least $W-1$ times
whenever \Min wins once.
%or the budget $B$ of \Max increases.
Formally, let $\mu = 1+ \epsilon$ and $\nu = W-1$, and for every finite play $\pi$, let $B(\pi)$ denote \Max's budget following $\pi$ and set
$
\diff(\pi) = \mu \cdot \investmax(\pi)- \nu \cdot \gainmax(\pi)
$, which intuitively keeps track of the difference between the number of biddings \Max loses and wins. That is, as $\diff(\pi)$ drops, \Max wins less biddings. The following claim establishes an invariant between \Max's budget and $\diff(\pi)$. Note that no luck is required since we are dealing with pure strategies.

\begin{restatable}{claimth}{APPdMaxInvariant}
\label{claim:APPdMax_invariant}
    For every finite play $\pi$ coherent with $f$, $B(\pi) \geq \init{B} \cdot (1+\alpha)^{-\frac{\diff(\pi)}{N}}$.
\end{restatable}
\noindent
We use this result to prove a lower bound for $\diff$.
Claim \ref{claim:APPdMax_invariant} implies that when $\diff(\pi)$ drops,
\Max's budget increases and will eventually be so high that his bid will be greater than $1$.
Recall that \Min's budget is set to $1$. 
Thus, \Max wins the bidding causing $\diff$ to go back up.
Therefore:
\begin{restatable}{claimth}{APPdMaxBound}
\label{claim:APPdMax_bound}
    There exists $M \in \mathbb{R}$ such that
    $\diff(\pi) \geq  M$
    for every finite play $\pi$ coherent with $f$.
\end{restatable}
\noindent
Since $\frac{\nu}{\mu + \nu} = 1 - \frac{\mu}{\mu + \nu} = p$, we can now conclude through the use of Corollary \ref{cor:magic}.
For every play $\pi$ coherent with $f$, if for all $n \in \mathbb{N}$ we denote by $\pi^n$
the prefix of $\pi$ of size $n$, we have
\[
\payoff(\pi)
\geq
\MP(\RT(\G, p)) +
\frac{\mu + \nu}{\mu\nu} \cdot \liminf_{n\rightarrow \infty}\frac{H(\pi^n)}{n}
\geq
\MP(\RT(\G, p)).
\qedhere
\]
\end{proof}

We now prove the claims that appear in the proof above.

\paragraph{Proof of Claim \ref{claim:APPdMax_invariant}}
Let $\pi$ be a finite play coherent with the strategy $f$.
We show by induction over the length of $\pi$ that
$B(\pi) \geq \init{B} \cdot (1+\alpha)^{-\frac{\diff(\pi)}{N}}$.
At the start of the game, the equation holds since $\init{B}$ is the initial budget of \Max,
and $H$ is $0$ initially.
For the induction step,
assume that the equation holds for some finite prefix $\rho$ of $\pi$
that ends in a vertex $v$ of strength $s = \St_p(v)$,
and consider the next round.
Let $x$ denote the bid of \Max,
and $y$ denote the bid of \Min.
We show that the equation still holds for the updated play $\rho'$
by differentiating the case where \Max loses the bidding and the one where he wins.
\begin{enumerate}
    \item
        Assume that \Max loses the bidding.
        We start with the definition of the budget update and
        we underapproximate $y$ with $x$, and use the fact that $\nu = W - 1$.
        Then, we apply the fact that, according to the strategy $f$,
        the bid $x$ of \Max is $\alpha \frac{s}{N} \cdot B(\pi)$,
        and we factorise by $B(\rho)$:
        \begin{align}
        B(\rho') = B(\rho) - x + Wy
        &\geq B(\rho) + \nu x
        = B(\rho) +  \frac{\nu s}{N} \alpha \cdot B(\rho)
        = B(\rho) \cdot \Big(1+  \frac{\nu s}{N} \alpha\Big).\nonumber
        \end{align}
        Since $\alpha > -1$ and
        $\frac{\nu s}{N} = \min(\frac{\nu s}{S_{max}},\frac{s}{S_{max}}) \in [0,1]$,
        we can apply Bernoulli's inequality.
        We then use the induction hypothesis,
        and we conclude by using the fact that after one bidding happening at a vertex of strength $s$,
        $\diff$ decreases by at most $\nu s$, i.e., $\diff(\rho') \geq \diff(\rho) - \nu s$.
        \begin{equation}
        B(\rho')
        \geq
        B(\rho) \cdot  (1+\alpha)^{\frac{\nu s}{N}}
        \geq
        \init{B} \cdot (1+\alpha)^{-\frac{\diff(\rho) - \nu s}{N}}
        \geq
        \init{B} \cdot (1+\alpha)^{-\frac{\diff(\rho')}{N}}. \nonumber
        \end{equation}
    \item
        Assume that \Max wins the bidding.
        Once more, we start with the definition of the budget update.
        This time we underapproximate $y$ with $0$,
        we use the fact that the bid $x$ of \Max is $\alpha \frac{s}{N} \cdot B(\rho)$,
        and we factorise by $B(\rho)$:
        \begin{align}
        B(\rho') = B(\rho) - x + W y
        &\geq B(\rho) - x
        = B(\rho) - \frac{s}{N} \alpha B(\rho)
        = B(\rho) \cdot \Big(1 - \frac{s}{N} \alpha\Big). \nonumber
        \end{align}
        Since $-\alpha > -1$ and
        $\frac{s}{N} = \min(\frac{s}{S_{max}},\frac{s}{\nu S_{max}}) \in [0,1]$,
        we can apply Bernoulli's inequality.
        Moreover, by Lemma \ref{lem:shift-function}, since $\lambda(\alpha) = \mu$ by definition,
        we get $(1 - \alpha)^{\frac{s}{N}} = (1 + \alpha)^{-\frac{\mu s}{N}}$.
        We then apply the induction hypothesis, and, finally, we use the fact that, since \Max won the bidding by supposition,
        and \Max chooses the successor $v^+$ whenever he wins,
        we have $\diff(\rho') = \diff(\rho) + \mu s$.
        \begin{equation}
        B(\rho')
        \geq
        B(\rho) \cdot (1-\alpha)^{\frac{s}{N}}
        =
        B(\rho) \cdot (1+\alpha)^{-\frac{\mu s}{N}}
        \geq
        \init{B} \cdot (1+\alpha)^{-\frac{\diff(\rho) + \mu s}{N}}
        =
        \init{B} \cdot (1+\alpha)^{-\frac{\diff(\rho')}{N}}. \nonumber
        \end{equation}
\end{enumerate}
\hfill$\triangleleft$

\paragraph{Proof of Claim \ref{claim:APPdMax_bound}}
Let $M =  - N \cdot \log_{1+\alpha}(\frac{N}{\alpha S_{\min}}) - \nu S_{\max}$.
We show that every finite play $\pi$ coherent with $f$ satisfies $\diff(\pi) \geq  M$.
Let us assume, towards building a contradiction,
that there exists a finite play $\pi$ coherent with $f$
such that the value of $\diff$ drops under $M$ along $\pi$.
Let $\rho'$ denote the smallest prefix of $\pi$ satisfying $\diff(\rho') < M$.
Note that $\rho'$ cannot be the empty play since the value of $\diff$ is initially $0$.
Let $\rho$ be the play obtained by deleting the last step of $\rho'$.
Since in one step the value of $\diff$ decreases by at most $\nu S_{\max}$,
we get that
$\diff(\rho) \leq \diff(\rho') + \nu S_{\max} < - N \cdot \log_{1+\alpha}(\frac{N}{\alpha S_{\min}})$.
By applying Claim \ref{claim:APPdMax_invariant}, we obtain
\begin{equation}\label{eq:APPdMax_bound}
B(\rho) \geq \init{B} \cdot (1+\alpha)^{-\frac{\diff(\rho)}{N}}
> \frac{N}{\alpha S_{\min}}.
\end{equation}
Let $s$ denote the strength of the vertex reached by $\rho$.
If $s = 0$, then $\diff(\rho) = \diff(\rho')$,
which contradicts the fact that $\rho'$ is the smallest prefix of $\pi$ satisfying $\diff(\rho') < M$.
Otherwise, we have $s \geq S_{\min}$, hence,
combining the definition of the strategy $f$ with Equation \eqref{eq:APPdMax_bound}
yields that the bid $x$ of \Max
in the step going from $\rho$ to $\rho'$ satisfies
\[
x = \alpha \frac{s}{N} \cdot B(\rho)
> \frac{\alpha s}{N} \cdot  \frac{N}{\alpha S_{\min}}
\geq 1.
\]
In other words, \Max bids more than $1$, which is the whole budget of \Min.
As a consequence, \Max is guaranteed to win the bidding going from $\rho$ to $\rho'$,
hence $\diff(\rho) < \diff(\rho')$,
which, once again, contradicts the fact that $\rho'$
is the smallest prefix of $\pi$ satisfying $\diff(\rho') < M$.
\hfill$\triangleleft$

We conclude by showing a matching lower bound, namely we show that $\sMP(\G,r) \leq \MP\big(\RT(\G, 1-\frac{C_0}{B_0})\big)$.
Recall that our definition of payoff favors \Min, thus by proving the claim for \Max, we prove a stronger claim. 
Let $\G$ be a strongly-connected mean-payoff all-pay poorman game. Lem.~\ref{lemma:MP_APP_up} below shows that for all $W < 1$, for every initial budgets and every pure strategy of \Min in the $W$-asymmetric bidding game $\norm{G}$,
\Max has a strategy ensuring a payoff arbitrarily close to $\MP\big(\RT(\G, W)\big)$.
Recall that Lem.~\ref{lem:asym} shows that each strategy of \Max in the $W$-asymmetric game $\norm{G}$
can be transformed into a strategy of \Max in $\G$ with the same payoff when the initial budgets $\init{B}$ and $\init{C}$ satisfy $W < \frac{\init{B}}{\init{C}}$. Therefore, Lem.~\ref{lemma:MP_APP_up} implies that
if $\init{B} < \init{C}$ and we fix a pure strategy of \Min,
\Max can ensure a payoff arbitrarily close to $\MP\big(\RT(\G, \frac{\init{B}}{\init{C}})\big)$.
By swapping the roles of \Min and \Max, we obtain as a corollary that $\sMP(\G,r) \leq \MP\big(\RT(\G, 1-\frac{C_0}{B_0})\big)$, as required.

We conclude by showing a matching lower bound, namely we show that $\sMP(\G,r) \leq \MP\big(\RT(\G, 1-\frac{C_0}{B_0})\big)$.
Recall that our definition of payoff favors \Min, thus by proving the claim for \Max, we prove a stronger claim. 
Let $\G$ be a strongly-connected mean-payoff all-pay poorman game. Lem.~\ref{lemma:MP_APP_up} below shows that for all $W < 1$, for every initial budgets and every pure strategy of \Min in the $W$-asymmetric bidding game $\norm{G}$,
\Max has a strategy ensuring a payoff arbitrarily close to $\MP\big(\RT(\G, W)\big)$.
Recall that Lem.~\ref{lem:asym} shows that each strategy of \Max in the $W$-asymmetric game $\norm{G}$
can be transformed into a strategy of \Max in $\G$ with the same payoff when the initial budgets $\init{B}$ and $\init{C}$ satisfy $W < \frac{\init{B}}{\init{C}}$. Therefore, Lem.~\ref{lemma:MP_APP_up} implies that
if $\init{B} < \init{C}$ and we fix a pure strategy of \Min,
\Max can ensure a payoff arbitrarily close to $\MP\big(\RT(\G, \frac{\init{B}}{\init{C}})\big)$.
By swapping the roles of \Min and \Max, we obtain as a corollary that $\sMP(\G,r) \leq \MP\big(\RT(\G, 1-\frac{C_0}{B_0})\big)$, as required.

\begin{lemma}\label{lemma:MP_APP_up}
Let $W \leq 1$. Consider a strongly-connected mean-payoff $W$-asymmetric bidding game $\norm{\G}$.
For every strategy $g$ of \Min, for every $\epsilon > 0$,
\Max has a pure budget-based strategy ensuring the payoff
$\MP\big(\RT(\G, (1-\epsilon)W)\big)$ against $g$.
\end{lemma}

\begin{proof}
Let $\epsilon > 0$, and let $g$ be a pure \Min strategy in $\norm{\G}$. We define \Max's strategy $f$ in $\norm{\G}$ as follows.
We set $p = (1-\epsilon)W$, and find the strengths of the vertices of $\norm{\G}$ using $\RT(\G, p)$. \Max fixes a threshold $t(s, B)$ depending on strengths and his budget. Suppose the token is placed on a vertex with strength $s$ and \Max's budget is $B$ and \Min's bid  according to $g$ is $y$.
If $y > t(s, B)$, \Max judges that the budget required to win the bidding is not worth it,
and stays out by bidding $0$.
If $y \leq t(s, B)$, \Max bids slightly above $y$,
wins the bidding, and moves to $v^+$.
Formally, let $N =(1+\epsilon W)S_{\max}$,
and let $\alpha \in (0,1)$
satisfying $\lambda(\alpha) =\frac{1}{1-\epsilon}$
(see Lemma~\ref{lem:shift-function}).
\Max acts as follows:
\begin{itemize}[noitemsep]
    \item
    \Max computes the bid $y$ of \Min according to the strategy $g$;
    \begin{enumerate}
	\item
	If $y > \frac{s}{N} \alpha \cdot B$, then \Max bids $0$;
    \item 
	If $y \leq \frac{s}{N} \alpha \cdot B$, then \Max bids $y + \epsilon W y$.
	\end{enumerate}
	\item Upon winning, \Max moves the token to $v^+$.
\end{itemize}
We show that the infinite play $\pi$ induced by $f$ and $g$ satisfies
$\payoff(\pi)
\geq
\MP(\RT(\G, p))$.
To start with, suppose that the budget of \Max is $B$, he bids $x$, and \Min bids $y$.
By definition of the strategy $f$, 
$x$ is either equal to $0$ or $(1+\epsilon W) \cdot y$,
hence we can express the updated budget $B'$ of \Max as follows:
\begin{equation*}
B'
=
B-x + Wy
=
\left\{
\begin{array}{ll}
B + W \cdot y & \textup{ if } x \leq y;\\
B - (1-W+\epsilon W) \cdot y & \textup{ if } x > y.
\end{array}
\right.
\end{equation*}

Therefore, \Max gains $\frac{W}{1-W+\epsilon W}$ more for a loss than what he pays for a win.
As a consequence, we can show that the win/loss ratio of \Max is close to $\frac{W}{1-W+\epsilon W}$. 
Formally, let $\nu = W$, and let
\begin{align*}
\mu &= \lambda(\alpha) \cdot (1-W+\epsilon W) = \frac{1}{1-\epsilon}-W.
\end{align*}
For every finite play $\pi$, let $B(\pi)$ be the budget of \Max following $\pi$,
and let 
$
\diff(\pi) = \mu \cdot \investmax(\pi)- \nu \cdot \gainmax(\pi)
$,
which intuitively keeps track of the difference between the number of biddings \Max loses and wins.
We establish an invariant between \Max's budget and $\diff(\pi)$.

\begin{restatable}{claimth}{APPdMinInvariant}
\label{claim:APPdMin_invariant}
    For every finite play $\pi$ coherent with $f$ and $g$,
    $B(\pi) \geq \init{B} \cdot (1+\alpha)^{-\frac{\diff(\pi)}{N}}$.
\end{restatable}
\noindent
We can now prove a lower bound on $\diff(\pi)$:
whenever $\diff(\pi)$ gets too low,
Claim \ref{claim:APPdMin_invariant} implies that the budget of \Max is so high
that he will outbid \Min in the next bid,
causing $\diff(\pi)$ to get back up.
Formally,
\begin{restatable}{claimth}{APPdMinBound}
\label{claim:APPdMin_bound}
    There exists $M \in \mathbb{R}$ such that
    $\diff(\pi) \geq  M$
    for every finite play $\pi$ coherent with $f$ and $g$.
\end{restatable}
\noindent
We can conclude by using Corollary \ref{cor:magic}
as $\frac{\nu}{\mu + \nu} = (1-\epsilon)W = p$:
Let $\pi$ be the infinite play coherent with $f$ and $g$.
If for all $n \in \mathbb{N}$ we denote by $\pi^n$
the prefix of $\pi$ of size $n$, we have
\[
\payoff(\pi)
\geq
\MP(\RT(\G, p)) +
\frac{\mu + \nu}{\mu\nu} \cdot \liminf_{n\rightarrow \infty}\frac{H(\pi^n)}{n}
\geq
\MP(\RT(\G, p)).
\qedhere
\]
\end{proof}

We now prove the claims that appear in the previous proof.

\paragraph{Proof of Claim \ref{claim:APPdMin_invariant}}
We show by induction over the length of a finite play $\pi$ coherent with the strategies $f$ and $g$
that
$B(\pi) \geq \init{B} \cdot (1+\alpha)^{-\frac{\diff(\pi)}{N}}$.
At the start of the game, the equation holds as $\init{B}$ is the initial budget of \Max,
and $H$ is $0$.
We now assume that the equation holds for some finite prefix $\rho$ of $\pi$
that ends in a vertex $v$ of strength $s = \St_p(v)$,
and we consider the next round.
Let $x$ be the bid of \Max
and $y$ be the bid of \Min.
We show that the equation still holds for the updated play $\rho'$
by differentiating the case where \Max loses the bidding and the one where he wins.
\begin{enumerate}
    \item
        Assume that \Max loses the bidding.
        We start with the definition of the budget update, and
        we apply the fact that, according to the strategy $f$,
        $x = 0$ and $y \geq \frac{s}{N} \alpha \cdot B(\rho)$.
        Then, we use the fact that $\nu = W$,
        and we factorise by $B(\rho)$:
        \begin{align}
        B(\rho') = B(\rho) - x + Wy
        &\geq B(\rho) + W \frac{s}{N}\alpha \cdot B(\rho)
        = B(\rho) \cdot \Big(1+  \frac{\nu s}{N} \alpha\Big).\nonumber
        \end{align}
        Note that $\alpha > -1$, and $\frac{\nu s}{N} \in [0,1]$ since $\nu = W<1$ and $N > S_{\max} \geq s$. 
        Therefore, we can apply Bernoulli's inequality.
        We then use the induction hypothesis,
        and we conclude by using the fact that after one bidding happening at a vertex of strength $s$,
        $\diff$ decreases by at most $\nu s$, i.e., $\diff(\rho') \geq \diff(\rho) - \nu s$.
        \begin{equation}
        B(\rho')
        \geq
        B(\rho) \cdot  (1+\alpha)^{\frac{\nu s}{N}}
        \geq
        \init{B} \cdot (1+\alpha)^{-\frac{\diff(\rho) - \nu s}{N}}
        \geq
        \init{B} \cdot (1+\alpha)^{-\frac{\diff(\rho')}{N}}. \nonumber
        \end{equation}
    \item
        Assume that \Max wins the bidding.
        Once again, we start with the definition of the budget update.
        By definition of $f$ we have $y < \frac{s}{N} \alpha \cdot B(\rho)$
        and $x = (1 + \epsilon W) \cdot y$.
        Then, we factorise by $B(\rho)$:
        \begin{align}
        B(\rho') = B(\rho) - x + W y
        &> B(\rho) - (1-W + \epsilon W) \frac{s}{N} \alpha \cdot B(\rho)
        = B(\rho) \cdot \Big(1 - \frac{(1-W + \epsilon W)s}{N} \alpha\Big). \nonumber
        \end{align}
        Note that $-\alpha > -1$, and
        $
        \frac{(1-W + \epsilon W)s}{N} \in [0,1]
        $
        since $N = (1+\epsilon W) S_{\max}$.
        Therefore, we can apply Bernoulli's inequality.
        Moreover, since $\mu = \lambda(\alpha)(1-W + \epsilon W)$,
        by Lemma \ref{lem:shift-function}
        we get that $(1 - \alpha)^{\frac{(1-W + \epsilon W)s}{N}} = (1 + \alpha)^{-\frac{\mu s}{N}}$.
        We then apply the induction hypothesis, and, finally, we use the fact that, since \Max won the bidding by supposition,
        and \Max chooses the successor $v^+$ whenever he wins,
        we have $\diff(\rho') = \diff(\rho) + \mu s$.
        \begin{equation}
        B(\rho')
        \geq
        B(\rho) \cdot (1 - \alpha)^{\frac{(1-W + \epsilon W)s}{N}}
        =
        B(\rho) \cdot (1 + \alpha)^{-\frac{\mu s}{N}}
        \geq
        \init{B} \cdot (1+\alpha)^{-\frac{\diff(\rho) + \mu s}{N}}
        =
        \init{B} \cdot (1+\alpha)^{-\frac{\diff(\rho')}{N}}. \nonumber
        \end{equation}
\end{enumerate}
\hfill$\triangleleft$

\paragraph{Proof of Claim \ref{claim:APPdMin_bound}}
Let $M =  - N \cdot \log_{1+\alpha}(\frac{N}{\alpha S_{\min}}) - \nu S_{\max}$.
We show that every finite play $\pi$ coherent with $f$ satisfies $\diff(\pi) \geq  M$.
Let us assume, towards building a contradiction,
that there exists a finite play $\pi$ coherent with $f$
such that the value of $\diff$ drops under $M$ along $\pi$.
Let $\rho'$ denote the smallest prefix of $\pi$ satisfying $\diff(\rho') < M$.
Note that $\rho'$ cannot be the empty play since the value of $\diff$ is initially $0$.
Let $\rho$ be the play obtained by deleting the last step of $\rho'$.
Since in one step the value of $\diff$ decreases by at most $\nu S_{\max}$,
we get that
$\diff(\rho) \leq \diff(\rho') + \nu S_{\max} < - N \cdot \log_{1+\alpha}(\frac{N}{\alpha S_{\min}})$.
By applying Claim \ref{claim:APPdMin_invariant}, we obtain
\begin{equation}\label{eq:APPdMin_bound}
B(\rho) \geq \init{B} \cdot (1+\alpha)^{-\frac{\diff(\rho)}{N}}
> \frac{N}{\alpha S_{\min}}.
\end{equation}
Let $s$ denote the strength of the vertex reached by $\rho$.
If $s = 0$, then $\diff(\rho) = \diff(\rho')$,
which contradicts the fact that $\rho'$ is the smallest prefix of $\pi$ satisfying $\diff(\rho') < M$.
Otherwise, we have $s \geq S_{\min}$.
Therefore, we can combine Equation \eqref{eq:APPdMin_bound}
with the fact that the bid $y$ of \Min
in the step going from $\rho$ to $\rho'$ satisfies $y \leq 1$
to get the following:
\[
y \leq 1 \leq \frac{s}{N}\alpha \cdot \frac{N}{\alpha S_{min}} < \frac{s}{N}\alpha \cdot B(\rho).
\]
Therefore, by definition of the strategy $f$,
\Max outbids \Min in the last round of $\rho'$,
hence $\diff(\rho) < \diff(\rho')$,
which, once again, contradicts the fact that $\rho'$
is the smallest prefix of $\pi$ satisfying $\diff(\rho') < M$.
\hfill$\triangleleft$

\subsection{Mean-payoff all-pay poorman games under mixed strategies}
\label{subsec:APPpMax}

%We now prove the part of Theorem~\ref{thm:AP-poor} concerning mixed strategies:
Let $\G$ be a strongly-connected mean-payoff all-pay poorman bidding game,
let $\init{B}$ be the initial budget of \Max and $\init{C}$ be the initial budget of \Min
(thus the initial ratio is $r = \frac{\init{B}}{\init{B}+\init{C}}$).
In this section we show that
\begin{equation}
\asMP(\G, r) =
\left\{
\begin{array}{ll}
     \MP\big(\RT(\G, 1- \frac{\init{C}}{2\init{B}})\big) & \textup{ if } \init{B}>\init{C};\\
     \MP\big(\RT(\G, \frac{\init{B}}{2\init{C}})\big)  & \textup{ if } \init{B}\leq\init{C}.
\end{array}
\right.
\end{equation}\label{eq:asmppoorman}

As in the previous section, the proof proceeds by studying $W$-asymmetric bidding games, identifying their almost-sure values depending on $W$, and using Lem.~\ref{lem:asym} that connects asymmetric and all-pay poorman bidding games to obtain the almost-sure values for the latter. 
Specifically, we show that for the $W$-asymmetric bidding game $\norm{G}$ and for every initial ratio $\norm{r}$, 
\[
\asMP(\norm{G}, \norm{r})  =
\left\{
\begin{array}{ll}
     \MP\big(\RT(\G, 1-\frac{1}{2W})\big) & \textup{ if } W > 1;\\
     \MP\big(\RT(\G, \frac{W}{2})\big)  & \textup{ if } W \leq 1.
\end{array}
\right.
\]

\stam{
We show both inequalities by using asymmetric games:
since $\MP\big(\RT(\G, p)\big)$ is continuous in $p$ (see \cite{Cha12,Sol03}),
by Lemma \ref{lem:asym}
it is sufficient to prove that for every $W > 0$,
if we now consider the $W$-asymmetric game $\norm{\G}$
played on the same graph as $\G$, then
for every initial ratio $\norm{r}$,
\[
\asMP(\norm{G}, \norm{r})  \geq
\left\{
\begin{array}{ll}
     \MP\big(\RT(\G, 1-\frac{1}{2W})\big) & \textup{ if } W > 1;\\
     \MP\big(\RT(\G, \frac{W}{2})\big)  & \textup{ if } W \leq 1.
\end{array}
\right.
\]
We state and demonstrate Lemma \ref{lem:APPpMax} to prove the case $W > 1$,
and then Lemma \ref{lem:APPpMin} to prove the case $W \leq 1$.

In order to get equality in eq.~\eqref{eq:asmppoorman} and thus deduce Theorem~\ref{thm:AP-poor}, we again use the symmetry argument for strategies of \Max and \Min, similarly as in Section~\ref{sec:mixedAPRich}. To obtain a strategy for \Min, we consider the game $\G^-$ obtained from $\G$ by negating the weight of each vertex. Then an optimal strategy for \Min in $\G$ with initial budget $r$ corresponds to an optimal strategy of $\Max$ in $\G^-$ with initial budget $1-r$, and these two strategies ensure payoffs that differ only in the sign. Hence, as the ratios $1- \frac{\init{C}}{2\init{B}}$ for $B_0>C_0$ and $\frac{\init{B}}{2\init{C}}$ for $B_0\leq C_0$ satisfy this symmetry, our construction of the strategy for \Max also induces a strategy for \Min. It follows that we have equalities in eq.~\eqref{eq:asmppoorman}, which proves the claim of Theorem~\ref{thm:AP-poor} for mixed strategies.
}%of stam

The following lemma addresses the case where $W>1$. The lemma shows that \Max can guarantee a payoff that is arbitrarily close to $\MP\big(\RT(\G, \frac{2W - 1}{2W})\big)$. To obtain equality, we again use the fact that $\MP\big(\RT(\G, p)\big)$ is continuous in $p$ (see \cite{Cha12,Sol03}) and the advantage of \Min in the definition of payoff.

\begin{lemma}
\label{lem:APPpMax}
Let $W > 1$, and consider a strongly-connected mean-payoff $W$-asymmetric bidding game $\norm{\G}$.
For every initial budget $\init{B} > 0$, for all $\epsilon > 0$, \Max has a mixed budget-based strategy that guarantees an almost-sure payoff of $\MP\big(\RT(\norm{\G}, \frac{2W - 1}{2W + \epsilon})\big)$.
\end{lemma}

% \begin{corollary}
% \label{cor:APPpMax}
% In a strongly-connected mean-payoff all-pay poorman game $\G$ where the initial budget $\init{B}$ of \Max
% is greater than the initial budget $Y$ of \Min,
% for every $\epsilon > 0$, \Max has a probabilistic budget-based strategy that guarantees an expected value of
% $\MP(\RT(\G, 1-(1+\epsilon)\frac{\init{C}}{2\init{B}}))$.
% \end{corollary}

\begin{proof}
Let $\init{B} > 0$ be the initial budget of \Max, and let $\epsilon > 0$.
We set $p = \frac{2W-1}{2W + \epsilon}$,
and find the strengths of the vertices of $\G$ using $\RT(\G, p)$.
The proof is mostly identical to the proof of Lem.~\ref{lem:MP-Rich},
hence we will focus on the differences.
The first change comes in the definition of the mixed strategy $f$ of \Max:
instead of always bidding uniformly at random in an interval,
if \Max's budget is high enough he deterministically bids higher than \Min's budget and forces a bidding win.
Since the budget of \Min is always $1$ in an asymmetric game,
this guarantees that \Max wins the next round,
which will be crucial in the proof of Claim \ref{claim:APPpMax_bound}.
% Such an argument was not needed to prove Lemma \ref{lem:MP-Rich}: in a Richman game:
% the budget of \Max is bounded by the sum of the initial budgets.

Formally, we define the mixed strategy $f$ as follows.
Let $\alpha \in (0,1)$ such that $\lambda(\alpha) = 1+\epsilon$ (see Lemma~\ref{lem:shift-function}).
When the token is on a vertex $v$ with strength $s = \St_{p}(v)$ and \Max's budget is $B$:
\begin{itemize}[noitemsep]
    \item
    If $s > 0$ and $B > \frac{2 W^2 S_{\max}}{\alpha s}$, then
    \Max deterministically bids $\frac{s}{2 W S_{\max}} \alpha B$, which is greater than $W$;
	\item
	Otherwise, \Max bids uniformly at random in the interval $[0,\frac{s}{W S_{\max}}\alpha B]$;
	\item
	Upon winning, \Max moves the token to $v^+$.
\end{itemize}

We now fix a strategy $g$ of \Min. As in previous sections, we assume wlog that 
%Since \Min has the tie-breaking advantage and does not profit from bidding higher than \Max,
%we assume that, according to $g$,
\Min never bids more than the maximal possible bid of \Max according to $f$.
Let $dist(f,g)$ be the probability distribution defined by $f$ and $g$.
To conclude the proof,
we show that
$\mathbb{P}_{\eta \sim dist(f,g)}[\payoff(\eta) \geq \MP\big(\RT(\G, p)\big)] = 1$.

Similar to the Richman setting, we show an invariant
between \Max's budget, his wins, his losses, and his luck.
Observe that contrary to all-pay Richman bidding, here, it is possible for \Max's budget to increase also when he wins a bidding. Indeed, recall that assuming \Max's budget is $B$, he bids $x$, and \Min bids $y$, then the budget update is $B' = B-x+Wy$. For example, \Max's budget increases when $x=1+\epsilon > 1=y$ and $W=2$. % Then, since $W>1$, even when $x > y$, we can have $B' > B$.
This difference leads to a more complicated definition of \Max's ``luck'' that distinguishes between three cases: (1) \Max loses and his budget increases, or he wins and (2a) his budget increases or (2b) his budget decreases. 

To prove the invariant we need several definitions. 
Let $\mu = 1 + \epsilon$, $\nu = 2W -1$,
and $\diff(\pi) = \mu \cdot \investmax(\pi)- \nu \cdot \gainmax(\pi)$.
We define inductively the luck $L$ over a finite play coherent with $f$ and $g$ as follows.
Initially the luck is $0$. Assuming \Max's budget is $B$, his bid is $x$, \Min's bid is $y$, and the luck is $L$, then the updated luck is $L' = L + \Delta L(x,y)$, where 
\[
\Delta L(x,y) = \left\{
\begin{array}{ll}
2 W S_{\max}\frac{W y - x}{\alpha B} - \nu s  & \textup{if $x \leq y$;}\\
2 W S_{\max}\frac{W y - x}{\alpha B} + \mu s & \textup{if $y < x \leq W y$;}\\
2 W \mu S_{\max}\frac{W y - x}{\alpha B}  + \mu s & \textup{if $x > W y$.}
\end{array}
\right.
\]

We get that for every finite play $\pi$ coherent with $f$ and $g$,
if \Max lost many biddings along $\pi$ ($\diff(\pi)$ is low),
then either he gained a lot of budget in exchange ($B(\pi)$ is high),
or he was particularly unlucky ($L(\pi)$ is low):
\begin{restatable}{claimth}{APPpMaxInvariant}
\label{claim:APPpMax_invariant}
    For every finite play $\pi$ coherent with $f$ and $g$,
    $B(\pi) \geq \init{B} \cdot (1+\alpha)^{\frac{L(\pi)-\diff(\pi)}{2W S_{\max}}}$.
\end{restatable}
\noindent
Claim \ref{claim:APPpMax_invariant} implies that if $L(\pi)-\diff(\pi)$ is very high,
then the budget of \Max is high enough to enable the first option of his strategy $f$,
which guarantees him a win in the next round, 
hence causes $L(\pi)-\diff(\pi)$ to decrease.
As a consequence, we get an upper bound for $L(\pi)-\diff(\pi)$:
\begin{restatable}{claimth}{APPpMaxBound}
\label{claim:APPpMax_bound}
    There exists $M \in \mathbb{R}$ such that
    $L(\pi)-\diff(\pi) \leq  M$
    for all finite play $\pi$ coherent with $f$ and $g$.
\end{restatable}
\noindent
Therefore, given an infinite play $\eta$ coherent with $f$ and $g$,
if for every $n \in \mathbb{N}$ we denote by $\eta^n$ the prefix of $\eta$ of size $n$,
we get that
$\liminf_{n\rightarrow\infty}\frac{H(\eta^n)}{n} \geq \liminf_{n\rightarrow\infty}\frac{L(\eta^n)}{n}$.
Moreover, we can prove that
\begin{restatable}{claimth}{APPpMaxLuck}
\label{claim:APPpMax_luck}
$\mathbb{P}_{\pi\sim dist(f,g)}[\liminf_{n\rightarrow\infty}\frac{L(\eta^n)}{n}\geq 0]=1$.
\end{restatable}
\noindent
Since $\frac{\nu}{\mu + \nu} = p$,
this allows us to conclude through the use of Corollary \ref{lem:magic}:
\begin{align*}
\payoff(\eta)
&\geq
\MP(\RT(\G, \frac{\nu}{\mu + \nu})) +
\frac{\mu + \nu}{\mu\nu} \cdot \liminf_{n\rightarrow \infty}\frac{H(\eta^n)}{n}
=
\MP(\RT(\G, p)) +
\frac{\mu + \nu}{\mu\nu} \cdot \liminf_{n\rightarrow \infty}\frac{L(\eta^n)}{n}\\
&\stackrel{a.s.}{\geq}
\MP(\RT(\G, p)).
\qedhere
\end{align*}
\end{proof}

We now present the proofs of the claims.

\paragraph{Proof of Claim \ref{claim:APPpMax_invariant}}
Let $\pi$ be a finite play coherent with $f$ and $g$.
We prove by induction over the length of $\pi$ that 
$B(\pi) \geq \init{B} \cdot (1+\alpha)^{\frac{L(\pi)-\diff(\pi)}{2W S_{\max}}}$.
At the start of the game, the equation holds since both $L$ and $\diff$ are $0$,
and $\init{B}$ is the initial budget of \Max.
For the induction step, suppose that the equation holds for some prefix $\rho$ of $\pi$
that ends in a vertex $v$ of strength $s = \St_{p}(v)$,
and consider the next round.
Let $x$ denote the bid of \Max,
and let $y$ denote the bid of \Min.
We show that the equation still holds for the updated play $\rho'$.
We start by using the definition of the budget update of a $W$-asymmetric game,
we multiply and divide the second part by $\alpha B(\rho)$,
and we factorise $B(\rho)$.
    \begin{align*}
    B(\rho') = B(\rho) - x + Wy
    = B(\rho) + \frac{Wy-x}{\alpha B(\rho)}\alpha B(\rho)
    = B(\rho) \cdot \Big(1+\frac{Wy-x}{\alpha B(\rho)}\alpha\Big).
    \end{align*}
To conclude, we differentiate three cases, the same as in the definition of $\Delta L(x,y)$.
\begin{enumerate}
    \item 
    Suppose that $x \leq y$.
    Then $\frac{\Delta L(x,y) + \nu s}{2 W S_{\max}} = \frac{Wy-x}{\alpha B(\rho)} \in [0,1]$,
    hence, as $\alpha \geq -1$,
    we can use Bernoulli's inequality.
    We then apply the induction hypothesis, and we conclude by using the fact
    that after one bidding happening at a vertex of strength $s$,
    $\diff$ decreases by at most $\nu s$, i.e., $\diff(\rho') \geq \diff(\rho) - \nu s$:
    \[
    B(\rho') \geq
    B(\rho) \cdot (1 + \alpha)^{\frac{\Delta L(x,y) + \nu s}{2W S_{\max}}}
    \geq  \init{B} \cdot (1 + \alpha)^{\frac{L(\rho) - \diff(\rho) + \Delta L(x,y) + \nu s}{2W S_{\max}}}
    \geq  \init{B} \cdot (1 + \alpha)^{\frac{L(\rho') - \diff(\rho')}{2W S_{\max}}}.
    \]
    \item 
    Suppose that $y < x \leq Wy$.
    Then  $\frac{\Delta L(x,y) - \mu s}{2 W S_{\max}} = \frac{Wy-x}{\alpha B(\rho)} \in [0,1]$,
    hence, as $\alpha \geq -1$,
    we can use Bernoulli's inequality.
    We follow by using the induction hypothesis and,
    finally, we use the fact that \Max chooses the successor $v^+$ whenever he wins,
    hence $\diff(\rho') = \diff(\rho) + \mu s$:
    \[
    B(\rho') \geq
    B(\rho) \cdot (1 + \alpha)^{\frac{\Delta L(x,y) - \mu s}{2W S_{\max}}}
    \geq  \init{B} \cdot (1 + \alpha)^{\frac{L(\rho) - \diff(\rho) + \Delta L(x,y) - \mu s}{2W S_{\max}}}
    \geq  \init{B} \cdot (1 + \alpha)^{\frac{L(\rho') - \diff(\rho')}{2W S_{\max}}}.
    \]
    \item
    Suppose that $x > Wy$.
    Then $\frac{\mu s-\Delta L(x,y)}{2 \mu W S_{\max}} = \frac{x-Wy}{\alpha B(\rho)} \in [0,1]$,
    hence, since $-\alpha \geq -1$,
    we can use Bernoulli's inequality.
    We then use Lemma \ref{lem:shift-function} combined with the fact that $\lambda(\alpha) = \mu$
    to get $(1-\alpha) = (1 + \alpha)^{-\mu}$.
    Finally, we apply the induction hypothesis, and we conclude by once again using the fact that, as \Max wins,
    $\diff(\rho') = \diff(\rho) + \mu s$:
    \begin{align*}
    B(\rho') &\geq
    B(\rho) \cdot (1 - \alpha)^{\frac{\mu s - \Delta L(x,y)}{2W \mu S_{\max}}}
    =B(\rho) \cdot (1 + \alpha)^{\frac{\Delta L(x,y)-\mu s}{2W S_{\max}}}
    \geq  \init{B} \cdot (1 + \alpha)^{\frac{L(\rho) - \diff(\rho) + \Delta L(x,y) - \mu s}{2W S_{\max}}}\\
    &\geq  \init{B} \cdot (1 + \alpha)^{\frac{L(\rho') - \diff(\rho')}{2W S_{\max}}}.
    \end{align*}
\end{enumerate}
\hfill$\triangleleft$

\paragraph{Proof of Claim \ref{claim:APPpMax_bound}}
Let $M =  2W S_{\max} \cdot \log_{1+\alpha}(\frac{2 W^2 S_{\max}}{\alpha S_{\min}}) + (2W^2 + 1)\mu S_{\max} + \nu S_{\max}$.
We prove that for every play $\pi$ coherent with $f$ and $g$,
we get $L(\pi) - \diff(\pi) \leq  M$.
Assume, towards building a contradiction,
that for some play $\pi$ coherent with $f$ and $g$
the value of $L-\diff$ exceeds $M$ along $\pi$.
Let $\rho'$ be the smallest prefix of $\pi$ such that $L(\rho') - \diff(\rho') > M$.
Note that $\rho'$ cannot be the empty play as the value of both $L$ and $\diff$ are initially $0$.
Let $\rho$ be the play obtained by deleting the last step of $\rho'$.
We prove that $L(\rho) - \diff(\rho)>L(\rho') - \diff(\rho')$,
which contradicts the fact that $\rho'$
is the smallest prefix of $\pi$ satisfying $L(\rho') - \diff(\rho') > M$.

Note that, in a single step, the value of $L$ cannot increase by more than $(2W^2 + 1)\mu S_{\max}$,
and the value of $\diff$ decreases by at most $\nu S_{\max}$.
Therefore, the value of $L - \diff$ increases by at most $(2W^2 + 1)\mu S_{\max} + \nu S_{\max}$ in a single step,
and we get that
\[
L(\rho) - \diff(\rho) \geq L(\rho') - \diff(\rho') - (2W + 1)\mu S_{\max} - \nu S_{\max} >
2W S_{\max} \cdot \log_{1+\alpha}(\frac{2 W^2 S_{\max}}{\alpha S_{\min}}).
\]
By applying Claim \ref{claim:APPpMax_invariant}, we obtain
\begin{equation*}
B(\rho) \geq \init{B} \cdot (1+\alpha)^{\frac{L(\rho) - \diff(\rho)}{W S_{\max}}}
> \frac{2W^2S_{\max}}{\alpha S_{\min}}.
\end{equation*}
Therefore, by definition of the strategy $f$,
in the round going from $\rho$ to $\rho'$,
\Max deterministically bids $x = \frac{s}{2 W S_{\max}} \alpha B(\rho) > W$.
Since the bid $y$ of \Min is at most $1$ (remember that in an asymmetric game the budget of \Min is always $1$),
this has two consequences:
First, \Max wins the auction, hence $\diff(\rho') = \diff(\rho) + \mu s$.
Second, by definition of the luck update,
$L(\rho') = L(\rho) + 2 \mu S_{\max}\frac{Wy-x}{\alpha B(\rho)} + \mu s < L(\rho) + \mu s$.
Therefore, we get
$L(\rho') - \diff(\rho') <  L(\rho) - \diff(\rho)$.
This contradicts the fact that $\rho'$
is the smallest prefix of $\pi$ satisfying $L(\rho') - \diff(\rho') > M$.
\hfill$\triangleleft$

\paragraph{Proof of Claim \ref{claim:APPpMax_luck}}
We show that, given an infinite play $\eta$ coherent with $f$ and $g$,
if for every $n \in \mathbb{N}$ we denote by $\eta^n$ the prefix of $\eta$ of size $n$,
then $\mathbb{P}_{\pi\sim dist(f,g)}[\liminf_{n\rightarrow\infty}\frac{L(\eta^n)}{n}\geq 0]=1$.
We use the same arguments as in the proof of  Claim~\ref{cl:LB-luck}:
We consider the luck as a stochastic process over the probability space
$(\Omega_\G,\mathcal{F}_\G,dist(f,g))$ defined by $\G$ and the mixed strategies $f$ and $g$.
By Lemma \ref{lemma:concentration}, the claim is proved if we can show that
that the luck is a submartingale (with respect to the canonical filtration)
whose maximum increase/decrease in a single step is bounded by a constant $c$.

Before going further, let us recall the definition of the luck update:
Let us consider a finite play $\pi$ coherent with $f$ and $g$,
and let $\beta = \frac{s}{W S_{\max}}\alpha B(\pi)$ denote the maximal possible
bid of \Max and \Min after $\pi$.
If \Max bids $x \in [0,\beta]$
and \Min bids $y \in [0,\beta]$,
then the luck $L(\pi')$ of the updated run $\pi'$ is defined as
$L(\pi) + \Delta L(x,y)$,
where
\[
\Delta L(x,y) = \left\{
\begin{array}{ll}
2 W S_{\max}\frac{W y - x}{\alpha B(\pi)} - \nu s
=2 s\frac{W y - x}{\beta} - \nu s  & \textup{if $x \leq y$;}\\
2 W S_{\max}\frac{W y - x}{\alpha B(\pi)} + \mu s
= 2 s\frac{W y - x}{\beta} + \mu s& \textup{if $y < x \leq W y$;}\\
2 W \mu S_{\max}\frac{W y - x}{\alpha B(\pi)} + \mu s
=2 \mu s\frac{W y - x}{\beta} + \mu s& \textup{if $x > W y$.}
\end{array}
\right.
\]
We immediately get that $|\Delta L (x,y)|< 2S_{\max}(W^2 + 1)(\mu+\nu)$,
hence the growth of the luck is bounded.
Therefore, to conclude, we just need to prove that the luck is a submartingale.
By Fubini's theorem, it is sufficient to show that for every possible value of the bid $y \in [0,\beta]$ of \Min,
then the expected value $\mathbb{E}_{x}[\Delta L(x,y)]$ of $\Delta L(x,y)$ when $x$ ranges in $[0,\beta]$
according to the strategy $f$ is greater than or equal to $0$.

If $B(\pi) > \frac{2W^2S_{\max}}{\alpha s}$,
computing the expected value of $\Delta L(x,y)$ is easy:
by definition of the strategy $f$,
\Max deterministically bids $x = \frac{s}{2WS_{\max}}\alpha B(\pi)$.
Note that this bid is greater than $1$, hence greater than $y$
since the budget of \Min in an asymmetric game is always $1$.
Therefore, we get
\[
\mathbb{E}_{x}(\Delta L(x,y)) = 2 W \mu S_{\max}\frac{W y - x}{\alpha B(\pi)} + \mu s
\geq -2 W \mu S_{\max}\frac{x}{\alpha B(\pi)} + \mu s
= 0.
\]

If $B(\pi) \leq \frac{2W^2S_{\max}}{\alpha s}$, the proof is more technical.
We need to show that
\[
\mathbb{E}_{x}(\Delta L(x,y)) = \frac{1}{\beta}\cdot \int_0^{\beta} \! \Delta(x,y) \, \mathrm{d}x \geq 0
\]
We differentiate the case where $W y > \beta$ and $W y \leq \beta$.
In both cases, we decompose the integral into parts,
and prove that their sum is greater than or equal to $0$.
Remember that, by definition,
$\mu-1 = \epsilon > 0$
and $\nu = 2W-1$.

\noindent
If $W y > \beta$, then the integral $\int_0^{\beta} \! \Delta(x,y) \, \mathrm{d}x$ is equal to 
\begin{align*}
    &\int_0^{y} \! 2 s\frac{W y - x}{\beta} - \nu s \, \mathrm{d}x +
    \int_y^{\beta} \! 2 s\frac{W y - x}{\beta} + \mu s \, \mathrm{d}x\\
    &=
    \underbrace{\int_0^{\beta} \! 2 s\frac{Wy-x}{\beta} \, \mathrm{d}x -
    \int_0^{y} \! \nu s \, \mathrm{d}x +
    \int_y^{\beta} \! s \, \mathrm{d}x}_{S} + 
     \underbrace{\int_y^{\beta} \! \epsilon s \, \mathrm{d}x}_T.
\end{align*}
\noindent
If $W y \leq \beta$, then the integral $\int_0^{\beta} \! \Delta(x,y) \, \mathrm{d}x$ is equal to 
\begin{align*}
    &\int_0^{y} \! 2 s\frac{W y - x}{\beta} - \nu s \, \mathrm{d}x +
    \int_y^{W y} \! 2 s\frac{W y - x}{\beta} + \mu s \, \mathrm{d}x + 
    \int_{W y}^{\beta} \! 2 s\frac{W y - x}{\beta} \mu + \mu s \, \mathrm{d}x\\
    &=
    \underbrace{\int_0^{\beta} \! 2S_{\max}\frac{Wy-x}{\beta} \, \mathrm{d}x -
    \int_0^{y} \! \nu s \, \mathrm{d}x +
    \int_y^{\beta} \! s \, \mathrm{d}x}_{S} + 
    \underbrace{\int_y^{\beta} \! \epsilon s \, \mathrm{d}x}_T +
    \underbrace{\int_{Wy}^{\beta} \! 2 s\frac{Wy-x}{\beta} \epsilon s \, \mathrm{d}x}_U.
\end{align*}
We compute the values of $S$, $T$, and $U$:
\begin{align*}
    S &= \int_0^{\beta} \! 2  s\frac{Wy-x}{\beta} \, \mathrm{d}x -
    \int_0^{y} \! \nu s \, \mathrm{d}x +
    \int_y^{\beta} \! s \, \mathrm{d}x
    = 
    (2Wy-\beta-\nu  y + \beta - y)s
    =
    0;\\
    T &= \int_y^{\beta} \! \epsilon s \, \mathrm{d}x
     = (\beta - y)\epsilon s;\\
    U &= \int_{W y}^{\beta} \! 2  s\frac{Wy-x}{\beta}\epsilon \, \mathrm{d}x
     = (2Wy-\beta
     - \frac{W^2}{\beta}y^2) \epsilon s.
\end{align*}
Therefore we always have $S + T \geq 0$.
Moreover, if $Wy \leq \beta$, then
\begin{align*}
    S + T +U &= (\beta - y)\epsilon s + (2Wy-\beta - \frac{W^2}{\beta}y^2) \epsilon s
    = \Big(2W-1 - \frac{W^2}{\beta}y\Big)\epsilon s y\\
    &> \Big(W - \frac{W^2}{\beta}y\Big)\epsilon s y
    = \frac{\beta - Wy}{\beta}W\epsilon s y
    \geq 0.
\end{align*}
As a consequence, in both cases, $\mathbb{E}[\Delta L(x,y)] \geq 0$.
\hfill$\triangleleft$

\medskip
The following lemma complements Lem.~\ref{lem:APPpMax} and shows optimal mixed strategies when \Max's ratio is less than $0.5$.

\begin{lemma}
\label{lem:APPpMin}
Let $W \in (0,1]$, and consider a strongly-connected mean-payoff $W$-asymmetric bidding game $\norm{\G}$.
For every initial budget $\init{B} > 0$, for all $\epsilon > 0$, \Max has a mixed budget-based strategy that guarantees an almost-sure payoff of $\MP\big(\RT(\norm{\G}, \frac{W - \epsilon}{2})\big)$.
\end{lemma}

\begin{proof}
Let $B_0>0$ be \Max's initial budget, and let $\epsilon > 0$.
We set $p =  \frac{W - \epsilon}{2}$,
and find the strengths of the vertices of $\G$ using $\RT(\G, p)$.
The proof is nearly identical to the proof of Lem.~\ref{lem:APPpMax}. The dual complications from that lemma stem from the fact that when \Max loses a bidding, he might lose budget instead of gaining budget.
In order to state the two main differences between the proofs,
let us recall the budget update in a $W$-asymmetric game:
Following a play $\pi$,
if \Max bids $x \in \mathbb{R}$ and \Min bids $y \in \mathbb{R}$,
then the updated budget of \Max is
\[
B(\pi') = B(\pi) - x + Wy.
\]
%Since $W \leq 1$, it is now possible for \Max to lose the bidding, and lose some budget on top of that, which could not happen in the previous cases.
For example, when $W = \frac{1}{2}$ and $x=y=1$, thus \Min wins, \Max budget decreases by $\frac{1}{2}$.
This has two consequences:
First, in order to reduce the risk of this unlucky event,
\Max cannot bid uniformly at random in an interval, and needs to bid $0$ with a higher probability.
Second, in the definition of the luck update, we have to consider the case where \Max loses the bid and his budget decreases
instead of the case where he wins the bid and his budget increases.

We now define the mixed strategy $f$ of \Max.
Let $\alpha \in (0,1)$ such that $\lambda(\alpha) = 1 + \epsilon$ (see Lemma~\ref{lem:shift-function}).
When the token is on a vertex $v$ with strength $s = \St_{p}(v)$
and \Max's budget is $B$:
\begin{enumerate}[noitemsep]
    \item
    If $s > 0$ and $B > \frac{2 S_{\max}}{\alpha s}$, then
    \Max deterministically bids $\frac{s}{2 S_{\max}} \alpha B$, which is greater than $1$;
	\item
	Otherwise, \begin{itemize}[nolistsep]
	    \item With probability $1 - W + \epsilon$, \Max bids $0$;
	    \item With probability $W - \epsilon$,
	    \Max picks his bid uniformly at random in $[0,\frac{s}{S_{\max}} \alpha B]$.
	\end{itemize}
	\item
	Upon winning, \Max moves the token to $v^+$.
\end{enumerate}
We now fix a strategy $g$ of \Min.
Since \Min has the tie-breaking advantage and does not profit from bidding higher than \Max,
we assume that, according to $g$,
\Min never bids more in a round than the maximal possible bid of \Max according to $f$.
We consider the probability distribution $dist(f,g)$ defined by $f$ and $g$,
and, to conclude the proof,
we show that
$\mathbb{P}_{\eta \sim dist(f,g)}[\payoff(\eta) \geq \MP\big(\RT(\G, p)\big)] = 1$.

As in the other proofs, we show a relation between the budget of \Max,
his wins, his losses, and his luck.
To this end, we need some formal definitions.
We set $\mu = 2-W + \epsilon$, $\nu = W-\epsilon$,
and
$\diff(\pi) = \mu \cdot \investmax(\pi)- \nu \cdot \gainmax(\pi)$.
We define inductively the luck $L$ over a finite play coherent with $f$ and $g$ as follows.
Initially the luck is $0$.
Then, following a play $\pi$,
if in the next round \Max bids $x \in [0, \frac{s}{WS_{\max}} \alpha B(\pi)]$ according to $f$
and \Min bids $y \in [0, \frac{s}{WS_{\max}} \alpha B(\pi)]$ according to $g$,
the luck $L(\pi')$ of the updated run is defined as $L(\pi) + \Delta L(x,y)$,
where 
\[
\Delta L(x,y) = \left\{
\begin{array}{ll}
2S_{\max}\frac{W y - x}{\alpha B(\pi)} - \nu s  & \textup{if $x \leq W y$;}\\
2(1+\epsilon)S_{\max}\frac{W y - x}{\alpha B(\pi)} - \nu s  & \textup{if $W y < x \leq y$;}\\
2(1+\epsilon)S_{\max}\frac{W y - x}{\alpha B(\pi)} + \mu s & \textup{if $x > y$.}
\end{array}
\right.
\]
The rest of the proof is nearly identical to the one of Lemma \ref{lem:APPpMax}.
We prove a relation between $B$, $\diff$ and $L$:
\begin{restatable}{claimth}{APPpMinInvariant}
\label{claim:APPpMin_invariant}
    For all finite play $\pi$ coherent with $f$ and $g$,
    $B(\pi) \geq \init{B} (1+\alpha)^{\frac{L(\pi)-\diff(\pi)}{2S_{\max}}}$.
\end{restatable}
\noindent
This implies that if $L(\pi)-\diff(\pi)$ gets too high,
then the budget of \Max is high enough to enable the first option of his strategy $f$,
which guarantees him a win in the next round, 
hence causes $L(\pi)-\diff(\pi)$ to decrease.
As a consequence, we get an upper bound for $L(\pi)-\diff(\pi)$:
\begin{restatable}{claimth}{APPpMinBound}
\label{claim:APPpMin_bound}
    There exists $M \in \mathbb{R}$ such that
    $L(\pi)-\diff(\pi) \leq  M$
    for all finite play $\pi$ coherent with $f$ and $g$.
\end{restatable}
\noindent
Therefore, given an infinite play $\eta$ coherent with $f$ and $g$,
if for every $n \in \mathbb{N}$ we denote by $\eta^n$ the prefix of $\eta$ of size $n$,
we get that
$\liminf_{n\rightarrow\infty}\frac{H(\eta^n)}{n} \geq \liminf_{n\rightarrow\infty}\frac{L(\eta^n)}{n}$.
Moreover, we can prove that
\begin{restatable}{claimth}{APPpMinLuck}
\label{claim:APPpMin_luck}
$\mathbb{P}_{\pi\sim dist(f,g)}[\liminf_{n\rightarrow\infty}\frac{L(\eta^n)}{n}\geq 0]=1$.
\end{restatable}
\noindent
Since $\frac{\nu}{\mu + \nu} = p$,
we can conclude by using Corollary \ref{lem:magic}:
\begin{align*}
\payoff(\eta)
&\geq
\MP(\RT(\G, \frac{\nu}{\mu + \nu})) +
\frac{\mu + \nu}{\mu\nu} \cdot \liminf_{n\rightarrow \infty}\frac{H(\eta^n)}{n}
=
\MP(\RT(\G, p)) +
\frac{\mu + \nu}{\mu\nu} \cdot \liminf_{n\rightarrow \infty}\frac{L(\eta^n)}{n}\\
&\stackrel{a.s.}{\geq}
\MP(\RT(\G, p)).
\qedhere
\end{align*}
\end{proof}

We conclude this section by proving the Claims that appear in Lemma \ref{lem:APPpMin}.

\paragraph{Proof of Claim \ref{claim:APPpMin_invariant}}
This proof is nearly identical to the proof of Claim \ref{claim:APPpMax_invariant},
the main changes are caused by the fact that the luck update is defined slightly differently.
We fix a finite play $\pi$ coherent with $f$ and $g$,
and we show by induction over the length of $\pi$ that 
$B(\pi) \geq \init{B} (1+\alpha)^{\frac{L(\pi)-\diff(\pi)}{2S_{\max}}}$.
At the start of the game, 
the initial budget of \Max is $\init{B}$,
and both $L$ and $\diff$ are $0$,
hence the equation holds.
We now suppose that the claim holds for a prefix $\rho$ of $\pi$
that ends in a vertex $v$ of strength $s = \St_{p}(v)$,
and we consider the next round.
Let $x$ be the bid of \Max
and $y$ be the bid of \Min.
We show that the equation still holds for the updated play $\rho'$.
We start by using the definition of the budget update of a $W$-asymmetric game,
we multiply and divide the second part by $\alpha B(\rho)$,
and we factorise $B(\rho)$.
    \begin{align*}
    B(\rho') = B(\rho) - x + Wy
    = B(\rho) + \frac{Wy-x}{\alpha B(\rho)}\alpha B(\rho)
    = B(\rho) \cdot \Big(1+\frac{Wy-x}{\alpha B(\rho)}\alpha\Big).
    \end{align*}
To conclude, we differentiate three cases, the same as in the definition of $\Delta L(x,y)$.
\hfill$\triangleleft$

\begin{enumerate}
    \item 
    Suppose that $x \leq Wy$.
    Then $\frac{\Delta L(x,y) + \nu s}{2 S_{\max}} = \frac{Wy-x}{\alpha B(\rho)} \in [0,1]$,
    hence, as $\alpha \geq -1$,
    we can use Bernoulli's inequality.
    We then apply the induction hypothesis, and we conclude by using the fact
    that after one bidding happening at a vertex of strength $s$,
    $\diff$ decreases by at most $\nu s$, i.e., $\diff(\rho') \geq \diff(\rho) - \nu s$:
    \[
    B(\rho') \geq
    B(\rho) \cdot (1 + \alpha)^{\frac{\Delta L(x,y) + \nu s}{2 S_{\max}}}
    \geq  \init{B} \cdot (1 + \alpha)^{\frac{L(\rho) - \diff(\rho) + \Delta L(x,y) + \nu s}{2S_{\max}}}
    \geq  \init{B} \cdot (1 + \alpha)^{\frac{L(\rho') - \diff(\rho')}{2S_{\max}}}.
    \]
    \item 
    Suppose that $Wy < x \leq y$.
    Then  $-\frac{\Delta L(x,y) + \nu s}{2 (1+\epsilon) S_{\max}} = \frac{x-Wy}{\alpha B(\rho)} \in [0,1]$,
    hence, as $\alpha \geq -1$,
    we can use Bernoulli's inequality.
    We then use Lemma \ref{lem:shift-function} combined with the fact that $\lambda(\alpha) = 1+\epsilon$
    to get $(1-\alpha) = (1 + \alpha)^{-(1+\epsilon)}$.
    We follow by using the induction hypothesis, and the fact that $\diff(\rho') \geq \diff(\rho) - \nu s$:
    \begin{align*}
    B(\rho') &\geq
    B(\rho) \cdot (1 - \alpha)^{-\frac{\Delta L(x,y) + \nu s}{2 (1+\epsilon) S_{\max}}}
    = B(\rho) \cdot (1 + \alpha)^{\frac{\Delta L(x,y) + \nu s}{2 S_{\max}}}
    \geq  \init{B} \cdot (1 + \alpha)^{\frac{L(\rho) - \diff(\rho) + \Delta L(x,y) + \nu s}{2S_{\max}}}\\
    &\geq  \init{B} \cdot (1 + \alpha)^{\frac{L(\rho') - \diff(\rho')}{2S_{\max}}}.
    \end{align*}
    \item
    Suppose that $x > y$.
    Then $\frac{\mu s-\Delta L(x,y)}{2 (1 + \epsilon) S_{\max}} = \frac{x-Wy}{\alpha B(\rho)} \in [0,1]$,
    hence, since $-\alpha \geq -1$,
    we can use Bernoulli's inequality.
    We then apply $(1-\alpha) = (1 + \alpha)^{-(1+\epsilon)}$,
    the induction hypothesis, and we conclude by using the fact that, as \Max wins,
    $\diff(\rho') = \diff(\rho) + \mu s$:
    \begin{align*}
    B(\rho') &\geq
    B(\rho) \cdot (1 - \alpha)^{\frac{\mu s-\Delta L(x,y)}{2 (1 + \epsilon) S_{\max}}}
    =B(\rho) \cdot (1 + \alpha)^{\frac{\Delta L(x,y)-\mu s}{2S_{\max}}}
    \geq  \init{B} \cdot (1 + \alpha)^{\frac{L(\rho) - \diff(\rho) + \Delta L(x,y) - \mu s}{2S_{\max}}}\\
    &\geq  \init{B} \cdot (1 + \alpha)^{\frac{L(\rho') - \diff(\rho')}{2 S_{\max}}}.
    \end{align*}
\end{enumerate}
\hfill$\triangleleft$

\paragraph{Proof of Claim \ref{claim:APPpMin_bound}}
We set $M =  2 S_{\max} \cdot \log_{1+\alpha}(\frac{2 S_{\max}}{\alpha S_{\min}}) + (2W + 1)\mu S_{\max} + \nu S_{\max}$.
The proof is nearly identical to the proof of Claim \ref{claim:APPpMax_bound}.
We prove that $L(\pi) - \diff(\pi) \leq  M$ for every play $\pi$ coherent with $f$ and $g$.
Assume, towards building a contradiction,
that for some play $\pi$ coherent with $f$ and $g$
the value of $L-\diff$ exceeds $M$ along $\pi$.
Let $\rho'$ be the smallest prefix of $\pi$ such that $L(\rho') - \diff(\rho') > M$.
Note that $\rho'$ cannot be the empty play as the value of both $L$ and $\diff$ are initially $0$.
Let $\rho$ be the play obtained by deleting the last step of $\rho'$.
We prove that $L(\rho) - \diff(\rho)>L(\rho') - \diff(\rho')$,
which contradicts the fact that $\rho'$
is the smallest prefix of $\pi$ satisfying $L(\rho') - \diff(\rho') > M$.

Note that, in a single step, the value of $L$ cannot increase by more than $(2W + 1)\mu S_{\max}$,
and the value of $\diff$ decreases by at most $\nu S_{\max}$.
Therefore, the value of $L - \diff$ increases by at most $(2W + 1)\mu S_{\max} + \nu S_{\max}$ in a single step,
and we get that
\[
L(\rho) - \diff(\rho) \geq L(\rho') - \diff(\rho') - (2W + 1)\mu S_{\max} - \nu S_{\max} >
2S_{\max} \cdot \log_{1+\alpha}(\frac{2 S_{\max}}{\alpha S_{\min}}).
\]
By applying Claim \ref{claim:APPpMin_invariant}, we obtain
\begin{equation*}\label{eq:APPpMin_bound}
B(\rho) \geq \init{B} \cdot (1+\alpha)^{\frac{L(\rho) - \diff(\rho)}{2S_{\max}}}
> \frac{2S_{\max}}{\alpha S_{\min}}.
\end{equation*}
Therefore, by definition of the strategy $f$,
in the round going from $\rho$ to $\rho'$,
\Max deterministically bids $x = \frac{s}{2 S_{\max}} \alpha B(\rho)$.
Since the bid $y$ of \Min is at most $1$ (remember that in an asymmetric game the budget of \Min is always $1$),
this has two consequences:
First, \Max wins the auction, hence $\diff(\rho') = \diff(\rho) + \mu s$.
Second, by definition of the luck update,
$L(\rho') < L(\rho) + \mu s$.
Therefore, we get
$L(\rho') - \diff(\rho') <  L(\rho) - \diff(\rho)$.
This contradicts the fact that $\rho'$
is the smallest prefix of $\pi$ satisfying $L(\rho') - \diff(\rho') > M$.
\hfill$\triangleleft$

\paragraph{Proof of Claim \ref{claim:APPpMin_luck}}
Let us begin by recalling the definition of the luck update:
Let us consider a finite play $\pi$ coherent with $f$ and $g$,
and let $\beta = \frac{s}{S_{\max}}\alpha B(\pi)$ denote the maximal possible
bid of \Max and \Min after $\pi$.
If \Max bids $x \in [0,\beta]$
and \Min bids $y \in [0,\beta]$,
then the luck $L(\pi')$ of the updated run $\pi'$ is defined as
$L(\pi) + \Delta L(x,y)$,
where
\[
\Delta L(x,y) = \left\{
\begin{array}{ll}
2S_{\max}\frac{W y - x}{\alpha B(\pi)} - \nu s  
=2s\frac{W y - x}{\beta} - \nu s  & \textup{if $x \leq W y$;}\\
2(1+\epsilon)S_{\max}\frac{W y - x}{\alpha B(\pi)} - \nu s  
=2(1+\epsilon)s\frac{W y - x}{\beta} - \nu s  & \textup{if $W y < x \leq y$;}\\
2(1+\epsilon)S_{\max}\frac{W y - x}{\alpha B(\pi)} + \mu s 
=2(1+\epsilon)s\frac{W y - x}{\beta} + \mu s & \textup{if $x > y$.}
\end{array}
\right.
\]
As we explain in the proofs of Claim~\ref{cl:LB-luck} and Claim \ref{claim:APPpMax_luck},
in order to prove the claim,
it is sufficient to prove that $|\Delta(x,y)|$ is uniformly bounded,
and that for every possible value of the bid $y \in [0,\beta]$ of \Min,
the expected value $\mathbb{E}_{x}[\Delta L(x,y)]$ of $\Delta L(x,y)$ when $x$ ranges in $[0,\beta]$
according to the strategy $f$ is greater than or equal to $0$.
Note that we immediately get the bound $|\Delta(x,y)| \leq (2(1+\epsilon) + \mu + \nu)S_{\max}$.
To conclude, we compute the value of $\mathbb{E}_{x}[\Delta L(x,y)]$ by differentiating two cases.

If $B(\pi) > \frac{2S_{\max}}{\alpha s}$,
computing the expected value of $\Delta L(x,y)$ is easy:
by definition of the strategy $f$,
\Max deterministically bids $x = \frac{s}{2S_{\max}}\alpha B(\pi)$.
Note that this bid is greater than $1$, hence greater than $y$
since the budget of \Min in an asymmetric game is always $1$.
Therefore, we get
\[
\mathbb{E}_{x}[\Delta L(x,y)] = 2 (1+\epsilon) S_{\max}\frac{W y - x}{\alpha B(\pi)} + \mu s
\geq - (1+\epsilon)s + \mu s
= (1-W)s \geq 0.
\]

If $B(\pi) \leq \frac{2S_{\max}}{\alpha s}$,
then with probability $1-\nu$ $\Max$ bids $0$,
and with probability $\nu$ \Max bids uniformly at random in the interval $[0,\beta]$,
hence
\[
\mathbb{E}_{x}[\Delta L(x,y)]
= (1-\nu) \cdot \Delta(0,y) + \frac{\nu}{\beta} \cdot \int_0^{\beta} \! \Delta(x,y) \, \mathrm{d}x.
\]
In order to prove that this is greater than or equal to $0$,
we decompose the integral $\int_0^{\beta} \! \Delta(x,y) \, \mathrm{d}x$ as a sum $S + T + U$,
and then we conclude by showing that $(1-\nu) \cdot \Delta(0,y) + \frac{\nu}{\beta}S = 0$
and $T+U \geq 0$.

The integral $\int_0^{\beta} \! \Delta(x,y) \, \mathrm{d}x$ is equal to 
\begin{align*}
    &\int_0^{W y} \! 2s\frac{W y - x}{\beta} - \nu s \, \mathrm{d}x +
    \int_{W y}^{y} \! 2s(1+\epsilon)\frac{W y - x}{\beta} - \nu s \, \mathrm{d}x + 
    \int_{y}^{\beta} \! 2s(1+\epsilon)\frac{W y - x}{\beta}  + \mu s \, \mathrm{d}x\\
    &=
    \underbrace{\int_0^{\beta} \! 2s\frac{W y - x}{\beta} \, \mathrm{d}x -
    \int_0^{y} \! W s \, \mathrm{d}x +
    \int_y^{\beta} \! (\mu-\epsilon) s \, \mathrm{d}x}_{S} + 
    \underbrace{\int_{W y}^{\beta} \! \frac{W y - x}{\beta}2 \epsilon s \, \mathrm{d}x}_T +
    \underbrace{\int_{W y}^{\beta} \! \epsilon s \, \mathrm{d}x}_U.
\end{align*}
We compute the values of $S$, $T$,
and $U$ (remember that $\mu = 2 - W+ \epsilon$):
\begin{align*}
    S &= \int_0^{\beta} \! \frac{W y - x}{\beta}2s \, \mathrm{d}x -
    \int_0^{y} \! W s \, \mathrm{d}x +
    \int_y^{\beta} \! (\mu-\epsilon) s \, \mathrm{d}x
    = (2 W y -\beta - W y + (\beta-y)(2-W))s\\
    &= (1-W)(\beta - 2y)s;\\
    T &= \int_{W y}^{\beta} \! \frac{W y - x}{\beta}2\epsilon s \, \mathrm{d}x
     = -\frac{(\beta-W y)^2}{\beta}\epsilon s;\\
    U &= \int_{0}^{\beta} \! \epsilon s \, \mathrm{d}x
     = \beta \epsilon s.
\end{align*}
Therefore, we get
\begin{align*}
&(1-\nu) \cdot \Delta(0,y) + \frac{\nu}{\beta}S = 
(1-\nu)(\frac{2 y}{\beta} - 1)\nu s+
\frac{\nu}{\beta}(1-\nu)(\beta - 2y)s
=0;\\
&T + U=\frac{2 \beta - Wy}{\beta}\epsilon W s y
\geq 0.\qedhere
\end{align*}
\hfill$\triangleleft$

% \begin{corollary}
% \label{cor:APPpMin}
% In a strongly-connected mean-payoff all-pay poorman game $\G$ where the initial budget $X$ of \Max is smaller than or equal to the initial budget $Y$ of \Min,
% for every $\epsilon > 0$, \Max has a probabilistic budget-based strategy that guarantees an expected value of
% $\MP(\RT(\G,\frac{\init{B}}{2\init{C}}-\epsilon))$.
% \end{corollary}

\section{Parity All-Pay Bidding Games}\label{sec:parityproof}
\label{sec:qual-full}
In this section, we prove Thm.~\ref{thm:parity} based on the solution to mean-payoff bidding games. Let $\P$ be a strongly-connected parity game in which the maximal parity index is $d \in \Nat$. We construct a mean-payoff game $\G$ by setting the weight of a vertex $v$ to be $1$ if the parity of $v$ is $d$, and otherwise $w(v) = 0$. The key property of this weight function is that any path $\eta$ with $\payoff(\eta)>0$ must visit a vertex with index $d$ infinitely many times, and thus satisfies the parity objective. 

\begin{lemma}\label{lemma:rtparity}
Let $\G$ be a strongly-connected mean-payoff game with non-negative weights and at least one strictly positive weight. Then, for every $p \in (0,1)$, we have $\MP\big(\RT(\G, p)\big)>0$.
\end{lemma}
\begin{proof}
Let $v_0\in V$ be a vertex whose weight is positive. Since $\G$ is strongly-connected, every vertex $v\in V$ admits a shortest path to $v_0$. Fix one such path for each $v$ and let $v'$ be the successor of $v$ along this path (for $v=v_0$, let $v'$ be any of its neighbors). Define the strategy $\sigma$ for \Max via $\sigma(v)=v'$, so \Max moves the token along the edge $\zug{v,v'}$ upon winning the coin toss. We show that this strategy guarantees a positive mean-payoff with probability $1$.

\medskip Let $|V|=n$. The length of a shortest path from any vertex in $\G$ to $v_0$ is at most $n-1$. Thus if \Max follows the strategy $\sigma$ and wins $n-1$ consecutive coin tosses, the token will reach $v_0$ at least once in those $n-1$ turns. As the coin tosses are pairwise independent, the probability of  \Max winning $n-1$ times in a row is $p^{n-1}$. We will use this observation to show that $\sigma$ ensures positive mean-payoff.

\medskip For an infinite game play $\pi$, let $\pi^m$ be its finite prefix of length $m$. Moreover, let $v_i(\pi)$ denote the $i$-th vertex along $\pi$. If we write $m=k\cdot(n-1)+r$ with $0\leq r<n-1$, the expected energy of $\pi^m$ under $\sigma$ and any fixed strategy of the opponent is
\begin{equation*}
\begin{split}
&\mathbb{E}[\pi^m] \geq \mathbb{E}^{\sigma}[\pi^{k\cdot (n-1)}] \geq w(v_0)\cdot \mathbb{E}[\#\{1\leq j\leq k\cdot(n-1)\mid v_j(\pi)=v_0\}]\\
&= w(v_0)\cdot\sum_{i=0}^{k-1}\mathbb{E}[\#\{1\leq j\leq n-1\mid v_{i\cdot (n-1)+j}(\pi)=v_0\}] \\
&\geq w(v_0)\cdot\sum_{i=0}^{k-1}\mathbb{E}[\mathbb{I}(v_{i\cdot (n-1)+j}(\pi)=v_0\text{ for some }1\leq j\leq n-1)] \\
&= w(v_0)\cdot\sum_{i=0}^{k-1}\mathbb{P}[v_{i\cdot (n-1)+j}(\pi)=v_0\text{ for some }1\leq j\leq n-1] \\
&\geq w(v_0)\cdot\sum_{i=0}^{k-1} b^{n-1} = w(v_0)\cdot k\cdot p^{n-1},
\end{split}
\end{equation*}
where the last inequality follows from the above observation. Since $k = (m-r)/(n-1)$ and $r<n-1$, we have $k>m/(n-1)-1$. Thus,
\begin{equation*}
\liminf_{m\rightarrow\infty}\frac{\mathbb{E}[\pi^m]}{m} \geq w(v_0)\cdot\liminf_{m\rightarrow\infty}\frac{(m/(n-1)-1)\cdot p^{n-1}}{m} = \frac{w(v_0)\cdot p^{n-1}}{n-1} >0,
\end{equation*}
and $\sigma$ ensures positive mean-payoff as claimed.
\end{proof}

The proof of Thm.~\ref{thm:parity} follows from combining Lem.~\ref{lemma:rtparity} with Thms.~\ref{thm:AP-Rich} and~\ref{thm:AP-poor}.

\section{Conclusions}
\label{sec:conc}
We study, for the first time, infinite-duration all-pay bidding games. In terms of applications, all-pay bidding, especially combined with poorman bidding, is often more favorable than first-price bidding since it accurately models settings in which bounded resources with little or no inherent value need to be invested. Technically, however, all-pay bidding is much more challenging than first-price bidding since mixed strategies need to be considered. Prior to this work, reachability all-pay bidding games were only recently studied and more questions were left open than closed. This work is thus the first to find rich mathematical structure for all-pay bidding in the form of equivalences with random-turn games. We hope that the techniques we develop here will assist in shedding light also on reachability all-pay bidding games. 

This work constitutes another step in the line of work that studies the intriguing equivalence between bidding games and random-turn games. Starting from reachability first-price bidding games \cite{LLPSU99} and continuing with mean-payoff first-price bidding games~\cite{AHC19,AHI18,AHZ19}. We find the results of all-pay bidding games particularly surprising and we believe they encourage further investigation to understand the elegant connection between bidding games and random-turn games. See \cite{AH20} for a list of concrete open questions on bidding games.

\small
\bibliographystyle{plain}
\bibliography{../../ga}

\begin{thebibliography}{10}

\bibitem{AAH19}
M.~Aghajohari, G.~Avni, and T.~A. Henzinger.
\newblock Determinacy in discrete-bidding infinite-duration games.
\newblock In {\em Proc. 30th CONCUR}, volume 140 of {\em LIPIcs}, pages
  20:1--20:17, 2019.

\bibitem{AHK02}
R.~Alur, T.~A. Henzinger, and O.~Kupferman.
\newblock Alternating-time temporal logic.
\newblock {\em J. {ACM}}, 49(5):672--713, 2002.

\bibitem{AD00}
R.B. Ash and C.~Dol{\'e}ans-Dade.
\newblock {\em Probability and Measure Theory}.
\newblock Harcourt/Academic Press, 2000.

\bibitem{AH20}
G.~Avni and T.~A. Henzinger.
\newblock A survey of bidding games on graphs.
\newblock In {\em Proc. 31st CONCUR}, volume 171 of {\em LIPIcs}, pages
  2:1--2:21. Schloss Dagstuhl - Leibniz-Zentrum f{\"{u}}r Informatik, 2020.

\bibitem{AHC19}
G.~Avni, T.~A. Henzinger, and V.~Chonev.
\newblock Infinite-duration bidding games.
\newblock {\em J. ACM}, 66(4):31:1--31:29, 2019.

\bibitem{AHI18}
G.~Avni, T.~A. Henzinger, and R.~Ibsen-Jensen.
\newblock Infinite-duration poorman-bidding games.
\newblock In {\em Proc. 14th WINE}, volume 11316 of {\em LNCS}, pages 21--36.
  Springer, 2018.

\bibitem{AHZ19}
G.~Avni, T.~A. Henzinger, and {\DJ}.~\v{Z}ikeli\'c.
\newblock Bidding mechanisms in graph games.
\newblock In {\em In Proc. 44th MFCS}, volume 138 of {\em LIPIcs}, pages
  11:1--11:13, 2019.

\bibitem{AIT20}
G.~Avni, R.~Ibsen{-}Jensen, and J.~Tkadlec.
\newblock All-pay bidding games on graphs.
\newblock In {\em Proc. 34th AAAI}, pages 1798--1805. {AAAI} Press, 2020.

\bibitem{azuma1967}
Kazuoki Azuma.
\newblock Weighted sums of certain dependent random variables.
\newblock {\em Tohoku Math. J. (2)}, 19(3):357--367, 1967.

\bibitem{BP09}
J.~Bhatt and S.~Payne.
\newblock Bidding chess.
\newblock {\em Math. Intelligencer}, 31:37--39, 2009.

\bibitem{Bor21}
E.~Borel.
\newblock La th\'eorie du jeu les \'equations int\'egrales \'a noyau
  sym\'etrique.
\newblock {\em Comptes Rendus de l'Acad\'emie}, 173(1304--1308):58, 1921.

\bibitem{CJ+17}
C.~Calude, S.~Jain, B.~Khoussainov, W.~Li, and F.~Stephan.
\newblock Deciding parity games in quasipolynomial time.
\newblock In {\em Proc. 49th STOC}, 2017.

\bibitem{Cha12}
K.~Chatterjee.
\newblock Robustness of structurally equivalent concurrent parity games.
\newblock In {\em Proc. 15th FoSSaCS}, pages 270--285, 2012.

\bibitem{CGV18}
K.~Chatterjee, A.~K. Goharshady, and Y.~Velner.
\newblock Quantitative analysis of smart contracts.
\newblock In {\em Proc. 27th ESOP}, pages 739--767, 2018.

\bibitem{CRN12}
Krishnendu Chatterjee, Johannes~G. Reiter, and Martin~A. Nowak.
\newblock Evolutionary dynamics of biological auctions.
\newblock {\em Theoretical Population Biology}, 81(1):69 -- 80, 2012.

\bibitem{Con92}
A.~Condon.
\newblock The complexity of stochastic games.
\newblock {\em Inf. Comput.}, 96(2):203--224, 1992.

\bibitem{DP10}
M.~Develin and S.~Payne.
\newblock Discrete bidding games.
\newblock {\em The Electronic Journal of Combinatorics}, 17(1):R85, 2010.

\bibitem{EJS93}
A.~E. Emerson, C.~S. Jutla, and P.~A. Sistla.
\newblock On model-checking for fragments of {\(\mathrm{\mu}\)}-calculus.
\newblock In {\em Proc. 5th CAV}, pages 385--396, 1993.

\bibitem{How60}
A.~R. Howard.
\newblock {\em Dynamic Programming and Markov Processes}.
\newblock MIT Press, 1960.

\bibitem{LW18}
U.~Larsson and J.~W\"astlund.
\newblock Endgames in bidding chess.
\newblock {\em Games of No Chance 5}, 70, 2018.

\bibitem{LLPSU99}
A.~J. Lazarus, D.~E. Loeb, J.~G. Propp, W.~R. Stromquist, and D.~H. Ullman.
\newblock Combinatorial games under auction play.
\newblock {\em Games and Economic Behavior}, 27(2):229--264, 1999.

\bibitem{LLPU96}
A.~J. Lazarus, D.~E. Loeb, J.~G. Propp, and D.~Ullman.
\newblock Richman games.
\newblock {\em Games of No Chance}, 29:439--449, 1996.

\bibitem{MKT18}
R.~Meir, G.~Kalai, and M.~Tennenholtz.
\newblock Bidding games and efficient allocations.
\newblock {\em Games and Economic Behavior}, 112:166--193, 2018.

\bibitem{MWX15}
M.~Menz, J.~Wang, and J.~Xie.
\newblock Discrete all-pay bidding games.
\newblock {\em CoRR}, abs/1504.02799, 2015.

\bibitem{PSSW09}
Y.~Peres, O.~Schramm, S.~Sheffield, and D.~B. Wilson.
\newblock Tug-of-war and the infinity laplacian.
\newblock {\em J. Amer. Math. Soc.}, 22:167--210, 2009.

\bibitem{PR89}
A.~Pnueli and R.~Rosner.
\newblock On the synthesis of a reactive module.
\newblock In {\em Proc. 16th POPL}, pages 179--190, 1989.

\bibitem{Put05}
M.~L. Puterman.
\newblock {\em Markov Decision Processes: Discrete Stochastic Dynamic
  Programming}.
\newblock John Wiley \& Sons, Inc., New York, NY, USA, 2005.

\bibitem{Rab69}
M.O. Rabin.
\newblock Decidability of second order theories and automata on infinite trees.
\newblock {\em Transaction of the AMS}, 141:1--35, 1969.

\bibitem{Sol03}
E.~Solan.
\newblock Continuity of the value of competitive markov decision processes.
\newblock {\em Journal of Theoretical Probability}, 16:831--845, 2003.

\bibitem{Tul80}
G.~Tullock.
\newblock {\em Toward a Theory of the Rent Seeking Society}, chapter Efficient
  rent seeking, pages 97--112.
\newblock College Station: Texas A\&M Press, 1980.

\bibitem{Williams:book}
D.~Williams.
\newblock {\em {Probability with Martingales}}.
\newblock {Cambridge Mathematical Textbooks}. Cambridge University Press,
  Cambridge, UK, 1991.

\end{thebibliography}

\end{document}